\documentclass[%
 aip,
 amsmath,amssymb,
 reprint,%
]{revtex4-1}

\usepackage{graphicx}
\usepackage{dcolumn}
\usepackage{bm}
\usepackage[utf8]{inputenc}
\usepackage[T1]{fontenc}
\usepackage{mathptmx}
\usepackage{etoolbox}
\usepackage{enumerate}
\usepackage{array,amsthm,color}
\usepackage{amsmath}

\newcommand{\ii}{\mathrm{i}}
\newcommand{\ee}{\mathrm{e}}

\newcommand{\T}{\mathrm{T}}
\newcommand{\cO}{\mathcal{O}}

\newtheorem{theorem}{Theorem}

\newtheorem{prop}{Proposition}

\newtheorem{remark}{Remark}

\makeatletter
\def\@email#1#2{%
 \endgroup
 \patchcmd{\titleblock@produce}
  {\frontmatter@RRAPformat}
  {\frontmatter@RRAPformat{\produce@RRAP{*#1\href{mailto:#2}{#2}}}\frontmatter@RRAPformat}
  {}{}
}%
\makeatother
\begin{document}

\preprint{AIP/123-QED}

\title{Rogue wave pattern of multi-component derivative nonlinear Schr{\"o}dinger equations}
\author{Huian Lin}
\author{Liming Ling}%
 \email{linglm@scut.edu.cn}
\affiliation{ 
School of Mathematics, South China University of Technology, Guangzhou, China, 510641
}%


\date{\today}

\begin{abstract}
This paper delves into the study of multi-component derivative nonlinear Schr{\"o}dinger ($n$-DNLS) equations featuring nonzero boundary conditions. Employing the Darboux transformation (DT) method, we derive higher-order vector rogue wave solutions for the $n$-DNLS equations. Specifically, we focus on the distinctive scenario where the $ (n+1) $-order characteristic polynomial possesses an explicit $(n+1)$-multiple root. Additionally, we provide an in-depth analysis of the asymptotic behavior and pattern classification inherent to the higher-order vector rogue wave solution of the $n$-DNLS equation, particularly when one of the internal parameters attains a significant magnitude. These patterns are related to the root structures in the generalized Wronskian-Hermite polynomial hierarchies.
\end{abstract}

\maketitle

%

\section{Introduction}

Integrable systems exhibit abundant mathematical structures and enjoy a wide array of applications in physics. Among these, the class of integrable systems linked to the nonlinear Schr\"odinger (NLS) equation holds a unique position. These systems hold particular significance, permeating various fields including optical fibers \cite{Kodama1987, 2000Nonlinear, 2007Self}, plasma physics \cite{1976Modified, 2002DNLS}, and Bose-Einstein condensates \cite{2009Matter}, and so on. In optics, the classical NLS equation finds purpose in describing the transmission of picosecond short pulses in single-mode fibers \cite{2000Nonlinear}. Furthermore, the derivative nonlinear Schr\"odinger (DNLS) equation plays a pivotal role as a fundamental mathematical model for describing the propagation of circular-polarized nonlinear Alfv\'en waves in plasmas \cite{1976Modified,2002DNLS}, as well as for elucidating the dynamics of nonlinear pulses in optical fibers \cite{Einar1976On, agrawal2000}. Owing to the profound physical significance associated with these integrable equations, researchers devote substantial effort to their comprehensive exploration. Employing various methodologies, they investigate the integrability properties inherent within these equations. Consequently, they derive analytical solutions, encompassing soliton solutions, positon solutions, breather solutions, and rogue wave solutions.  

In intricate physical systems, adopting multi-component integrable systems as research models is a common strategy. For instance, within the nonlinear optics, as the two-component extensions of the NLS equation, the coupled NLS equation offers a potent means of describing the concurrent emergence of two or more wave packets possessing distinct carrier frequencies under specific complex physical conditions \cite{2000Nonlinear}. Moreover, Hoefer \cite{2012Dark} observed multi-solitons corresponding to multi-component integrable systems in experiments involving Bose-Einstein condensates. They demonstrated that the fast counterflow-induced modulational instability (MI) can lead to the generation of vector solitons. {Thus,} many researchers have investigated multi-component integrable systems\cite{Hisakado1994, Hisakado1995,Tsuchida1999, 2012Guo, Ling2010, 2018The, GUO2019,JIA2021,2022Riemann, YAN2020,rao2022, ma2023}, which is essential and meaningful.

In this article, we focus on the $ n $-DNLS equations \cite{Tsuchida1999}
\begin{equation}\label{eq:nDNLS}
	\begin{aligned}
		&\ii \mathbf{q}_{t}+\mathbf{q}_{xx}-\ii \left( \| \mathbf{q}\|_{2}^{2}\mathbf{q}\right) _{x} =\mathbf{0}, \\ &\mathbf{q}=\left( q_{1}(x,t),q_{2}(x,t),\cdots,q_{n}(x,t)\right) ^{T}, 
	\end{aligned}
\end{equation}
where $ \mathbf{q} $ is the complex vector function, the subscripts $ x $ and $ t $ denote the partial derivatives concerning the variables, and $ \|\cdot\|_{2}^{2} $ is the standard Euclidean norm. It is a straightforward vector extension of the scalar DNLS equation, exhibiting complete integrability. Specifically, when $ n=2 $, Eq. \eqref{eq:nDNLS} is the coupled DNLS equation. Extensive scholarly efforts have been dedicated to the investigation of this equation, yielding notable outcomes including vector soliton solutions, breather solutions, rogue wave solutions, and its long-time asymptotics \cite{2012Guo, Ling2010, GUO2019,JIA2021,2022Riemann, YAN2020}. Noteworthy contributions to the exploration of the integrability characteristics of the $ n $-DNLS equations were made by Hisakado and Wadati in 1994 \cite{Hisakado1994, Hisakado1995}, wherein they additionally established a gauge transformation connecting the $ n $-DNLS equations with the $ n $-component NLS equations. Expanding their insights further, Tsuchida and Wadati \cite{Tsuchida1999} generalized the $ n $-DNLS equations to a class of novel multi-component DNLS equations through gauge transformation in 1999. Furthermore, Liu and Geng \cite{2018The} delved into the initial-boundary problems of the $ n $-DNLS equations on the half-line. Ma and Zhu \cite{ma2023} explored the Riemann–Hilbert problem and $ N $-soliton solutions inherent to the $ n $-DNLS equation. 

To the best of our current knowledge, the {DT} and rogue wave solutions about the $ n $-DNLS equation \eqref{eq:nDNLS} remain unexplored. The primary objectives of this paper are as follows:
\begin{itemize}
	\item To establish a systematic DT for the $ n $-DNLS equations, thereby facilitating the derivation of determinant expressions for its higher-order solutions.
	
	\item To leverage the presented DT to construct rogue wave solutions with the nonzero backgrounds. This endeavor will be undertaken specifically when the characteristic polynomial exhibits maximal multiple roots. 
	
	\item Taking a comprehensive exploration of the asymptotic behavior and structural patterns for the higher-order rogue wave solutions with a large parameter.
	
	\item Recognizing the profound relationship between higher-order rogue wave patterns and the root configurations of generalized Wronskian-Hermite polynomials, this study delves extensively into the root distributions of the generalized Wronskian-Hermite polynomial hierarchies.
\end{itemize}
In the course of deriving the DT and calculating rogue wave solutions for the $ n $-DNLS equations, {inspiration draw from the Darboux transformation of the scalar DNLS equation} \cite{guo2013}, as well as in generating the vector rogue wave solutions of the $ n $-component NLS equation \cite{zhang2021}. Additionally, we refer to the investigations of generalized Wronskian-Hermite polynomial hierarchies by researchers such as Yang, Zhang et al. \cite{felder2012, bonneux2020, yang2021, yang2023, zhang2022}. It is noteworthy that the intricate matter of root structures about generalized Wronskian-Hermite polynomials with arbitrary jump parameters has hitherto remained uncharted. Consequently, we present a meticulous analysis of root structures for generalized Wronskian-Hermite polynomials, incorporating arbitrary jump integers.

The structure of this paper is as follows: In Section \ref{sec2}, we present the DT of $ n $-DNLS equations \eqref{eq:nDNLS} and discuss the root configurations of generalized Wronskian-Hermite polynomials with arbitrary jump integers. In Section \ref{sec3}, we study the scenarios where characteristic polynomial has the roots of maximal multiplicity. Further, we provide the formula for higher-order rogue wave solutions under such conditions. In Section \ref{sec4}, an investigation into the asymptotic behavior and structural patterns of higher-order rogue wave solutions is conducted, along with the specific instances for illustration. Finally, we give conclusions and discussions in Section \ref{sec5}.

\section{Preliminaries}\label{sec2}

In this section, we construct the DT of the $ n $-DNLS equations to derive the determinant formula of the solutions. Additionally, to analyze the patterns of the higher-order vector rogue wave in the following text, we will also study the generalized Wronskian-Hermite polynomial hierarchies in detail. 

\subsection{DT of the $ n $-DNLS equations}

The $ n $-DNLS equations \eqref{eq:nDNLS} admit the following Lax pair \cite{lenells2010, guo2013}:
\begin{equation} \label{eq:sp}
	\left\{
	\begin{aligned}
		&\Phi_{x}=\mathbf{U}(\lambda;x,t)\Phi,  \\
		&\Phi_{t}=\mathbf{V}(\lambda;x,t)\Phi, 
	\end{aligned}
	\right.
\end{equation}
where
\begin{equation}\label{eq:uv}
	\begin{aligned}
		&\mathbf{U}=\ii \lambda^{-2}\sigma_{3} +\lambda^{-1}\mathbf{Q},\\
		&\mathbf{V}= 2\ii \lambda^{-4} \sigma_{3} +2\lambda^{-3}\mathbf{Q}+ \ii\lambda^{-2}\sigma_{3}\mathbf{Q}^{2}+\lambda^{-1} \left( \mathbf{Q}^{3}+ \ii \mathbf{Q}_{x}\sigma_{3} \right),\\
		&\sigma_{3}={\rm diag}(1,-\mathbb{I}_{n}), \qquad
		\mathbf{Q}=\begin{pmatrix}
			0 & \mathbf{q}^{\dagger} \\
			\mathbf{q} & \mathbf{0}_{n\times n} 
		\end{pmatrix}, 
	\end{aligned}
\end{equation}
$ \Phi=\Phi(\lambda;x,t) $ is the complex matrix spectral function, $ \mathbb{I}_{n} $ is a $ n\times n $ identity matrix, $ \dagger $ denotes the conjugate transpose, and $ \lambda \in \mathbb{C} $ is a spectral parameter. In evidence, the coefficient matrices $ \mathbf{U}$ and $ \mathbf{V} $ possess the following symmetric relationships:
\begin{equation}\label{eq:sy}
	\begin{aligned}
		&\mathbf{U}(\lambda;x,t)=-\sigma_{3} (\mathbf{U}(\lambda^{\ast};x,t))^{\dagger} \sigma_{3} = \sigma_{3} \mathbf{U}(-\lambda;x,t) \sigma_{3},  \\ &\mathbf{V}{(\lambda;x,t)}=-\sigma_{3}(\mathbf{V}(\lambda^{\ast};x,t))^{\dagger}\sigma_{3} = \sigma_{3} \mathbf{V}(-\lambda;x,t) \sigma_{3}.
	\end{aligned}
\end{equation} 
Thus, if the matrix function $ \Phi(\lambda; x,t) $ solves the Lax pair \eqref{eq:sp}, the matrix functions $ \sigma_{3}(\Phi(\lambda^{*}; x,t))^{\dagger}\sigma_{3} $ and $ \sigma_{3}\Phi(\lambda; x,t)\sigma_{3} $ satisfy the adjoint Lax pair, respectively, as follows: 
\begin{equation} \label{eq:ad1}
	\left\{
	\begin{aligned}
		&-\left( \sigma_{3}(\Phi(\lambda^{*}; x,t))^{\dagger}\sigma_{3}\right) _{x}=\left( \sigma_{3} (\Phi(\lambda^{*};x,t))^{\dagger} \sigma_{3}\right) \mathbf{U}(\lambda;x,t),  \\ 
		&-\left( \sigma_{3}(\Phi(\lambda^{*}; x,t))^{\dagger}\sigma_{3}\right) _{t} =\left( \sigma_{3} (\Phi(\lambda^{*};x,t))^{\dagger} \sigma_{3}\right) \mathbf{V}(\lambda;x,t), 
	\end{aligned}
	\right.
\end{equation}
and
\begin{equation} \label{eq:ad2}
	\left\{
	\begin{aligned}
		&\left( \sigma_{3}\Phi(\lambda; x,t) \sigma_{3}\right)_{x}= \mathbf{U}(-\lambda;x,t)\left( \sigma_{3}\Phi(\lambda; x,t)\sigma_{3}\right),  \\ 
		&\left( \sigma_{3} \Phi(\lambda; x,t)\sigma_{3} \right)_{t}= \mathbf{V}(-\lambda;x,t)\left( \sigma_{3}\Phi(\lambda; x,t)\sigma_{3}\right). 
	\end{aligned}
	\right.
\end{equation}
On the other hand, when $ \Phi{(\lambda; 0,0)}= \mathbb{I}_{n+1}$, we derive
\begin{equation}\label{eq:syp}
	\begin{aligned}
	&\sigma_{3}(\Phi(\lambda^{*}; x,t))^{\dagger}\sigma_{3}= \left( \Phi(\lambda; x,t) \right)^{-1}, \\
	&\sigma_{3}\Phi(\lambda; x,t) \sigma_{3} = \Phi(-\lambda; x,t),
	\end{aligned}
\end{equation}
by the uniqueness and existence of differential equations since the matrix function $ \left( \Phi(\lambda; x,t) \right)^{-1} $ and $ \Phi(-\lambda; x,t) $ also satisfy the above systems \eqref{eq:ad1} and \eqref{eq:ad2}, separately.

Now, we establish the DT of the $ n $-DNLS equations referring to the DT of the scalar DNLS equation presented in Ref. \onlinecite{guo2013}. First, we can convert the system \eqref{eq:sp} into a new one that is
\begin{equation} \label{eq:uv2}
	\left\{
	\begin{aligned}
		&\Phi^{[1]}_{x}=\mathbf{U}^{[1]}\Phi^{[1]},  \qquad \mathbf{U}^{[1]}=\mathbf{U}|_{\mathbf{Q}=\mathbf{Q}^{[1]}},\\ &\Phi^{[1]}_{t}=\mathbf{V}^{[1]}\Phi^{[1]},  \qquad \mathbf{V}^{[1]}=\mathbf{V}|_{\mathbf{Q}=\mathbf{Q}^{[1]}},
	\end{aligned}
	\right.
\end{equation}
by the one-fold DT
\begin{equation}\label{eq:phi1}
	\begin{aligned}
		\Phi^{[1]}(\lambda;x,t)=\mathbf{T}_{1}(\lambda;x,t)\Phi(\lambda;x,t),
	\end{aligned}
\end{equation}
where
\begin{equation}\label{eq:uvt}
	\left\{
	\begin{aligned}
		&\mathbf{U}^{[1]}={(\mathbf{T}_{1})}_{x}\mathbf{T}_{1}^{-1}+ \mathbf{T}_{1}\mathbf{U}\mathbf{T}_{1}^{-1},\\
		&\mathbf{V}^{[1]}={(\mathbf{T}_{1})}_{t}\mathbf{T}_{1}^{-1}+ \mathbf{T}_{1}\mathbf{V}\mathbf{T}_{1}^{-1}.
	\end{aligned}
	\right.
\end{equation}
In order to enforce the symmetric relationships \eqref{eq:sy} on the novel potential functions $ \mathbf{U}^{[1]} $ and $ \mathbf{V}^{[1]} $, it is necessary to impose the conditions that the DT matrix $ \mathbf{T}_{1}(\lambda;x,t) $ adheres to the following equations:
\begin{equation}\label{eq:syt1}
	\begin{aligned}
	&\mathbf{T}_{1}(\lambda;x,t)=\sigma_{3}\mathbf{T}_{1}(-\lambda;x,t) \sigma_{3}, \\
	&(\mathbf{T}_{1}(\lambda;x,t))^{-1}=\sigma_{3}(\mathbf{T}_{1}(\lambda^{\ast};x,t))^{\dagger} \sigma_{3}.
	\end{aligned}
\end{equation}
Therefore, we can assume that
\begin{equation}\label{eq:t1}
	\mathbf{T}_{1}(\lambda;x,t)=\mathbb{I}_{n+1}+\frac{\mathbf{A}_{1}}{\lambda-\lambda_{1}^{\ast}}-\frac{\sigma_{3}\mathbf{A}_{1}\sigma_{3}}{\lambda+\lambda_{1}^{\ast}},
\end{equation}
where $ \mathbf{A}_{1} = |x_{1}\rangle \langle y_{1}| \sigma_{3} $, $ |x_{1}\rangle= |x_{1}(x,t)\rangle $ and $ \langle y_{1}|= \langle y_{1}(x,t)| $ are the $ (n+1)$-dimensional column vector and row vector, respectively. Moreover, the inverse matrix of $ \mathbf{T}_{1}(\lambda;x,t) $ can be expressed as
\begin{equation}\label{eq:ti1}
	(\mathbf{T}_{1}(\lambda;x,t))^{-1}=\mathbb{I}_{n+1}+ \frac{\sigma_{3}\mathbf{A}_{1}^{\dagger} \sigma_{3} }{\lambda-\lambda_{1}}-\frac{\mathbf{A}_{1}^{\dagger}}{\lambda+ \lambda_{1}}.
\end{equation}
When $ \phi_{1}(\lambda_{1};x,t) $ is the kernel of DT matrix $ \mathbf{T}_{1}(\lambda;x,t) $ with $ \lambda=\lambda_{1} $, we obtain the equation:
\begin{equation}\label{eq:ker}
	\left( \mathbb{I}_{n+1}+\frac{|x_{1}\rangle \langle y_{1}| \sigma_{3}}{\lambda_{1}-\lambda_{1}^{\ast}}-\frac{\sigma_{3}|x_{1}\rangle \langle y_{1}|}{\lambda_{1}+\lambda_{1}^{\ast}} \right) \phi_{1} = 0.
\end{equation}
Since
\begin{equation}\label{eqt1}
	\mathbf{T}_{1}(\lambda;x,t)(\mathbf{T}_{1}(\lambda;x,t))^{-1} =\mathbb{I}_{n+1},
\end{equation}
we have $ \mathrm{Res}|_{\lambda=\lambda_{1}} (\mathbf{T}_{1}(\lambda;x,t)(\mathbf{T}_{1}(\lambda;x,t))^{-1})=0 $, i.e.,
\begin{equation}\label{eq:res}
	\left(\mathbb{I}_{n+1}+\frac{|x_{1}\rangle \langle y_{1}| \sigma_{3}}{\lambda_{1}-\lambda_{1}^{\ast}}-\frac{\sigma_{3}|x_{1}\rangle \langle y_{1}| }{\lambda_{1}+\lambda_{1}^{\ast}} \right) \langle y_{1}|^{\dagger} |x_{1}\rangle^{\dagger} \sigma_{3}=0.
\end{equation}
As $ |x_{1}\rangle^{\dagger} \sigma_{3} \ne \mathbf{0} $, comparing Eqs. \eqref{eq:ker} and \eqref{eq:res}, we get $ \langle y_{1}|^{\dagger}=c_{1}\phi_{1} $, where $ c_{1} $ is a arbitrary constant. For convenience, by setting $ \langle y_{1}|^{\dagger}=\phi_{1} $, Eq. \eqref{eq:ker} is rewritten as
\begin{equation}\label{eq:ker2}
	\phi_{1}+\frac{\phi_{1}^{\dagger} \sigma_{3}\phi_{1}}{\lambda_{1}- \lambda_{1}^{\ast}}|x_{1}\rangle- \frac{ \phi_{1}^{\dagger}\phi_{1} }{\lambda_{1}+\lambda_{1}^{\ast}}\sigma_{3}|x_{1}\rangle=0.
\end{equation}
Suppose $ \phi_{1}=\left( \phi_{1,1}, \phi_{1,2}, \ldots, \phi_{1,n+1}\right) ^{T} $, then we calculate the vector $ |x_{1}\rangle $, as follows:
\begin{equation}\label{eq:px1}
	|x_{1}\rangle=((\lambda_{1}^{\ast})^{2}-\lambda_{1}^{2})\mathrm{diag}(\alpha_{1}^{-1}, \alpha_{2}^{-1}, \ldots, \alpha_{2}^{-1})\phi_{1},
\end{equation}
where
\begin{equation}\label{eq:ab}
	\begin{aligned}
		&\alpha_{1}=2\left(\lambda_{1}^{*} |\phi_{1,1}|^{2}- \lambda_{1} \sum_{i=2}^{n+1}|\phi_{1,i}|^{2} \right), \\ 
		&\alpha_{2}= 2\left(\lambda_{1} |\phi_{1,1}|^{2}- \lambda_{1}^{*} \sum_{i=2}^{n+1}|\phi_{1,i}|^{2} \right).
	\end{aligned}
\end{equation}
Furthermore, by expanding $ \mathbf{T}_{1}(\lambda;x,t) $ in the vicinity of $ \infty $:
\begin{equation}\label{eq:dti}
	\mathbf{T}_{1}(\lambda;x,t)=\mathbb{I}_{n+1}+\left( |x_{1}\rangle \langle y_{1}| \sigma_{3}-\sigma_{3}|x_{1}\rangle \langle y_{1}|  \right) \lambda^{-1}+\cO(\lambda^{-2}),
\end{equation}
we group the terms according to the power of $ \lambda $ in Eq. \eqref{eq:uvt}. To the term $ \cO(\lambda^{-1}) $, we yield
\begin{equation}\label{eq:qt1}
	q_{j}^{[1]}=q_{j}+\left[ 2((\lambda_{1}^{*})^{2}-\lambda_{1}^{2})\alpha_{2}^{-1}\phi_{1,j+1}\phi_{1,1}^{*}\right] _{x}, \quad j=1,2,\ldots,n.
\end{equation}

Next, the $ N $-fold DT of Eq. \eqref{eq:sp} is established by iterating the above DT step by step, as shown in the following theorem. For more details, refer to Ref. \onlinecite{guo2013}. 
\begin{theorem}\label{th1}
	Suppose that $\mathbf{q}(x,t) \in \mathbf{L}^{\infty}(\mathbb{R}^{2})\cap \mathbf{C}^{\infty}(\mathbb{R}^{2})$, and { the matrix functions $ \phi_{i}(\lambda_{i};x,t)=(\phi_{i,1},\phi_{i,2},\ldots, \phi_{i,n+1})^{T}$ is the solutions of the Lax pair \eqref{eq:sp}}. Then, the $ N $-fold DT is represented as
	\begin{equation}\label{eq:dtn}
		\begin{aligned}
			{\mathbf{T}_{N}(\lambda;x,t)}=\mathbb{I}_{n+1}+\sum_{i=1}^{N}\left( \frac{\mathbf{A}_{i}}{\lambda-\lambda_{i}^{*}}-\frac{\sigma_{3}\mathbf{A}_{i}\sigma_{3}}{\lambda+\lambda_{i}^{*}} \right) ,
		\end{aligned}
	\end{equation}
	converting the Lax pair \eqref{eq:sp} into the new one by replacing the potential functions
	\begin{equation}\label{eq:qn}
		\begin{aligned}
			\mathbf{Q}^{[N]}(\lambda;x,t)=\mathbf{Q}+\sum_{i=1}^{N}\left( \mathbf{A}_{i}-\sigma_{3}\mathbf{A}_{i}\sigma_{3} \right)_{x},
		\end{aligned}
	\end{equation}
	i.e.,
	\begin{equation}\label{eq:qn2}
		\begin{aligned}
			q_{j}^{[N]}(\lambda;x,t)=q_{j}+2\,\left( \mathbf{Y}_{N,j+1}\mathbf{M}_{2}^{-1}\mathbf{Y}_{N,1}^{\dagger} \right)_{x}, \quad 1\leq j\leq n,
		\end{aligned}
	\end{equation}
	where $ \mathbf{A}_{i} = |x_{i}\rangle \phi_{i}^{\dagger} \sigma_{3} $, $ |x_{i}\rangle $ are the $ (n+1) $-dimensional complex vectors,
	\begin{equation}\label{eq:xmi}
		\begin{aligned}
			&\mathbf{Y}_{N,k}= \left( \phi_{1,k}, \phi_{2,k}, \ldots, \phi_{N,k}\right), \quad 1\leq k\leq n+1,\\
			&\left( \left| x_{1,1}\right\rangle , \left|x_{2,1}\right\rangle, \ldots, \left|x_{N,1}\right\rangle \right) = \left( \phi_{1,1}, \phi_{2,1}, \dots, \phi_{N,1} \right) \mathbf{M}_{1}^{-1}, \\
			&\left( \left|x_{1,l}\right\rangle, \left|x_{2,l}\right\rangle, \ldots, \left|x_{N,l}\right\rangle \right) = \left( \phi_{1,l}, \phi_{2,l}, \dots, \phi_{N,l} \right) \mathbf{M}_{2}^{-1}, \\
		\end{aligned}
	\end{equation}
with $ 2\leq l\leq n+1 $ and
	\begin{equation}\label{eq:m12}
	\begin{aligned}
			&\mathbf{M}_{1}=\left( M^{[1]}_{i,j}\right)_{N\times N}, \quad \mathbf{M}_{2}=\left( M^{[2]}_{i,j}\right)_{N\times N}, \\
		&M^{[1]}_{i,j}=\frac{\phi_{i}^{\dagger}\sigma_{3}\phi_{j}}{\lambda_{i}^{*}-\lambda_{j}}+\frac{\phi_{i}^{\dagger}\phi_{j}}{\lambda_{i}^{*}+\lambda_{j}}, \quad
		M^{[2]}_{i,j}=\frac{\phi_{i}^{\dagger}\sigma_{3}\phi_{j}}{\lambda_{i}^{*}-\lambda_{j}}-\frac{\phi_{i}^{\dagger}\phi_{j}}{\lambda_{i}^{*}+\lambda_{j}}.
	\end{aligned}
	\end{equation}
\end{theorem}

\subsection{Generalized Wronskian-Hermite polynomial hierarchies}

As reported in Refs. \onlinecite{yang2021, yang2023, zhang2022}, the structures of rogue wave solutions for many $ (1+1) $-dimensional integrable equations (counting multi-component systems) are closely related to the root structures of the particular Wronskian-Hermite polynomial hierarchies. Thus, to effectively study the patterns and asymptotic behavior of higher-order rogue waves, we will introduce the generalized Wronskian-Hermite polynomial hierarchies in this subsection, as presented in Ref. \onlinecite{zhang2022}.

First, we introduce the elementary Schur polynomials $ S_{i}(\mathbf{z}) $ with $ \mathbf{z}=\left( z_{1}, z_{2}, \ldots\right)  \in \mathbb{C}^{\infty} $, where $ S_{i}(\mathbf{z}) $ are given by
\begin{equation}\label{eq:scpo}
	\sum_{i=0}^{\infty}S_{i}(\mathbf{z})\varepsilon^{i}=\exp\left(\sum_{j=1}^{\infty}z_{j}\varepsilon^{j} \right).
\end{equation}
Then, we define the special form $ p^{[m]}_{j}(z) $ of Schur polynomials by
\begin{equation}\label{eq:p1}
	\sum_{j=0}^{\infty}p^{[m]}_{j}(z)\varepsilon^{j}=\exp (z\varepsilon+\varepsilon^{m}),
\end{equation}
where $m>1 $ is an integer, $ p^{[m]}_{j}(z) = 0 $ for $ j<0 $, and $ p^{[m]}_{j}(z)=p^{[m]\prime}_{j+1}(z)$. The generalized Wronskian-Hermite polynomial hierarchies $ W_{N}^{[m,k,l]}(z) $ with the jump parameter $ k>0 $ are given by the following determinant:
\begin{equation}\label{eq:p2}
	\begin{aligned}
		W_{N}^{[m,k,l]}(z)=c_{N}^{[m,k,l]}\det_{1\leq i,j\leq N}\left( p^{[m]}_{k(j-1)+l-i+1}(z)\right), 
	\end{aligned}
\end{equation}
where $ c_{N}^{[m,k,l]} = k^{-\frac{N(N-1)}{2}} \prod_{i=1}^{N}\left( \dfrac{(k(i-1)+l)!}{(i-1)!} \right) $, $ m, k, l\in \mathbb{N}^{+} $, and $ 1\leq l<k $. 

The Yablonskii-Vorob’ev polynomials $ W^{[2m+1,2,1]}_{N}(z) $ and the Okamoto polynomials $ W^{[3m+j,3,l]}_{N}(z) (j,l=1,2)  $ are two particular cases of the above generalized Wronskian-Hermite polynomial hierarchies \eqref{eq:p2} with the jump parameters $ k=2 $ and $ k=3 $, which are studied in Refs. \onlinecite{yang2021, yang2023}. Additionally, Zhang et al. \cite{zhang2022} investigated the root structures of the generalized Wronskian-Hermite polynomials with jump parameters $k=4$ and $k=5$.
However, to our knowledge, there has yet to be an exploration into the root structures of generalized Wronskian-Hermite polynomials with the arbitrary jump parameter $k$.

Now, we present the following theorem for the root properties of the generalized Wronskian-Hermite polynomials $ W_{N}^{[m,k,l]}(z) $, where the jump parameter $ k $ is taken as an arbitrary positive integer. The theorem holds significance for our subsequent investigations into rogue wave patterns.

\begin{theorem}\label{th:whp}
	For the arbitrary positive integers $ m, k$, and $ l $ with $ m>1 $, $ k \nmid m $, and $ 1\leq l <k $, the generalized Wronskian-Hermite polynomials $ W_{N}^{[m,k,l]}(z) $ \eqref{eq:p2} are monic with the degree
	\begin{equation}\label{eq:gamma}
		\Gamma=\frac{N}{2}\left((k-1)(N-1)+2l \right),
	\end{equation}
	and are expressed as
	\begin{equation}\label{eq:whp}
		W_{N}^{[m,k,l]}(z)=z^{\Gamma_{0}}\hat{W}_{N}^{[m,k,l]}(\hat{z}), \quad \hat{z}=z^{m},
	\end{equation}
	where $ \hat{W}_{N}^{[m,k,l]}(\hat{z}) $ is a monic polynomial with real coefficients and a nonzero constant term, and $ \Gamma_{0} $ is the multiplicity of the zero root in $ W_{N}^{[m,k,l]}(z) $ given by
	\begin{equation}\label{eq:gamma0}
		\Gamma_{0}= \sum_{i=1}^{k-1}\frac{N_{i}}{2}((k-1)(N_{i}-1)+2i)-\sum_{1\leq i< j\leq k-1}N_{i}N_{j},
	\end{equation} 
	with $ N_{i} $ being defined by Eqs. \eqref{eq:ni1}--\eqref{eq:ni4} and \eqref{eq:nie1}--\eqref{eq:nie2}.
\end{theorem}
\begin{proof}
	See Appendix \ref{app4}.
\end{proof}

\begin{remark}\label{re1}
	It is worth noting that we make an assumption here: all non-zero roots of the generalized Wronskian-Hermite polynomials ${W}_{N}^{[m,k,l]}(z) $ are simple roots. Additionally, from Theorem \ref{th:whp}, we can generate three significant pieces of information: {firstly}, when $ k \mid m $, the terms containing the parameter $ a $ in determinant \eqref{scdex} can be effectively eliminated by the column transformations. Thus, based on Eq. \eqref{eq01}, we conclude that for $ k \mid m $, the polynomials ${W}_{N}^{[m,k,l]}(z)= z^{\Gamma} $. {Secondly}, for the polynomials $ {W}_{N}^{[m,k,l]}(z) $, the total count of non-zero roots, including their multiplicities, equals $ \Gamma-\Gamma_{0} $. {Thirdly}, if $ z_{0} $ is a non-zero root of $ {W}_{N}^{[m,k,l]}(z) $, then $ z_{0}\ee^{2\ii\pi j/m} $ also constitutes a non-zero polynomial root, where $ j=0,1,\ldots, m-1 $.  
\end{remark}

\section{Rogue wave solutions}\label{sec3}

In this section, we will construct the vector rogue wave solutions with nonzero seed solution $ n $-DNLS equation by applying the DT in Theorem \ref{th1}. Specifically, our primary emphasis lies in exploring rogue wave solutions under a specific scenario characterized by the presence of maximal multiple roots within the spectral characteristic polynomial.

First of all, we begin with the general nonzero seed (plane-wave)
solution
\begin{equation}\label{eq:seed}
	\begin{aligned}
		&{\mathbf{q}}=({q}_{1}, {q}_{2}, \dots, {q}_{n})^{T},\quad {q}_{j}=a_{j}e^{\ii \theta_{j}},\\
		&\theta_{j}=b_{{j}}(x+ (\| \mathbf{a}\|_{2}^{2} - b_{j}) t), \quad 1\leq j\leq n,
	\end{aligned}	
\end{equation}
where $\mathbf{a}=(a_{1}, a_{2}, \dots, a_{n}),\,  \, a_{j}\neq0$, and $a_{j}, b_{j}\in \mathbb{R}$. The MI analysis of plane-wave solution \eqref{eq:seed} is provided in Appendix \ref{app5} by the method of squared eigenfunction. Utilizing the squared eigenfunction method, we can obtain the perturbation vector functions given in Eq. \eqref{aa9}. Therefore, we can better comprehend the connection between rogue wave generation and MI.

Then, we reduce the variable coefficient partial differential equations \eqref{eq:sp} to the following linear constant-coefficient equations:
	\begin{equation} \label{eq:sp2}
		\left\lbrace 
		\begin{aligned}
			&\Psi_{x}=\ii\,\mathbf{H}\Psi,\\ 
			&\Psi_{t}=\ii\left[ {\mathbf{{H}}}^{2}+ \left(2{\lambda}^{-2}+\|\mathbf{a}\|_{2}^{2} \right) { \mathbf{H}} + {\lambda}^{-2}\left( \|\mathbf{a}\|_{2}^{2} -{\lambda}^{-2} \right)\mathbb{I}_{n+1}\right] \Psi,
		\end{aligned}
		\right. 
	\end{equation}
by employing a gauge transformation $ \Phi=\mathbf{G}\Psi$ with $\mathbf{G}=\mathrm{diag}(1,\mathrm{e}^{\ii\theta_{1}},\mathrm{e}^{\ii\theta_{2}},\dots,\mathrm{e}^{\ii\theta_{n}})$, where 
\begin{equation}\label{eq:hb}
	\mathbf{H}= \begin{pmatrix}
		\lambda^{-2} & -\ii\lambda^{-1}\mathbf{a} \\
		-\ii\lambda^{-1}\mathbf{a}^{T} & -\lambda^{-2}\,\mathbb{I}_{n}-\mathbf{b}
	\end{pmatrix},\quad \mathbf{b}=\mathrm{diag}(b_{1},b_{2},\dots,b_{n}).
\end{equation} 

Thus, we generate the vector solution of the system \eqref{eq:sp} with the plane-wave solution \eqref{eq:seed}, as follows:
\begin{equation} \label{eq22} 
	\begin{aligned}
		\Phi=\mathbf{G}\mathbf{E}\,\mathrm{diag}\left( \ee^{\zeta_{1}},\ee^{\zeta_{2}}, \ldots, \ee^{\zeta_{n+1}} \right) , \qquad 
	\end{aligned}
\end{equation}
where  
\begin{equation} \label{eq22b}  
	\begin{aligned}
		&\mathbf{E}=
		\begin{pmatrix}
			1 & 1 & \cdots & 1\\
			\frac{-\ii a_{1}}{\lambda(\chi_{1}+b_{1})} & \frac{-\ii a_{1}}{\lambda(\chi_{2}+b_{1})} & \cdots & \frac{-\ii a_{1}}{\lambda(\chi_{n+1}+b_{1})}\\
			\frac{-\ii a_{2}}{\lambda(\chi_{1}+b_{2})} & \frac{-\ii a_{2}}{\lambda(\chi_{2}+b_{2})} & \cdots & \frac{-\ii a_{2}}{\lambda(\chi_{n+1}+b_{2})}\\
			\vdots &\vdots & \ddots & \vdots\\
			\frac{-\ii a_{n}}{\lambda(\chi_{1}+b_{n})} & \frac{-\ii a_{n}}{\lambda(\chi_{2}+b_{n})} & \cdots & \frac{-\ii a_{n}}{\lambda(\chi_{n+1}+b_{n})}
		\end{pmatrix},\\
		&\zeta_{j}=\ii\left( (\chi_{j}-\lambda^{-2})x + (\chi_{j}^{2} +\|\mathbf{a}\|_{2}^{2}\chi_{j} -2\lambda^{-4} )t \right) +d_j(\lambda),
	\end{aligned}
\end{equation}
with $ 1\leq j\leq n+1 $, the arbitrary parameter $d_{j}(\lambda)$ is independent of $x$ and $t$, $\chi_{j}$ {satisfies} the following characteristic polynomial
\begin{equation} \label{eq:cp1} 
	\begin{aligned}
		D(\chi,\lambda)=\left( (\chi-2\lambda^{-2})+\lambda^{-2}\sum_{i=1}^{n}\frac{{a_{i}^{2}}}{\chi+b_{i}}\right) \prod_{j=1}^{n}(\chi+b_{j}).		
	\end{aligned}
\end{equation}

\subsection{Maximal multiple roots of the spectral characteristic polynomial}

Here, we will investigate a specific case in which the coefficient matrix $ \mathbf{H} $ possesses an $ (n+1) $-multiple eigenvalue at $ \lambda=\lambda_{0} $, i.e., the equation $ D(\chi,\lambda_{0})=0 $ has an $ (n+1) $-multiple root ${\chi_0}$. Then, the existence condition of $ (n+1) $-multiple eigenvalue is determined in the following proposition. 

\begin{prop}\label{pr1}
	The spectral characteristic polynomial \eqref{eq:cp1} has an $ (n+1) $-multiple root $ \chi_{0} $ if and only if $ b_{i}\ne b_{j}$  $ (i\ne j)$ and 
	\begin{equation}\label{eq:ab2}
		a_{j}^{2}=\frac{(b_{j}+\chi_{0})^{n+1}}{\xi_{0}\prod_{i\ne j}(b_{j}-b_{i})}, \quad 
		\chi_{0}= \frac{2\xi_{0}-\sum_{j=1}^{n}b_{j}}{n+1}, \quad \xi_{0}=\lambda_{0}^{-2},
	\end{equation}
	where $ a_{j} $ and $ b_{j} $ are {defined in Eq.} \eqref{eq:seed}.
\end{prop}
\begin{proof}
	See Appendix \ref{app6}.
\end{proof}

In order to obtain suitable parameters $ a_{j} $ and $ b_{j} $ that satisfy the conditions in Proposition \ref{pr1}, we present the following propositions.

\begin{prop}\label{pr2}
	When $ \Re(\lambda_{0})=0 $ or $ \Im(\lambda_{0})=0 $, the possibility of simultaneously satisfying condition \eqref{eq:ab2} with non-zero real values of $ a_{j} $ and distinct real values of $ b_{j} $ becomes unattainable. In other words, $ {\chi_{0}} $ is not an $ (n+1) $-multiple eigenvalue of $ \mathbf{H} $ in this scenario. 
\end{prop}
\begin{proof}
	Without loss of generality, suppose that $ b_{i}>b_{j} (1\leq i<j\leq n) $.
	If $ \Im(\lambda_{0})=0 $, we get the fact that $ \xi_{0}>0 $ and $ \chi_{0}\in \mathbb{R} $. For $ n $ is odd, we have $ a_{2}^{2}=\frac{(b_{2}+\chi_{0})^{n+1}}{\xi_{0}\prod_{i\ne 2}(b_{2}-b_{i})} <0 $, which is a contradiction.
	For $ n $ is even, if $ b_{2}+\chi_{0}>0 $, we derive $ a_{2}^{2}=\frac{(b_{2}+\chi_{0})^{n+1}}{\xi_{0}\prod_{i\ne 2}(b_{2}-b_{i})} <0 $, which is contradictory. In contrast, if $ b_{2}+\chi_{0}\leq 0$, we can infer $ b_{3}+\chi_{0}< 0 $ and calculate $ a_{3}^{2}=\frac{(b_{3}+\chi_{0})^{n+1}}{\xi_{0}\prod_{i\ne 3}(b_{3}-b_{i})} <0 $, which is also a contradiction.
	
	{Similarly}, we can prove the case as $ \Re(\lambda_{0})=0 $. Thus, we prove Proposition \ref{pr2}. 
\end{proof}

\begin{prop}\label{pr3}
	For $ \Re(\lambda_{0})\ne 0 $ and $ \Im(\lambda_{0})\ne 0 $, if the parameters $ a_{j} $ and $ b_{j} $ satisfy
	\begin{equation}\label{eq:ab3}
		\begin{aligned}
			&a_{j}=\dfrac{2\Im(\xi_{0})\csc \omega_{j}}{(n+1)|\xi_{0}|^{1/2}} \prod_{i\ne j}\sqrt{\dfrac{\sin \omega_{i}}{|\sin \frac{(i-j)\pi}{n+1}|}}, \\
			&b_{j}=\frac{2\Im(\xi_{0})}{n+1} \left( \cot \omega_{j}+\sum_{i=1}^{n}\cot \omega_{i} \right) -2\Re(\xi_{0}), \\
			& \omega_{j}=\frac{(j-1)\pi+\arg \xi_{0}}{n+1}, 
		\end{aligned}
	\end{equation}
	the polynomial \eqref{eq:cp1} has an $ (n+1) $-multiple root
	\begin{equation}\label{chi0}
		\chi_{0}=2\Re(\xi_{0}) +\frac{2\Im(\xi_{0})}{n+1}\left( \ii- \sum_{i=1}^{n}\cot \omega_{i} \right).
	\end{equation}
\end{prop}
\begin{proof}
	As $ a_{j}, b_{j}\in \mathbb{R} $, based on the expression of $ a_{j}^{2} $ in Eq. \eqref{eq:ab2}, we have
	\begin{equation}\label{eq1}
		 \Im(a_{j}^{2})= \Im(\frac{(b_{j}+\chi_{0})^{n+1}}{\xi_{0}})=0
	\end{equation}
	Then, we can assume that
	\begin{equation}\label{eq:bchi}
		b_{j}+\chi_{0} = \rho_{j} \ee^{\ii \omega_{j}},
	\end{equation}
	where $ \omega_{j} $ is shown in Eq. \eqref{eq:ab3}. According to the expression of $ \chi_{0} $ in Eq. \eqref{eq:ab2}, we obtain
	\begin{equation}\label{eq2}
		b_{j}+ \frac{2\xi_{0}-\sum_{j=1}^{n}b_{j}}{n+1}= \rho_{j} \ee^{\ii \omega_{j}},
	\end{equation}
	By comparing the real and imaginary parts of both sides of the above Eq. \eqref{eq2}, we derive
	\begin{equation}\label{rhob}
		\begin{aligned}
			&\rho_{j}= 2\Im(\xi_{0})\csc\omega_{j} ,\\
			&b_{j}=\frac{2\Im(\xi_{0})}{n+1}\cot\omega_{j}- \frac{1}{(n+1)}\left( 2\Re(\xi_{0})-\sum_{i=1}^{n}b_{i} \right).
		\end{aligned}
	\end{equation}
	Further, we calculate
	\begin{equation}\label{sumbj}
		\sum_{i=1}^{n}b_{i}=\sum_{i=1}^{n}\rho_{i}\cos\omega_{i}-2n\Re(\xi_{0}).
	\end{equation}
	Then, by substituting Eq. \eqref{sumbj} into the expressions of $ b_{j} $ and $ \chi_{0} $ in Eq. \eqref{rhob} and \eqref{eq:ab2}, we can get the expressions of $ b_{j} $ and $ \chi_{0} $ in \eqref{eq:ab3}. Furthermore, we can deduce 
	\begin{equation}\label{eq3}
		a_{j}=\pm \dfrac{2\Im(\xi_{0})\csc \omega_{j}}{(n+1)|\xi_{0}|^{1/2}} \prod_{i\ne j}\sqrt{\dfrac{\sin \omega_{i}}{|\sin \frac{(i-j)\pi}{n+1}|}},
	\end{equation}
	by substituting the expressions \eqref{eq:ab3} of $ b_{j}$ and $ \chi_{0} $ into the expression of $ a_{j} $ in \eqref{eq:ab2}. Without loss of generality, we take the positive sign for the expression of $ a_{j} $ in \eqref{eq3}. This completes the proof.

\end{proof}

\subsection{The determinant formula of rogue wave solutions}

In this subsection, referring to the rogue solution formula obtained using the bilinear method, we employ the degenerate $ N $-fold DT to generate a determinant formula of higher-order vector rogue wave solutions for $ n $-DNLS equations \eqref{eq:nDNLS} when the above characteristic polynomial \eqref{eq:cp1} possesses a highest-multiple root.

Based on the nonzero seed solution \eqref{eq:seed}, we rewrite the solution formula of $ n $-DNLS equation \eqref{eq:nDNLS} in Eq. \eqref{eq:qn2} as
\begin{equation}\label{eq:qn3}
	{q_{k}^{[N]}(x,t)}=\left( -\frac{\ii {q}_{k}}{b_{k}}\frac{\det(\mathbf{M}^{[k]})}{\det(\mathbf{M}^{[0]})}\right)_{x}, \quad 1\leq k\leq n,
\end{equation}
where $ \mathbf{M}^{[k]}=\mathbf{M}^{[0]}+2\ii b_{k}{q}_{k}^{-1} \mathbf{Y}_{N,1}^{\dagger} \mathbf{Y}_{N,k+1} $, $ \mathbf{M}^{[0]}=\mathbf{M}_{2}$, $ \mathbf{M}_{2} $ and $ \mathbf{Y}_{N,k} $ are defined in Theorem \ref{th1}. Furthermore, by setting the eigenfunction $ \phi_{i}={\Phi	(c_{i,1}, c_{i,2}, \cdots, c_{i,n+1})^{T}}|_{\lambda=\lambda_{i}}  $ given in Eq.\eqref{eq22} {with the arbitrary constants $ c_{i,v}$ $  (1\leq v\leq n+1) $}, the elements $ M_{i,j}^{[s]} $ of the matrices $ \mathbf{M}^{[s]} $ $ (0\leq s\leq n) $ in Eq. \eqref{eq:qn3} can be simplified to
\begin{equation}\label{eqm1}
	\begin{aligned}
		&M_{i,j}^{[0]}= \sum_{p,r=1}^{n+1}c_{i,p}^{*}c_{j,r} \frac{2\,\chi_{i,p}^{*}}{\chi_{j,r}-\chi_{i,p}^{*}} \ee^{\zeta_{i,p}^{*}+\zeta_{j,r}},\\
		&M_{i,j}^{[k]}= \sum_{p,r=1}^{n+1}c_{i,p}^{*}c_{j,r} \frac{2\,\chi_{j,r}}{\chi_{j,r}-\chi_{i,p}^{*}} \frac{\chi_{i,p}^{*}+b_{k}}{\chi_{j,r}+b_{k}} \ee^{\zeta_{i,p}^{*}+\zeta_{j,r}},\,1\leq k\leq n,
	\end{aligned}
\end{equation}
where $ \chi_{i,p}^{*} $ and $ \chi_{j,r} $ represent the $ p $-th and $ r $-th roots of Eq. \eqref{eq:cp1} corresponding to $ \lambda=\lambda_{i}^{*}  $ and $ \lambda=\lambda_{j} $, respectively. Additionally, $\zeta_{i,p}^{*}  $ and $\zeta_{j,r}  $ are determined by Eq. \eqref{eq22b} under the parameter condition of $ \chi=\chi_{i,r}^{*} $, $ \lambda=\lambda_{i}^{*} $, and $ \chi=\chi_{j,r} $, $ \lambda=\lambda_{j} $, individually. 

Next, we introduce a perturbation for the parameter $ \xi_{0} $ with $ \xi_{0}=\lambda_{0}^{-2} $. Let
\begin{equation}\label{xichi}
	\xi(\varepsilon) =\xi_{0} +\xi^{[1]}\varepsilon^{n+1}, \qquad \chi(\varepsilon)=\chi_{0}+ \sum_{i=1}^{\infty}\chi^{[i]}\varepsilon^{i}, 
\end{equation}
where $ \varepsilon $ is a small perturbation parameter, and $ \chi_{0} $ is given in Proposition \ref{pr3}. Then, we can obtain the solution of spectral problem \eqref{eq:sp} with seed solution \eqref{eq:seed}, as follows:
\begin{equation}\label{phiep}
	\phi(\varepsilon)=\Phi(\varepsilon)\mathbf{C},
\end{equation}
where $ \mathbf{C} $ is an $ (n+1) $-dimensional arbitrary constant column vector, $ \Phi(\varepsilon) $ is determined by Eq. \eqref{eq22} with $ \lambda=\lambda(\varepsilon) $, $ \chi_{j}=\chi(\varepsilon {\hat{\omega}_{n}^{j-1}}) $, and $ {\hat{\omega}_{n}}=\ee^{\frac{2\pi \ii}{n+1}} $. Moreover, for the convenience of subsequent proofs, we denote $ (\phi({\varepsilon}))^{\dagger} $ as $ \psi(\hat{\varepsilon}) $, where $ \hat{\varepsilon}=\varepsilon^{*} $. When $ \phi(\varepsilon) $ and $ \psi(\hat{\varepsilon}) $ are expanded at $ \varepsilon= \hat{\varepsilon}=0 $, i.e.,
\begin{equation}\label{eq:psi}
	\phi(\varepsilon)= \sum_{i=0}^{\infty}\phi^{[i]}\varepsilon^{i}, \quad \psi(\hat{\varepsilon})= \sum_{i=0}^{\infty}\psi^{[i]}\hat{\varepsilon}^{i}, 
\end{equation}
the coefficients of expansions admit $ \psi^{[i]} =(\phi^{[i]})^{\dagger} $.

To facilitating the generation of the formula for higher-order rogue wave solutions of $ n $-DNLS equations \eqref{eq:nDNLS}, we present the following proposition here.
\begin{prop}\label{pr4}
	For the arbitrary positive integer $ j $ and $ n $, the parameter $ \hat{\omega}_{n}=\ee^{\frac{2\pi \ii}{n+1}} $ satisfies the following equation:
	\begin{equation}\label{ompr}
		\sum_{i=0}^{n}\hat{\omega}_{n}^{ij}= \left\lbrace
		\begin{array}{ll}
			 n+1, & n+1 \mid j,\\
			0, & n+1 \nmid j.
		\end{array}
		\right.
	\end{equation} 
	
\end{prop}
\begin{proof}
	Since $ \hat{\omega}_{n}=\ee^{\frac{2\pi \ii}{n+1}} $ and $ \hat{\omega}_{n}^{n+1}=1 $, for $ (n+1) \mid j $, we directly calculate $ \sum_{i=0}^{n}\hat{\omega}_{n}^{ij}=n+1 $. Moreover, for $ (n+1) \nmid j $, we obtain
	\begin{equation}\label{ompr2}
		\sum_{i=0}^{n}\hat{\omega}_{n}^{ij}=\frac{\hat{\omega}_{n}^{j(n+1)}-1}{\hat{\omega}_{n}^{j}-1}=0.
	\end{equation}
\end{proof}

Furthermore, let us introduce some coefficients and notations:
$ \zeta^{[l]}=\left( \zeta_{1}^{[l]},\zeta_{2}^{[l]}, \cdots\right)  $ and $ \mathbf{h}_{i}=\left( h_{i,1},h_{i,2},\cdots\right) $ $(1\leq i\leq 4)$ are determined by
\begin{equation}\label{eq:zh}
	\begin{aligned}
		&\zeta^{[l]}_{j}={\gamma}_{j}x+\hat{\gamma}_{j} t+ d^{[l]}_{j}, \quad 
		\ln\left( \frac{\chi(\varepsilon)}{\chi_{0}}\right)= \sum_{i=1}^{\infty}h_{1,i}\varepsilon^{i}, \\
		& 	\ln\left( \dfrac{(\chi_{0}-\chi_{0}^{*})}{\chi^{[1]}\varepsilon} \dfrac{(\chi(\varepsilon)-\chi_{0})}{(\chi(\varepsilon)-\chi_{0}^{*})}\right)  =\sum_{i=1}^{\infty}h_{2,i}\varepsilon^{i},\\
		&\ln\left( \dfrac{(\chi_{0}-\chi_{0}^{*})}{(\chi(\varepsilon)-\chi_{0}^{*})}\right)= \sum_{i=0}^{\infty}h_{3,i}\varepsilon^{i}, \,
		\ln\left( \frac{\chi(\varepsilon)+b_{k}}{\chi_{0}+b_{k}}\right) = \sum_{i=0}^{\infty}h_{4,i}^{[k]}\varepsilon^{i}, 
	\end{aligned}
\end{equation}
where $ 1\leq k \leq n $, $ d^{[l]}_{j} $ are arbitrary constants, $\gamma_{j}$ and $\hat{\gamma}_{j} $ are defined by the following expansions
\begin{equation}\label{eq:albe}
	\begin{aligned}
		&\ii\left( \chi(\varepsilon)-\xi(\varepsilon)\right)  = \sum_{j=0}^{\infty}\gamma_{j} \varepsilon^{j}, \\
		&\ii\left(\chi(\varepsilon)^{2} +\|\mathbf{a}\|_{2}^{2}\chi(\varepsilon) -2\xi(\varepsilon)^{2} \right)  = \sum_{j=0}^{\infty}\hat{\gamma}_{j}\varepsilon^{j}.
	\end{aligned}
\end{equation}

Now, we consider all eigenvalues $ \lambda \rightarrow \lambda(\varepsilon)$, and perform a Taylor expansion around $ \varepsilon $. Subsequently, we generate the formula of higher-order rogue wave solutions for $ n $-DNLS equations \eqref{eq:nDNLS}, as shown in the following theorem.

\begin{theorem}\label{th3}
	Given an $ n $-dimension integer vector $ \mathcal{N}= [N_{1}, N_{2}, \cdots, N_{n}] $ with $N_{i}\geq 0$ $ (1 \leq i\leq n)$ and $ \sum_{i=1}^{n}N_{i}=N $, the formula of vector rogue wave solutions $ \mathbf{q}^{[\mathcal{N}]}=[{q}_{1}^{[\mathcal{N}]}, {q}_{2}^{[\mathcal{N}]}, \cdots, {q}_{n}^{[\mathcal{N}]} ] $ for $ n $-DNLS equations \eqref{eq:nDNLS} is expressed as
	\begin{equation}\label{eq:qnrw}
		q_{k}^{[\mathcal{N}]}= \left( -\frac{\ii {q}_{k}}{b_{k}}\frac{\det(\mathbf{M}^{[k]})}{\det(\mathbf{M}^{[0]})}\right)_{x}, \qquad 1\leq k \leq n,
	\end{equation}
	where $ \mathbf{M}^{[s]} (0\leq s\leq n)$ are defined as the following $ n\times n $ block matrix
	\begin{equation}\label{eq:qnmk}
		\begin{aligned}
			&\mathbf{M}^{[s]}=\begin{pmatrix} 
				M^{[s;{1},{1}]} &  M^{[s;{1},{2}]} & \cdots & M^{[s;{1},{n}]} \\
				M^{[s;{2},{1}]} &  M^{[s;{2},{2}]} & \cdots & M^{[s;{2},n]}\\
				\vdots & \vdots  & \ddots  & \vdots \\
				M^{[s;n,1]} &  M^{[s;n,2]} &  \cdots & M^{[s;n,n]}
			\end{pmatrix}_{N\times N},  \\
			&M^{[s;l_{1},l_{2}]}= \left( \tau^{[s]}_{(n+1)(i-1)+l_{1},(n+1)(j-1)+l_{2}} \right)_{1\leq i\leq N_{l_{1}}, 1\leq j\leq N_{l_{2}}},
		\end{aligned}
	\end{equation}
	with $ 1\leq l_{1},l_{2}\le n $, the elements $ \tau^{[s]}_{i,j} $ are given by
		\begin{equation}\label{eq:tau}
			\begin{aligned}
				\tau^{[0]}_{i,j}=\sum_{\mu=0}^{\min{(i,j)}} &\left( C_{1}^{\mu} \, S_{i-\mu}((\zeta^{[l_{1}]})^{*} +\mathbf{h}_{1}^{*} +v\mathbf{h}_{2}^{*} +\mathbf{h}_{3}^{*}) \right. \\ &\left. \times S_{j-\mu}(\zeta^{[l_{2}]} +\mu\mathbf{h}_{2} +\mathbf{h}_{3}) \right) , \\
				\tau^{[k]}_{i,j}=\sum_{\mu=0}^{\min{(i,j)}} &\left( C_{1}^{\mu} \, S_{i-\mu}((\zeta^{[l_{1}]})^{*}+\mu\mathbf{h}_{2}^{*}+\mathbf{h}_{3}^{*} +{(\mathbf{h}_{4}^{[k]})}^{*}) \right.\\ 
				&\left.\times  S_{j-\mu}(\zeta^{[l_{2}]} +\mathbf{h}_{1} +\mu\mathbf{h}_{2} +\mathbf{h}_{3}-{\mathbf{h}_{4}^{[k]}})\right) ,
			\end{aligned}
		\end{equation}
	the constant $ C_{1}= \frac{|\chi^{[1]}|^{2}}{|\chi_{0}-\chi_{0}^{*}|^{2}} $, $ S_{i} $ are defined by Eq. \eqref{eq:scpo}, $ \zeta^{[l]}$ and $ \mathbf{h}_{i}$ are given by Eq. \eqref{eq:zh}.

\end{theorem}
\begin{proof}
	See Appendix \ref{app10}.
\end{proof}

Note that $ \mathcal{N}=[N_{1},N_{2},\cdots,N_{n}] $ largely affects the structures of rogue waves. Moreover, the size of block matrix \eqref{eq:qnmk} is determined by the number  $ N_{i} $ of non-zeros in $ \mathcal{N}$. For $ N_{i}=0 $, block matrix \eqref{eq:qnmk} does not have the submatrices $ M^{[s;i,j]} $ and $ M^{[s;j,i]} $ with $ 1\leq j\leq n $. 
Additionally, another factor that affects the structure of higher-order rogue waves is the lots of internal free parameters $ (d_{1}^{[l]}, d_{2}^{[l]}, \ldots)$ $ (1\leq l\leq n) $ in solution expression \eqref{eq:qnrw}. 

In this paper, we mainly study the rogue wave solution with the condition that is $ N_{l}=N $ and $ N_{i}=0$ $ (i\ne l, 1\leq l,i\leq n) $. Now, we obtain
\begin{equation}\label{eq:qnmk2}
	\begin{aligned}
		\det ({ \mathbf{M}^{[s]}})=\det_{1\leq i,j\leq N}\left( \tau_{(n+1)(i-1)+l,(n+1)(j-1)+l}^{[s]} \right),
	\end{aligned}
\end{equation}	
where $ 0\leq s\leq n $, $ \tau_{i,j}^{[s]} $ are given in Eq. \eqref{eq:tau}. Then, we rewrite the above determinant \eqref{eq:qnmk2} as
\begin{equation}\label{eq:qnmk3}
	\det ({ \mathbf{M}^{[s]}})=\begin{vmatrix} 
		\mathbf{0}_{N\times N} & -M_{l}^{[s,-]} \\
		M_{l}^{[s,+]} & \mathbb{I}_{(n+1)N\times (n+1)N}
	\end{vmatrix},
\end{equation}
where the matrices $ M_{l}^{[s,+]}=(\tau^{[s,+]}_{l;i,j})_{(n+1)N\times N} $ and $ M_{l}^{[s,-]}=(\tau_{l;i,j}^{[s,-]})_{N\times (n+1)N} $ with
\begin{widetext}
	\begin{equation}\label{eq:tau2}
		\begin{aligned}
			&\tau^{[0,+]}_{l;i,j}=C_{1}^{\frac{i-1}{2}} S_{(n+1)(j-1)+l-i+1}\left(\zeta^{[l]}+ (i-1)\mathbf{h}_{2} +\mathbf{h}_{3} \right), \\ 
			&\tau^{[k,+]}_{l;i,j}=C_{1}^{\frac{i-1}{2}} S_{(n+1)(j-1)+l-i+1}\left(\zeta^{[l]}+\mathbf{h}_{1}+(i-1)\mathbf{h}_{2}+\mathbf{h}_{3}-{\mathbf{h}_{4}^{[k]}} \right), \\ 
			&\tau^{[0,-]}_{l;i,j}=C_{1}^{\frac{j-1}{2}} S_{(n+1)(i-1)+l-j+1}\left((\zeta^{[l]})^{*}+\mathbf{h}_{1}^{*}+ (j-1) \mathbf{h}_{2}^{*}+\mathbf{h}_{3}^{*}\right),\\
			&\tau^{[k,-]}_{l;i,j}=C_{1}^{\frac{j-1}{2}} S_{(n+1)(i-1)+l-j+1}\left((\zeta^{[l]})^{*}+(j-1)\mathbf{h}_{2}^{*}+\mathbf{h}_{3}^{*} +{(\mathbf{h}_{4}^{[k]})}^{*} \right), 
		\end{aligned}
	\end{equation}
\end{widetext}
the superscript $ 1\leq k\leq n $, and the free parameters $ (d_{1}^{[l]}, d_{2}^{[l]}, \ldots) $. For convenience, we denote $ d_{j}^{[l]} $ $(j=1, 2,\ldots) $ as $ d_{j} $. 

In particular, when taking $ N=1 $, $ l=1 $, and $ d_{1}=0 $ and employing formula \eqref{eq:qnrw} in Theorem \ref{th3}, we can directly calculate the fundamental vector rogue wave solutions $ \mathbf{q}^{[1]}= ( {q}_{1}^{[1]}, {q}_{2}^{[1]}, \cdots, {q}_{n}^{[1]} ) $ of $ n $-DNLS equations \eqref{eq:nDNLS}, as follows:
\begin{equation}\label{eq:q1}
	q_{k}^{[1]}=a_{k}\ee^{\ii \theta_{k}}\dfrac{L_{0}^{*} L_{k} }{L_{0}^{2}},\quad 1\leq k\leq n,
\end{equation}
where 
\begin{equation}\label{eq:q12}
	\begin{aligned}
		&L_{0}=\bar{x}^{2}+4\Im(\chi_{0})^{2}t^{2} +\frac{\ii}{\chi_{0}^{*}}(\bar{x} +2\ii \Im(\chi_{0})t) +\frac{1}{4\Im(\chi_{0})^{2}}, \\ 
		&L_{k}=\bar{x}^{2}+4\Im(\chi_{0})^{2}t^{2} + {P_{1,k}\bar{x}+P_{2,k}t+P_{3,k}}, \\
		&\bar{x}= x+(2\Re(\chi_{0}) +\|\mathbf{a}\|^{2}_{2})t +\frac{1}{2\Im{(\chi_{0})}}, \\
		&P_{1,k}=
		\ii\dfrac{(2\chi_{0}^{*}+b_{k}) (2\Re(\chi_{0})+b_{k})-(\chi_{0}^{*})^{2}}{\chi_{0}^{*}|\chi_{0}+b_{k}|^{2}}, \\
		&P_{2,k}=-\frac{2\Im(\chi_{0})\left( \chi_{0}^{2}+4\Im(\chi_{0})^{2}+b_{k}(2\Re(\chi_{0})+b_{k})\right) }{\chi_{0}^{*}|\chi_{0}+b_{k}|^{2}},\\
		& P_{3,k}={\frac {{\chi_{{0}}}^{3}-2\, \left( {\chi_{0}^{*}}^{2}+\ii\chi_{{0}}\Im \left( \chi_{{0
				}} \right)  \right)  \left( 2\,\ii\Im \left( \chi_{{0}} \right) -b_{{k}}
				\right) +\chi_{0}^{*}\,{b_{{k}}}^{2}
			}
			{ 4 \chi_{0}^{*}\left(\Im \left( \chi_{{0}}
				\right)  \right) ^{2} | \chi_{{0}}+b_{{k}}|^{2} }}.
	\end{aligned}
\end{equation}
It is evident that the $ 1 $st-order rogue wave $ |q_{k}^{[1]}|^{2} \rightarrow a_{k}^{2} $ when $ x\rightarrow \infty $ and $ t \rightarrow \infty $. The maximum amplitude of $ |q_{k}^{[1]}|^{2} $ is
\begin{equation}\label{q1max}
	a_{k}^{2}\left|1-\frac{4\Im(\chi_{0})^{2}\chi_{0}(2\chi_{0}^{*}+b_{k}) }{|\chi_{0}|^{2} |\chi_{0}+b_{k}|^{2}} \right|^{2},
\end{equation}
which appears at the position $ (x,t)=(-\frac{1}{2\Im(\chi_{0})}, 0) $. Then, by performing a coordinate transformation {$ x\to x+\frac{1}{2\Im(\chi_{0})} $},  we can make the maximum amplitude of the {first}-order rogue wave occur at the origin of the $ (x, t) $ plane.

\section{Rogue wave pattern and asympototics analysis}\label{sec4}

In this section, based on selecting different parameters $ l $ and various maximal internal parameters $ d_{m} (m\geq 1)$, we obtain higher-order vector rogue wave solutions for the $ n $-DNLS equation \eqref{eq:nDNLS} with different structures. Additionally, a detailed asymptotic analysis of these high-order rogue waves is provided. We discover a significant correlation between the structure of these rogue wave solutions and the generalized Wronskian-Hermite polynomials. Finally, we present several illustrative examples to showcase and validate our findings vividly.

\subsection{Higher-order rogue wave patterns}

When taking $ \mathcal{N}_{l}=N\mathbf{e}_{l} $, we obtain the $ l $-type $ \mathcal{N}_{l} $-order vector rogue wave solution of $ n $-DNLS equation \eqref{eq:nDNLS} by the formulas \eqref{eq:qnrw} and \eqref{eq:qnmk2}, where $ \mathbf{e}_{l} $ is the $ l $-th standard unit vector in the $ n $-dimensional vector. For this type of rogue wave solutions, selecting different values for the internal large parameter $ d_{m} $ results in distinct structural patterns. These patterns closely resemble the root structures of the generalized Wronskian-Hermite polynomials $ W_{N}^{[m,n+1,l]}(z) $. Moreover, given that the influence of $ d_{1} $ on the rogue wave structure is trivial, we will refrain from discussing it in detail here. Similarly, as $ d_{i(n+1)} (i\geq 1)$ does not affect the rogue wave structure, it is permissible to set $ d_{i(n+1)}=0 $. Consequently, we present the main result regarding the patterns of $ l $-type $ \mathcal{N}_{l} $-order rogue wave in the ensuing theorem.

\begin{theorem}\label{th4}
	When $d_{m} (m\geq2)$ is large enough and all other parameters are $ \cO(1) $, the $ l $-type $ \mathcal{N}_{l} $-order vector rogue wave solution
	\begin{equation}\label{eq:qlm}
		\mathbf{q}^{[\mathcal{N}_{l}]}(x,t)=\left[ {q}_{1}^{[\mathcal{N}_{l}]}(x,t), {q}_{2}^{[\mathcal{N}_{l}]}(x,t), \cdots, {q}_{n}^{[\mathcal{N}_{l}]}(x,t) \right]^{{T}} 
	\end{equation}
	of the $ n $-DNLS equation \eqref{eq:nDNLS} admits the following asympototics:
	\begin{enumerate}[1.]
		\item In the outer region with $ \sqrt{x^{2}+t^{2}}\geq \mathcal{O}(|d_{m}|^{1/m}) $, the solutions $ q^{[ \mathcal{N}_{l} ]}_{k}(x,t) (1\leq k\leq n)$ separate into $ \Gamma-\Gamma_{0} $ 1st-order rogue waves $ q^{[1]}_{k}(x-x_{0},t-t_{0}) $, where $ \Gamma $ and $ \Gamma_{0} $ are given in Eqs. \eqref{eq:gamma} and \eqref{eq:gamma0} with $ k=n+1 $, $ q^{[1]}_{k}(x,t) $ are presented in Eq. \eqref{eq:q1},
		\begin{equation}\label{eq:xt0}
			\begin{aligned}
				x_{0}\,&=\,\Im \left[\frac{z_{0}d_{m}^{1/m} }{\chi^{[1]}}\right]+ \frac{2\Re(\chi_{0})+\|\mathbf{a}\|^{2}_{2} }{2\Im(\chi_{0})} \Re\left[\frac{z_{0}d_{m}^{1/m} }{\chi^{[1]}}\right]  + \mathcal{O}(1),\\
				t_{0}\,&=\,-\frac{1}{2\Im(\chi_{0})}\Re \left[\frac{z_{0}d_{m}^{1/m} }{\chi^{[1]}}\right] + \mathcal{O}(1), 
			\end{aligned}
		\end{equation}
		$ \chi^{[1]} $ is defined by Eq. \eqref{xichi}, and $ z_{0} $ represents each of the non-zero simple roots of the generalized Wronskian-Hermite polynomial $ W_{N}^{[m,n+1,l]}(z) $. Additionally, when $ ({x}-x_{0})^{2}+ (y-y_{0})^{2} = \mathcal{O}(1) $, the solutions $ q_{k}^{[\mathcal{N}_{l}]}(x,t) $ possess the following approximate expressions:
		\begin{equation}\label{eq:qko}
			q_{k}^{[\mathcal{N}_{l}]}(x,t)= q^{[1]}_{k}(x-x_{0},t-t_{0})+ \mathcal{O}\left( |d_{m}|^{-1/m}\right).
		\end{equation}
		
		\item In the inner region with $ \sqrt{x^{2}+t^{2}}= \mathcal{O}(1) $, if zero is a root of the polynomial $ W_{N}^{[m,n+1,l]}(z) $, then the solution $ \mathbf{q}^{[\mathcal{N}_{l}]}(x,t)$ is approximately a lower-order rogue wave
		\begin{equation}\label{eq:qlow}
			\mathbf{q}^{[{\mathcal{\hat{N}}_{l}}]}(x,t)=\left[ {q}_{1}^{[{\mathcal{\hat{N}}}_{l}]}(x,t), {q}_{2}^{[{\mathcal{\hat{N}}}_{l}]}(x,t), \cdots, {q}_{n}^{[{\mathcal{\hat{N}}}_{l}]}(x,t) \right]^{T},
		\end{equation}
		where $ {\mathcal{\hat{N}}}_{l} = \sum_{i=1}^{n} N_{l,i}\mathbf{e}_{i} $ and $ N_{l,i} $ refers to the value of $ N_{i} $ against $ l\in\{1,2, \ldots,n \} $ given in Theorem \ref{th:whp}. Moreover, the internal parameters are given by
		\begin{equation}\label{eq:para}
			\begin{aligned}
				\hat{d}^{[r]}_{j}=\left\lbrace 
				\begin{aligned}
					&d_{j}+ \mu_{0} h_{2,j},  \quad j\ne m, \quad\\
					&0, \qquad j=m,
				\end{aligned}
				\right.
			\end{aligned}
		\end{equation} 
		where $  1\leq j, $ $  1\leq r\leq n $, and
		\begin{equation}\label{eq:mu0}
			\mu_{0}=\left\lbrace 
			\begin{array}{lll}
				N-\sum_{i=1}^{n}N_{l,i}, &  (n+1,m)=1, \\
				1, & (n+1,m) \ne 1, \, (n+1,m)\mid l, \\
				0, & (n+1,m) \ne 1, \, (n+1,m)\nmid l.
			\end{array}
			\right. 
		\end{equation}
		Likewise, when $ x^{2}+ t^{2} = \cO(1) $, the approximate expressions of  $ q_{k}^{[\mathcal{N}_{l}]}(x,t) $ are
		\begin{equation}\label{eq:qki}
			q_{k}^{[\mathcal{N}_{l}]}(x,t)= q^{[{\mathcal{\hat{N}}}_{l}]}_{k}(x,t)+ \mathcal{O}\left( |d_{m}|^{-1}\right).
		\end{equation}

	\end{enumerate}
\end{theorem}
\begin{proof}
	See Appendix \ref{app11}.
\end{proof}

The significance of Theorem \ref{th4} becomes apparent when considering sufficiently large internal parameter $ d _{m}$. In this context, it reveals that the positions of each peak in the outer region of the patterns of $l$-type higher-order rogue wave solutions for $ n $-DNLS equations \eqref{eq:nDNLS} are closely tied to the non-zero roots of their corresponding generalized Wronskian-Hermite polynomials $ W_{N}^{[m,n+1,l]}(z) $. It is important to note that due to the presence of the term $ \Delta^{[s]}(z_{0}) $ within the approximated expression \eqref{eq:del1}, a constant nonlinear error emerges between these two aspects. Moreover, each peak of the outer region tends to converge toward a 1-order rogue wave, accompanied by an error term of $ \cO(d_{m}^{-1/m}) $. Similarly, in the inner region, the convergence of higher-order rogue wave solutions gives way to lower-order rogue waves, with the error element proportional to $ \cO(d_{m}^{-1}) $. 
Therefore, we have the asymptotic expression as $ d_{m}\gg 1 $, shown in the following theorem.
\begin{theorem}\label{th5}
	When $ d_{m} $ is large enough, the $l$-type higher-order vector rogue wave solutions $ \mathbf{q}^{[\mathcal{N}_{l}]}(x,t) $ for $ n $-DNLS equations \eqref{eq:nDNLS} exist the following asymptotic expression 
	\begin{equation}\label{eq44a}
		\begin{aligned}
			q_{k}^{[\mathcal{N}_{l}]}(x,t)=
			\sum_{i=1}^{\Gamma-\Gamma_{0}}\left( q_{k}^{[1]}(x-x_{0}^{(i)},t-t_{0}^{(i)}) \right) + q_{k}^{[\mathcal{\hat{N}}_{l}]}+ \mathcal{O}(|d_{m}|^{-1}),
		\end{aligned}
	\end{equation}
	where $ q_{k}^{[\mathcal{\hat{N}}_{l}]} $, $ \Gamma_{0},  \Gamma $ and $ (x_{0}^{(i)}, y_{0}^{(i)})$ are defined in Theorem \ref{th4}.
\end{theorem}
In conclusion, Theorems \ref{th4} and \ref{th5} show that all the errors in the asymptotics derived from the Theorem \ref{th4} diminish progressively as the parameter $ d_{m}$ increases. Notably, as $ d_{m}$ approaches infinity, these errors become inconsequential and can be safely disregarded.

\subsection{Examples}

In this subsection, we plan to provide several practical examples along with corresponding dynamic figures regarding the $ l $-type $ \mathcal{N}_{l} $-order vector rogue wave solutions $ \mathbf{q}^{[\mathcal{N}_{l}]} =[{q}_{1}^{[\mathcal{N}_{l}]}, {q}_{2}^{[\mathcal{N}_{l}]}, \ldots, {q}_{n}^{[\mathcal{N}_{l}]}] $ of $ n $-DNLS equation \eqref{eq:nDNLS}. These will further illustrate and demonstrate our analysis of higher-order rogue waves for the $ n $-DNLS equation \eqref{eq:nDNLS} outlined in Theorem \ref{th4}. For computational convenience, we assume $ \xi=i $ without loss of generality. Next, we will provide a comprehensive discussion of higher-order vector rogue wave patterns for $ n=2 $ and $ n=3 $, namely, the $ 2 $-DNLS equation and the $ 3 $-DNLS equation.

For the $ 2 $-DNLS equation, Proposition \ref{pr3} informs us that when the parameters
\begin{equation}\label{parab1}
	\begin{aligned}
		a_{1}=\frac{4\sqrt{6}\,3^{1/4}}{9}, \quad a_{2}=\frac{2\,3^{3/4}}{9},\quad
		b_{2}=\frac{2\sqrt{3}}{3}, \quad  b_{1}=\frac{4\sqrt{3}}{3},
	\end{aligned}
\end{equation}
the characteristic polynomial \eqref{eq:cp1} exhibits a $ 3 $-multiple root $ \chi_{0}=-\frac{2\sqrt{3}}{3}+\frac{2\ii}{3} $. Additionally, for the $ 3 $-DNLS equation, we calculate the parameters
\begin{equation}\label{parab2}
	\begin{aligned}
		&a_{1}=\frac{2^{1/4}+2^{3/4}}{2}, \quad a_{2}=-\frac{{2}^{3/4}\sqrt {2+\sqrt {2}} \left( -2+\sqrt {2} \right)}{4}, \\ &a_{3}=-\frac{\sqrt {2}\sqrt {2+\sqrt {2}} \left( -2+\sqrt {2} \right)}{4},\\
		&b_{1}=1+\sqrt{2},\quad  b_{2}=\sqrt{2},\quad b_{3}=1,
	\end{aligned}
\end{equation}
and the $ 4 $-multiple root $ \chi_{0}=-\frac{1+\sqrt{2}}{2}+\frac{\ii}{2} $ of the polynomial \eqref{eq:cp1}. Subsequently, by employing the vector rogue wave solutions formulas \eqref{eq:qnrw} and \eqref{eq:qnmk2} in Theorem \ref{th3}, and varying the parameters $ l $ and $ d_{m} $, we can generate the $ l $-type higher-order rogue wave solutions for the $ 2 $-DNLS equation and $ 3 $-DNLS equation.

\begin{enumerate}[1.]
	\item \textbf{Case of the $ 2 $-DNLS equation}
	
	For the $ 2 $-DNLS equation, we consider some cases of $ 1 $-type $ (3,0) $-order and $ 2 $-type $ (0,2) $-order vector rogue wave solutions, where the large parameter is one of
	\begin{equation}\label{para1}
		d_{m}=10^{m},\quad m=2, 4, 5, 7,
	\end{equation}
	and other internal parameters $ d_{j}=0\, (j\ne m) $. Then, we present their structural characteristics in Tables \ref{ta1} and \ref{ta2} and plot them in Figs. \ref{Fig1} and \ref{Fig2}.
	By reviewing the corresponding results in Theorems \ref{th:whp} and \ref{th4}, it is evident that the patterns of $ 1 $-type and $ 2 $-type higher-order rogue waves are closely related to the roots structure of the polynomials $ W_{N}^{[m,3,1]}(z) $ and $  W_{N}^{[m,3,2]}(z)  $, respectively.
	\begin{center}
	\begin{table}
		\renewcommand\arraystretch{1.4}
		\caption{$ (3,0) $-order $ 1 $-type vector rogue waves of the $ 2 $-DNLS equation}
		\label{ta1}
		\begin{tabular}{|p{0.3\linewidth}|>{\centering\arraybackslash} p{0.13 \linewidth}|>{\centering\arraybackslash}p{0.16\linewidth}|>{\centering\arraybackslash}p{0.14\linewidth}|>{\centering\arraybackslash}p{0.14\linewidth}|}
			\hline	
			\centering {m} & 2 & 4 & 5 & 7 \\ 
			\hline
			{Number of the 1st-order rogue wave in the outer region} 
			&  8 &  8 &  5 &  7 \\
			\hline
			Order of the lower-order rogue wave in the inner region 
			& (1,0) & (1,0) & (2,1) & (0,1) \\
			\hline
			\centering Overall shape
			& triangle &  quadrangle & pentagon & heptagon\\
			\hline
		\end{tabular}
	\end{table}
	\end{center}

	\begin{center}
	\begin{table}
		\renewcommand\arraystretch{1.4}
		\caption{$ (0,2) $-order $ 2 $-type vector rogue waves of the $ 2 $-DNLS equation}
		\label{ta2}
		\begin{tabular}{|p{0.32\linewidth}|>{\centering\arraybackslash} p{0.14\linewidth}|>{\centering\arraybackslash}p{0.2\linewidth}|>{\centering\arraybackslash}p{0.15\linewidth}|}
			\hline	
			\centering m & 2 & 4 & 5  \\ 
			\hline
			Number of the 1st-order rogue wave in the outer region 
			& 6 &  4 &  5  \\
			\hline
			Order of the lower-order rogue wave in the inner region 
			& (0,0) & (1,1) & (1,0)  \\
			\hline
			\centering Overall shape
			& double-triangles &  quadrangle & pentagon \\
			\hline
		\end{tabular}
	\end{table}
	\end{center}

	\item  \textbf{Cases of the $ 3 $-DNLS equation} 
	
	For the $ 3 $-DNLS equation, we consider some cases of $ 1 $-type $ (4,0,0) $-order, $ 2 $-type $ (0,3,0) $-order, and $ 3 $-type $ (0,2) $-order vector rogue wave solutions, where the large parameter is one of
	\begin{equation}\label{para3}
		d_{m}=10^{m},\quad m=2,3,5,6,7,
	\end{equation}
	and other internal parameters $ d_{j}=0\, (j\ne m) $. Then, we present their structural characteristics, which can be found in Tables \ref{ta3}-\ref{ta5}, and these features are graphically represented in Figs. \ref{Fig3}-\ref{Fig5}. Upon reviewing Theorems \ref{th:whp} and \ref{th4}, it becomes apparent that the patterns of $ l $-type $ (1\leq l \leq 3) $ higher-order rogue waves are connected to the root structures of the polynomials $ W_{N}^{[m,4,l]}(z) $. 
	
	\begin{table}
		\centering
		\renewcommand\arraystretch{1.45}
		\caption{$ (4,0,0) $-order $ 1 $-type vector rogue waves of $ 3 $-DNLS equation}
		\label{ta3}
		\begin{tabular}{|p{0.25\linewidth}|>{\centering\arraybackslash} p{0.13 \linewidth}|>{\centering\arraybackslash}p{0.11\linewidth}|>{\centering\arraybackslash}p{0.14\linewidth}|>{\centering\arraybackslash}p{0.12\linewidth} |>{\centering\arraybackslash}p{0.13\linewidth}|}
			\hline	
			\centering {m} & 2 & 3 & 5 & 6 & 7 \\ 
			\hline
			{Number of the 1st-order rogue wave in the outer region} 
			&  12 &  21 &  20 &  12 & 14 \\
			\hline
			Order of the lower-order rogue wave in the inner region 
			& (2,0,2) & (1,0,0) & (1,1,0) & (2,0,2) & (2,2,0) \\
			\hline
			\centering Overall shape
			& double-triangles & triangle &  pentagon & hexagon & heptagon\\
			\hline
		\end{tabular}
	\end{table}
	
	\begin{table}
		\centering
		\renewcommand\arraystretch{1.4}
		\caption{$ (0,3,0) $-order $ 2 $-type vector rogue waves of $ 3 $-DNLS equation}
		\label{ta4}
		\begin{tabular}{|p{0.25\linewidth}|>{\centering\arraybackslash} p{0.13 \linewidth}|>{\centering\arraybackslash}p{0.11\linewidth}|>{\centering\arraybackslash}p{0.14\linewidth}|>{\centering\arraybackslash}p{0.12\linewidth} |>{\centering\arraybackslash}p{0.13\linewidth}|}
			\hline	
			\centering {m} & 2 & 3 & 5 & 6 & 7 \\ 
			\hline
			{Number of the 1st-order rogue wave in the outer region} 
			&  12 &  15 &  15 &  12 & 7 \\
			\hline
			Order of the lower-order rogue wave in the inner region 
			& (1,0,1) & (0,0,0) & (0,0,0) & (2,0,0) & (0,2,1) \\
			\hline
			\centering Overall shape
			& double-triangles & triangle &  pentagon & hexagon & heptagon\\
			\hline
		\end{tabular}
	\end{table}
	
	\begin{table}
		\centering
		\renewcommand\arraystretch{1.42}
		\caption{$ (0,0,2) $-order $ 3 $-type vector rogue waves of $ 3 $-DNLS equation}
		\label{ta5}
		\begin{tabular}{|p{0.26\linewidth}|>{\centering\arraybackslash} p{0.12 \linewidth}|>{\centering\arraybackslash}p{0.12\linewidth}|>{\centering\arraybackslash}p{0.14\linewidth}|>{\centering\arraybackslash}p{0.12\linewidth} |>{\centering\arraybackslash}p{0.13\linewidth}|}
			\hline	
			\centering {m} & 2 & 3 & 5 & 6 & 7 \\ 
			\hline
			{Number of the 1st-order rogue wave in the outer region} 
			&  6 &  9 &  5 &  6 & 7 \\
			\hline
			Order of the lower-order rogue wave in the inner region 
			& (1,0,1) & (0,0,0) & (0,1,1) & (1,0,1) & (0,1,0) \\
			\hline
			\centering Overall shape
			& double-triangles & triangle &  pentagon & hexagon & heptagon\\
			\hline
		\end{tabular}
	\end{table}

\end{enumerate}

\section{Conclusions and Discussions}\label{sec5}

In this paper, we establish the DT theory for the $ n $-DNLS equation \eqref{eq:nDNLS}, enabling us to construct determinant expressions for high-order rogue wave solutions. Subsequently, we conduct a comprehensive analysis of the asymptotic behavior and structural patterns of these higher-order rogue waves, each characterized by distinct internal large parameters $ d_{m} $. To facilitate this analysis, we also demonstrate the structural properties of generalized Wronskian-Hermite polynomials with arbitrary jump parameter $ k $. Through this investigation, we discover a profound connection between the structural patterns of the $ l $-type higher-order rogue waves with sufficiently large parameter $ d_{m} $ and the root structures of these polynomials $ W_{N}^{[m,n+1,l]}(z) $. We categorize each structural model of high-order rogue waves into two regions: an outer region, which closely resembles the locations of non-zero roots of the polynomials $ W_{N}^{[m,n+1,l]}(z) $ with a relative error, and an inner region. In the inner region, either a lower-order rogue wave exists when the polynomial possesses zero roots, or the solutions tend toward a plane wave background when there are no zero roots. Finally, we provide examples demonstrating rogue wave patterns for the two-component and three-component DNLS equations.

It is worth noting that our analysis in this paper exclusively considers cases where the characteristic polynomial has maximum multiple roots. For more general scenarios, we plan to explore in future research. The findings presented in this paper have significant implications and positive impacts on the study of multi-component integrable systems and rogue waves.

\begin{figure*}[!htbp]
	\centering
	\includegraphics[width=\textwidth]{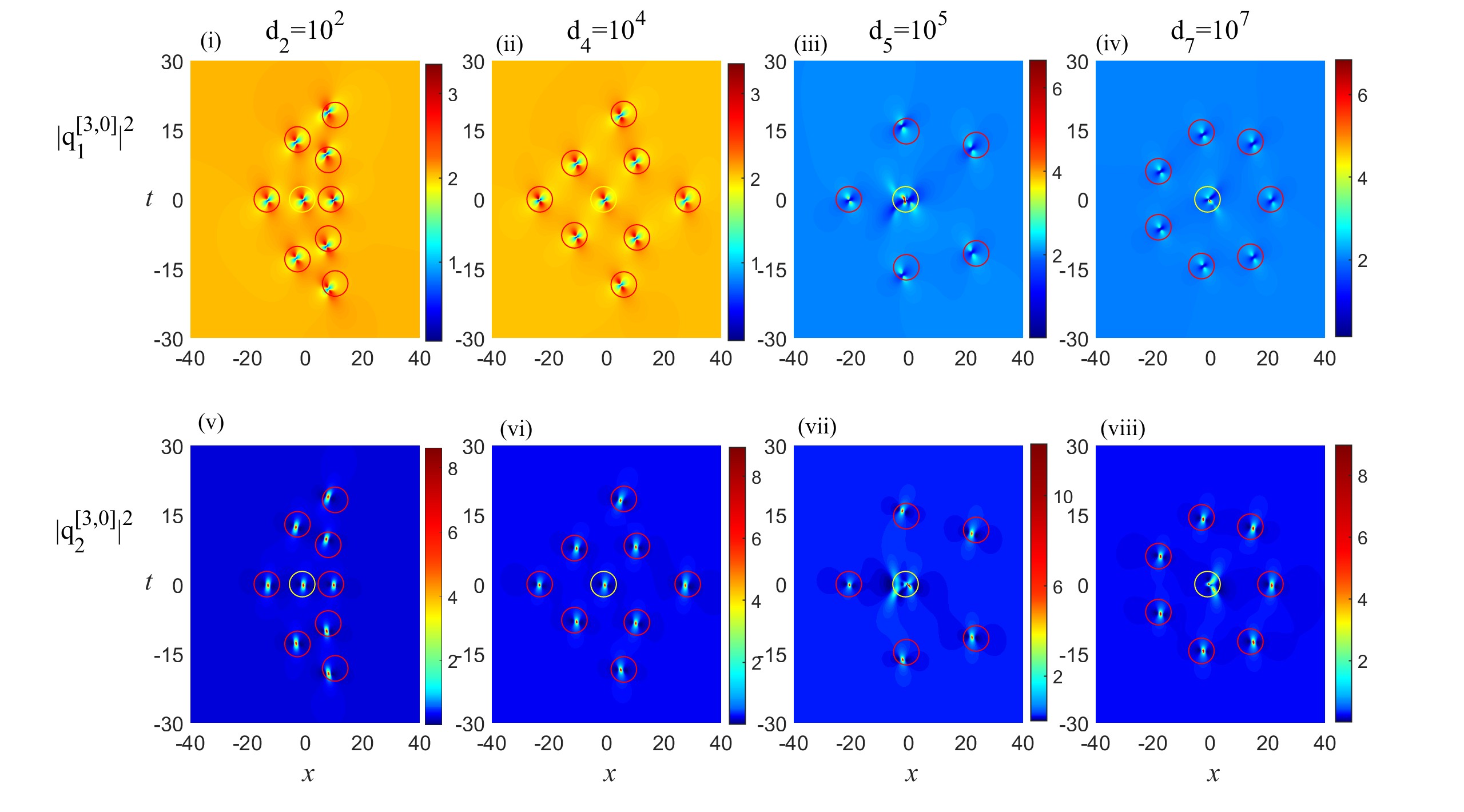}
	\caption{The $ (3,0) $-order $ 1 $-type vector rogue wave solutions $ \mathbf{q}^{[3,0]} $ of the $ 2 $-DNLS equation with $ N=3 $, $ l=1 $, $ \xi^{[1]}=2\ii $, and $ d_{i}=10^{i} (i=2,4,5,7)$. (i-iv) The first component $ |{q}_{1}^{[3,0]}|^{2} $ of $ 1 $-type rogue waves. (v-viii)  The second component $ |{q}_{2}^{[3,0]}|^{2} $ of $ 1 $-type rogue waves. In these figures, from the first column to the fourth column, their waveform structures are triangle, quadrilateral, pentagon, and heptagon, respectively. Additionally, the numbers of the first-order rogue wave in their outer regions are $ 8, 8, 5 $, and $ 7 $, individually. The lower-order rogue wave in the inner regions are $ (1,0) $-order, $ (1,0) $-order, $ (2,1) $-order, and $ (0,1) $-order, separately. These red circles represent the predicted positions of each 1st-order rogue wave in the outer region, and these yellow circles represent the predicted positions of the lower-order rogue wave in the inner region. }
	\label{Fig1}
\end{figure*}	
\begin{figure*}[!htbp]
	\centering
	\includegraphics[width=\textwidth]{{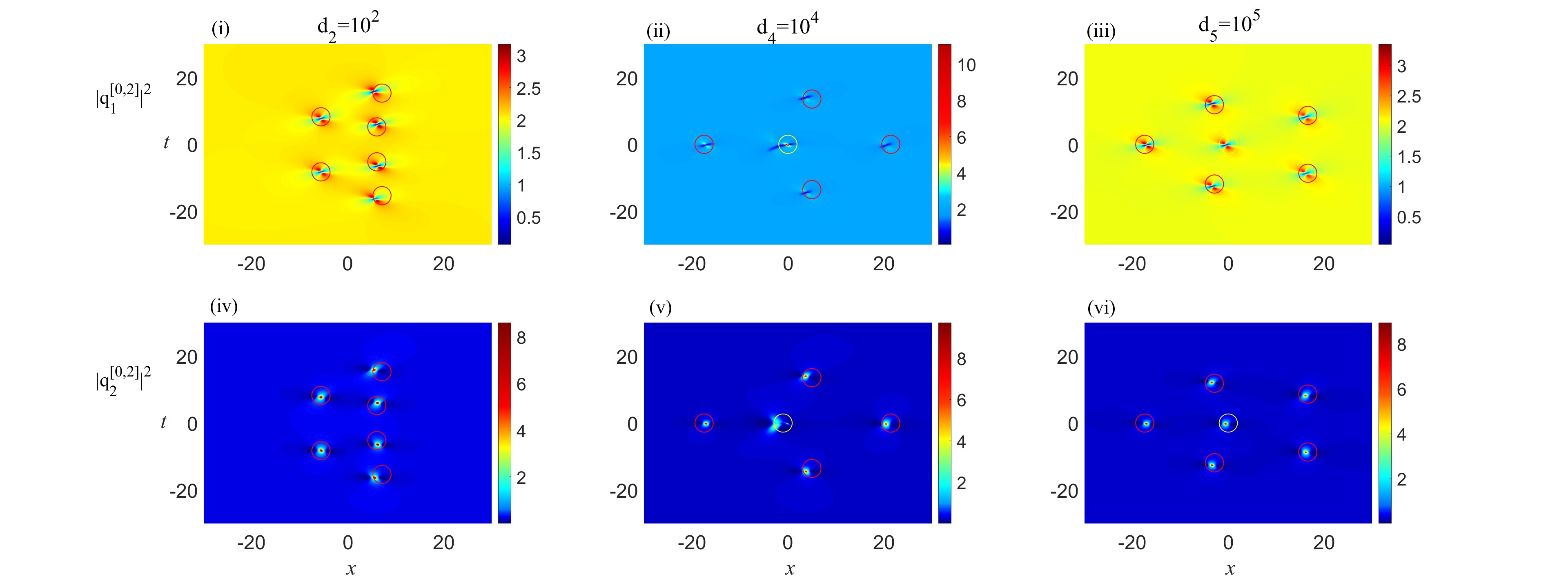}}
	\caption{The $ (0,2) $-order $ 2 $-type vector rogue wave solutions $ \mathbf{q}^{[0,2]} $ of the $ 2 $-DNLS equation with $ N=2 $, $ l=2 $, $ \xi^{[1]}=2\ii $, and $ d_{i}=10^{i} (i=2,4,5)$. (i-iv) The first component $ |{q}_{1}^{[0,2]}|^{2} $ of $ 2 $-type rogue waves. (v-viii)  The second component $ |{q}_{2}^{[0,2]}|^{2} $ of $ 2 $-type rogue waves. In these figures, from the first column to the third column, their waveform structures are double-triangles, quadrilateral, and pentagon, respectively. Additionally, the numbers of the first-order rogue wave in their outer regions are $ 6,4 $, and $ 5 $, individually. The lower-order rogue wave in the inner regions are $ (0,0) $-order, $ (1,1) $-order, and $ (1,0) $-order, separately. These red circles represent the predicted positions of each 1st-order rogue wave in the outer region, and these yellow circles represent the predicted positions of the lower-order rogue wave in the inner region. }
	\label{Fig2}
\end{figure*}

\appendix
\renewcommand\thesection{A}
\section*{Appendix A}

\subsection{\label{app4} The proof of Theorem \ref{th:whp}}

\begin{proof}
	For Theorem \ref{th:whp}, we divide its proof into three parts: the degree \eqref{eq:gamma}, the zero root multiplicity \eqref{eq:gamma0}, and the expression \eqref{eq:whp} of the polynomials $ W_{N}^{[m,k,l]}(z) $.

\noindent \textbf{1.1. The degree $ \Gamma $ of the polynomials $ W_{N}^{[m,k,l]}(z) $}

	First, we plan to calculate the degree $ \Gamma $ of the polynomials $ W_{N}^{[m,k,l]}(z) $, where $ m, k $, and $ l $ are defined by Theorem \ref{th:whp}. It is evident that $ p_{j}^{[m]}(z) $ are the $ j $-order polynomials to $ z $, and the subscript $ j $ of $ p_{j}^{[m]}(z) $ decreases by $ 1 $ for each corresponding column element in the determinant of $ {W}_{N}^{[m,k,l]}(z) $ when the row index increases by $ 1 $. Then, we can calculate the sum $ \frac{(k(N-1)+2l)N}{2} $ of the highest-order powers with respect to $ z $ among the elements in the first row of the determinant in $ {W}_{N}^{[m,k,l]}(z) $ \eqref{eq:p2}.
	Further, applying the properties of the determinant expansion, we subtract any redundant orders $ \frac{N(N-1)}{2} $ to obtain the degree $ \Gamma $ of the polynomials, as shown in Eq. \eqref{eq:gamma}.

\noindent \textbf{1.2. The zero root multiplicity $ \Gamma_{0} $ of the polynomials $ W_{N}^{[m,k,l]}(z) $}

Next, we will discuss the multiplicity $ \Gamma_{0} $ of the zero root in $ W_{N}^{[m,k,l]}(z) $, namely, the lowest-order power of the variable $ z $. Here, we introduce a new particular Schur polynomials $ S_{j}^{[m]}(z;a) $ and the polynomials $ \bar{W}_{N}^{[m,k,l]}(z; a) $, which are defined by
\begin{subequations}
	\begin{align}
		&\sum_{j=0}^{\infty}S^{[m]}_{j}(z;a)\varepsilon^{j}=\exp (z\varepsilon+a\varepsilon^{m}), \label{eq:scp2a}\\
		&\bar{W}_{N}^{[m,k,l]}(z;a)={c}_{N}^{[m,k,l]}\det_{1\leq i,j\leq N}\left( S^{[m]}_{k(j-1)+l-i+1}(z;a)\right), \label{eq:scp2b}
	\end{align}
\end{subequations}
where $ a $ is a constant, $ S^{[m]}_{j}(z;a)=0 $ if $ j<0 $,  $ m,k,l $, and $ {c}_{N}^{[m,k,l]} $ are identical with the definition in Eq. \eqref{eq:p2}. 
	
According to the definitions \eqref{eq:p1} and \eqref{eq:scp2a} of $  p^{[m]}_{j}(z) $ and $ S^{[m]}_{j}(z;a) $, we yield
\begin{subequations}
	\begin{align}
		&S^{[m]}_{j}(z;a)=a^{j/m}p^{[m]}_{j}(\bar{z}),	 \quad \bar{z}=a^{-1/m}z, \label{eq:spr} \\ 
		&\bar{W}_{N}^{[m,k,l]}(z;a) = a^{\Gamma/m}{W}_{N}^{[m,k,l]}(\bar{z}).	 \label{eq:spw}
	\end{align}
\end{subequations}

Suppose that each term of $ \bar{W}_{N}^{[m,k,l]}(z;a) $ is $ a^{i}z^{j} $, then we group the terms according to the power of $ z $ in Eq. \eqref{eq:spw}. To the term $\cO(z^{j})$, we have
\begin{equation}\label{eq01}
	im+j=\Gamma,
\end{equation}
which implies that the power $ j $ of $ z $ in each term of $ \bar{W}_{N}^{[m,k,l]}(z;a) $ is lower when the power of $ a $ is larger. Since the multiplicity of the zero root in $ W_{N}^{[m,k,l]}(z)$ equals that of $ \bar{W}_{N}^{[m,k,l]}(z;a) $, we can consider the highest-power term of $ a $ in $ \bar{W}_{N}^{[m,k,l]}(z;a) $ to obtain the multiplicity of the zero root $ \Gamma_{0} $.
	
We expand $ S^{[m]}_{j}(z;a) $ into the powers of $ a $ as 
\begin{equation}\label{eq:scpa}
	S^{[m]}_{j}(z;a)=\sum_{i=0}^{\lfloor j/m\rfloor} \frac{a^{i}}{i!(j-im)!}z^{j-im},
\end{equation}
where $ \lfloor a \rfloor $ represents the largest integer less than or equal to $ a $. Then, by using this expansion \eqref{eq:scpa}, we expand all elements of the determinant in $ \bar{W}_{N}^{[m,k,l]}(z;a) $ \eqref{eq:scp2b}, as follows:
\begin{widetext}
	\begin{equation}\label{scdex}
		\bar{W}_{N}^{[m,k,l]}(z; a)={c}_{N}^{[m,k,l]}\det_{1\leq i,j\leq N}\left( \sum_{s=0}^{\lfloor \frac{k(j-1)+l-i+1}{m}\rfloor}\frac{a^{s}}{s!(k(j-1)+l-i+1-sm)!} \, z^{k(j-1)+l-i+1-sm}\right),
	\end{equation}
\end{widetext}
where the exponents of $ a $ and $ z $ in the determinant \eqref{scdex} are all non-negative integers, and this property holds similarly in the subsequent proof. 
	
Since the subscript $ j $ of $ S_{j}^{[m]}(z;a) $ decreases by $ 1 $ for each corresponding column element in $ \bar{W}_{N}^{[m,k,l]}(z) $ \eqref{eq:scp2b} when the row index increases by $ 1 $,  we focus on the elements in the first row of determinant \eqref{scdex} for the highest-order term of $ a $ of $ \bar{W}_{N}^{[m,k,l]}(z; a) $.
	
	It becomes evident that for the element of the $ j $-th column of the first row in determinant \eqref{scdex}, the exponent of $ z $ in the highest-order term of $ a $ corresponds to the element $ \beta_{j} $ of the following set $ \mathcal{B} $:
		\begin{equation}\label{setb}
		\qquad	\mathcal{B} = \left\lbrace \beta_{j} \mid  \beta_{j}=(l+(j-1)k){\,\mathrm{mod}\,} m, \, 1\leq j\leq N \right\rbrace,
		\end{equation}
	where the symbol $ (a \, \mathrm{mod} \, b) $ denotes the remainder of $ a $ divided by $ b $. To ascertain the power of $ z $ in the highest-order term of $ a $ in the first row of determinant \eqref{scdex}, we need to consider the properties of the above set $ \mathcal{B} $. 
	
	We firstly prove that the elements of set $ \mathcal{B} $ have at most ${\hat{m}} $ distinct elements with
	\begin{equation}\label{hatm}
		\hat{m}=\frac{m}{(m,k)}.
	\end{equation}
	For this purpose, we construct a mapping $ f: \beta_{j}\longmapsto \hat{\beta}_{j} $ between the set $ \mathcal{B} $ and the abelian group $ \hat{\mathcal{B}} = \left\lbrace \hat{\beta}_{j} \mid  \hat{\beta}_{j}=((j-1)k) \,\mathrm{mod}\, m, \, j\in \mathbb{N} \right\rbrace $.
	Since  the order of $ \hat{\mathcal{B}} $ is $ {\hat{m} } $, we can derive that when $ N={\hat{m} }$, the mapping $ f $ is a one-to-one correspondence. Therefore, all elements in the set $ \mathcal{B} $ are distinct with $ N={\hat{m} } $ because all elements of $ \hat{\mathcal{B}} $ are different. 
	On the other hand,  due to $ \beta_{j}=\beta_{j+{\hat{m} }} $ for the arbitrary integer $ j $, we find that the powers of $ z $ in the highest-order terms of $ a $ in the first row of determinant \eqref{scdex} are period, and one period is
	\begin{equation}\label{per1}
		z^{\beta_{1}}, z^{\beta_{2}}, \ldots, z^{\beta_{{\hat{m} }}}, 
	\end{equation}
	where $ \beta_{j} $ $ (1\leq j\leq {\hat{m} }) $ are given in Eq. \eqref{setb}. Here, we omit the constant coefficients of each term and perform this operation in the subsequent proof process for convenience.
	
	Additionally, we assume $ N_{0} $ and $ m_{0} $ are the remainders of $ N $ divided by $ {\hat{m} } $ and $ m $ divided by $ k $, respectively, i.e.,
	\begin{equation}\label{eq:mn}
		N={k_{N} }{\hat{m} }+N_{0}, \quad m=m_{1}k+m_{0}, 
	\end{equation}
	where $ \hat{m} $ is given in Eq. \eqref{hatm}.
	
	According to the period \eqref{per1}, we apply the column transformations for the determinant \eqref{scdex} to eliminate the preceding $ j-1 $ higher-order terms of $ a $ of the first row elements in the $ j $-th $ (1< j\leq k_{N}) $ set of ${\hat{m} } $ columns. Then, the polynomials $ \bar{W}_{N}^{[m,k,l]}(z; a) $ are approximatively reduced as
\begin{widetext}
	\begin{equation}\label{scdex2}
	 \begin{vmatrix}
			a^{s_{1}}z^{\beta_{1}} + \cO(a^{s_{1}-1}) & a^{s_{1}}z^{\beta_{1}-1} + \cO(a^{s_{1}-1}) & \cdots \\
			a^{s_{2}}z^{\beta_{2}} + \cO(a^{s_{2}-1}) & a^{s_{2}}z^{\beta_{2}-1} + \cO(a^{s_{2}-1}) & \cdots \\
			\vdots & \vdots & \ddots\\
			a^{s_{{\hat{m} }}}z^{\beta_{{\hat{m} }}} + \cO(a^{s_{{\hat{m} }}-1}) & a^{s_{{\hat{m} }}}z^{\beta_{{\hat{m} }}-1} + \cO(a^{s_{{\hat{m} }}-1}) & \cdots \\
			a^{s_{{\hat{m} }+1}-1}z^{\beta_{1}+m} + \cO(a^{s_{{\hat{m} }+1}-2}) & a^{s_{{\hat{m} }+1}-1}z^{\beta_{1}+m-1} + \cO(a^{s_{{\hat{m} }+1}-2}) & \cdots \\
			a^{s_{{\hat{m} }+2}-1}z^{\beta_{2}+m} + \cO(a^{s_{{\hat{m} }+2}-2}) & a^{s_{{\hat{m} }+2}-1}z^{\beta_{2}+m-1} + \cO(a^{s_{{\hat{m} }+2}-2}) & \cdots \\
			\vdots & \vdots & \ddots\\
			a^{s_{2{\hat{m} }}-1}z^{\beta_{{\hat{m} }}+m} + \cO(a^{s_{2{\hat{m} }}-2}) & a^{s_{2{\hat{m} }}-1}z^{\beta_{{\hat{m} }}+m-1} + \cO(a^{s_{2{\hat{m} }}-2}) & \cdots \\
			\vdots & \vdots & \ddots\\
			a^{s_{(k_{N}-1)\hat{m} +1}-(k_{N}-1)}z^{\beta_{1}+(k_{N}-1)m} + \cO(a^{s_{(k_{N}-1)\hat{m}+1}-k_{N}}) & a^{s_{(k_{N}-1)\hat{m}+1}-(k_{N}-1)}z^{\beta_{1}+(k_{N}-1)m-1} + \cO(a^{s_{(k_{N}-1)\hat{m}+1}-k_{N}}) & \cdots \\
			a^{s_{(k_{N}-1)\hat{m}+2}-(k_{N}-1)}z^{\beta_{2}+(k_{N}-1)m} + \cO(a^{s_{(k_{N}-1)\hat{m}+2}-k_{N}}) & a^{s_{(k_{N}-1)\hat{m}+2}-(k_{N}-1)}z^{\beta_{2}+(k_{N}-1)m-1} + \cO(a^{s_{(k_{N}-1)\hat{m}+2}-k_{N}}) & \cdots \\
			\vdots & \vdots & \ddots\\
			a^{s_{k_{N}\hat{m}}-(k_{N}-1)}z^{\beta_{\hat{m}}+(k_{N}-1)m} + \cO(a^{s_{k_{N}\hat{m}}-k_{N}}) & a^{s_{k_{N}\hat{m}}-(k_{N}-1)}z^{\beta_{\hat{m}} +(k_{N}-1)m-1} + \cO(a^{s_{k_{N}\hat{m}}-k_{N}}) & \cdots \\
			a^{s_{k_{N}\hat{m}+1}-k_{N}}z^{\beta_{1}+k_{N}m} + \cO(a^{s_{k_{N}\hat{m}+1}-k_{N}-1}) & a^{s_{k_{N}\hat{m}+1}-k_{N}}z^{\beta_{1}+k_{N}m-1} + \cO(a^{s_{k_{N}\hat{m}+1}-k_{N}-1}) & \cdots \\
			a^{s_{k_{N}\hat{m}+2}-k_{N}}z^{\beta_{2}+k_{N}m} + \cO(a^{s_{k_{N}\hat{m}+2}-k_{N}-1}) & a^{s_{k_{N}\hat{m}+2}-k_{N}}z^{\beta_{2}+k_{N}m-1} + \cO(a^{s_{k_{N}\hat{m}+2}-k_{N}-1}) & \cdots \\
			\vdots & \vdots & \ddots\\
			a^{s_{k_{N}\hat{m}+N_{0}}-k_{N}}z^{\beta_{N_{0}}+k_{N}m} + \cO(a^{s_{k_{N}\hat{m}}-k_{N}-1}) & a^{s_{k_{N}\hat{m}+N_{0}}-k_{N}}z^{\beta_{N_{0}} +k_{N}m-1} + \cO(a^{s_{k_{N}\hat{m}+N_{0}}-k_{N}-1}) & \cdots \\
		\end{vmatrix}_{N\times N}^{T},
	\end{equation}
\end{widetext}
where the constant coefficients are omitted, and $ s_{i} $ $ (1\leq i\leq N) $ are the highest-order powers of $ a $ in the elements of the first row of determinant \eqref{scdex}. Thus, the power of $ z $ in the highest-order term of $ a $ in $ \bar{W}_{N}^{[m,k,l]}(z; a) $ can be expressed as
\begin{widetext}
	\begin{equation}\label{bldet0}
		\left| \begin{array}{*{13}{c}}
			z^{\beta_{1}} & \cdots & z^{\beta_{\hat{m}}} & z^{\beta_{1}+m} & \cdots & z^{\beta_{\hat{m}}+m} & \cdots & z^{\beta_{1}+(k_{N}-1)m} & \cdots &  z^{\beta_{\hat{m}} +(k_{N}-1)m} & z^{\beta_{1}+k_{N}m} & \cdots & z^{\beta_{N_{0}}+k_{N}m} \\
			z^{\beta_{1}-1} & \cdots & z^{\beta_{\hat{m}} -1} & z^{\beta_{1}+m -1} & \cdots & z^{\beta_{\hat{m}}+m -1} & \cdots & z^{\beta_{1}+(k_{N}-1)m -1} & \cdots &  z^{\beta_{\hat{m}} +(k_{N}-1)m -1} & z^{\beta_{1}+k_{N}m -1} & \cdots & z^{\beta_{N_{0}}+k_{N}m -1} \\
			\vdots & \ddots & \vdots  & \vdots & \ddots & \vdots & \ddots & \vdots & \ddots & \vdots & \vdots & \ddots & \vdots
		\end{array}\right| _{N\times N} 
	\end{equation}
\end{widetext}
 
\begin{figure*}[!htbp]
	\centering
	\includegraphics[width=\textwidth]{{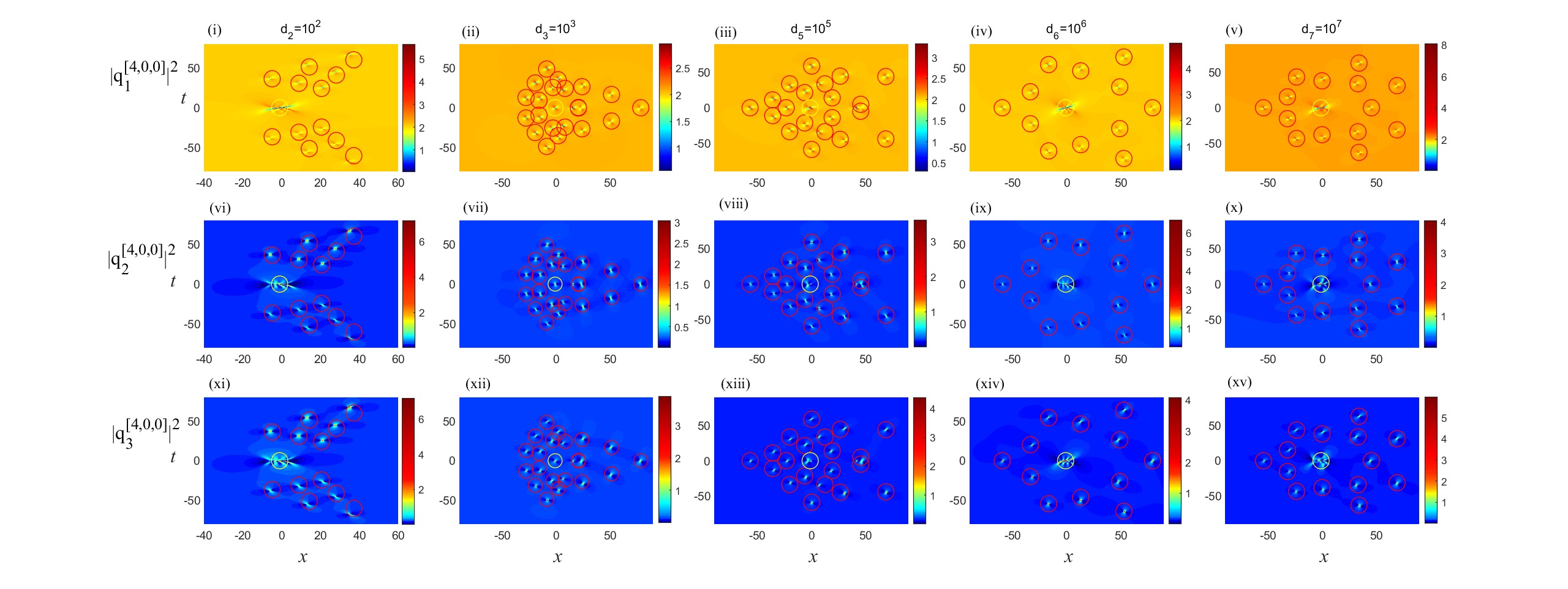}}
	\caption{The $ (4,0,0) $-order $ 1 $-type vector rogue wave solutions $ \mathbf{q}^{[4,0,0]} $ of the $ 3 $-DNLS equation with $ N=4 $, $ l=1 $, $ \xi^{[1]}=2\ii $, and $ d_{i}=10^{i} (i=2,3,5,6,7)$. (i-v) The first component $ |{q}_{1}^{[4,0,0]}|^{2} $ of $ 1 $-type rogue waves. (vi-x)  The second component $ |{q}_{2}^{[3,0,0]}|^{2} $ of $ 1 $-type rogue waves. (xi-xv)  The third component $ |{q}_{3}^{[4,0,0]}|^{2} $ of $ 1 $-type rogue waves. In these figures, from the first column to the fifth column, their waveform structures are double-triangles, triangle, pentagon, hexagon, and heptagon, respectively. Additionally, the numbers of the first-order rogue wave in their outer regions are $ 12, 21,20,12 $, and $ 14 $, individually. The lower-order rogue wave in the inner regions are $ (2,0,2) $-order, $ (1,0,0) $-order, $ (1,1,0) $-order, $ (2,0,2) $-order, and $ (2,2,0) $-order, separately. These red circles represent the predicted positions of each 1st-order rogue wave in the outer region, and these yellow circles represent the predicted positions of the lower-order rogue wave in the inner region. }
	\label{Fig3}
\end{figure*}

\begin{figure*}[!htbp]
	\centering
	\includegraphics[width=\textwidth]{{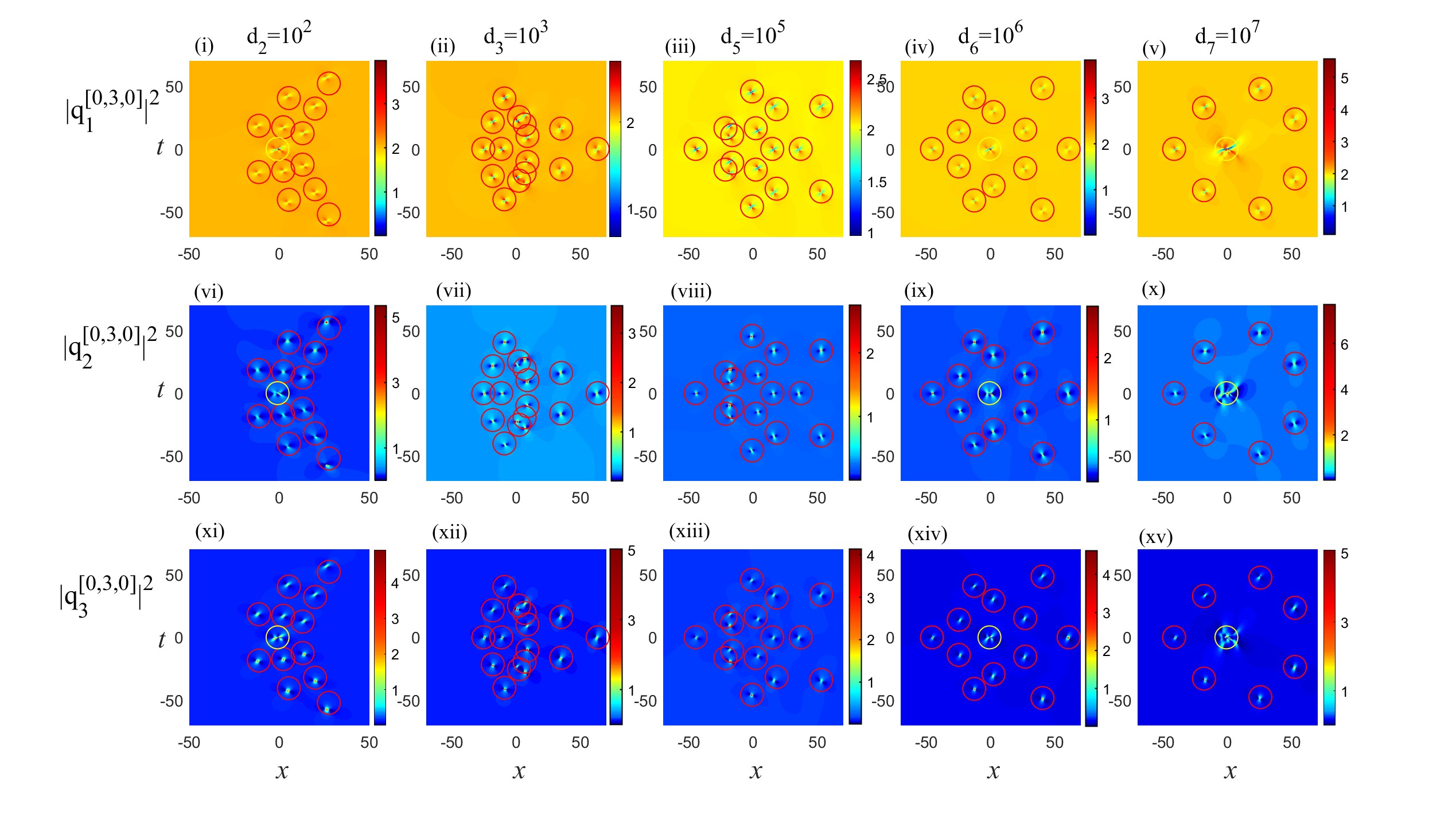}}
	\caption{The $ (0,3,0) $-order $ 2 $-type vector rogue wave solutions $ \mathbf{q}^{[0,3,0]} $ of the $ 3 $-DNLS equation with $ N=3 $, $ l=2 $, $ \xi^{[1]}=2\ii $, and $ d_{i}=10^{i} (i=2,3,5,6,7)$. (i-v) The first component $ |{q}_{1}^{[0,3,0]}|^{2} $ of $ 2 $-type rogue waves. (vi-x)  The second component $ |{q}_{2}^{[0,3,0]}|^{2} $ of $ 2 $-type rogue waves. (xi-xv)  The third component $ |{q}_{3}^{[0,3,0]}|^{2} $ of $ 2 $-type rogue waves. In these figures, from the first column to the fifth column, their waveform structures are double-triangles, triangle, pentagon, hexagon, and heptagon, respectively. Additionally, the numbers of the first-order rogue wave in their outer regions are $ 12, 15,15,12 $, and $ 7 $, individually. The lower-order rogue wave in the inner regions are $ (1,0,1) $-order, $ (0,0,0) $-order, $ (0,0,0) $-order, $ (2,0,0) $-order, and $ (0,2,1) $-order, separately. These red circles represent the predicted positions of each 1st-order rogue wave in the outer region, and these yellow circles represent the predicted positions of the lower-order rogue wave in the inner region. }
	\label{Fig4}
\end{figure*}

\begin{figure*}[!htbp]
	\centering
	\includegraphics[width=\textwidth]{{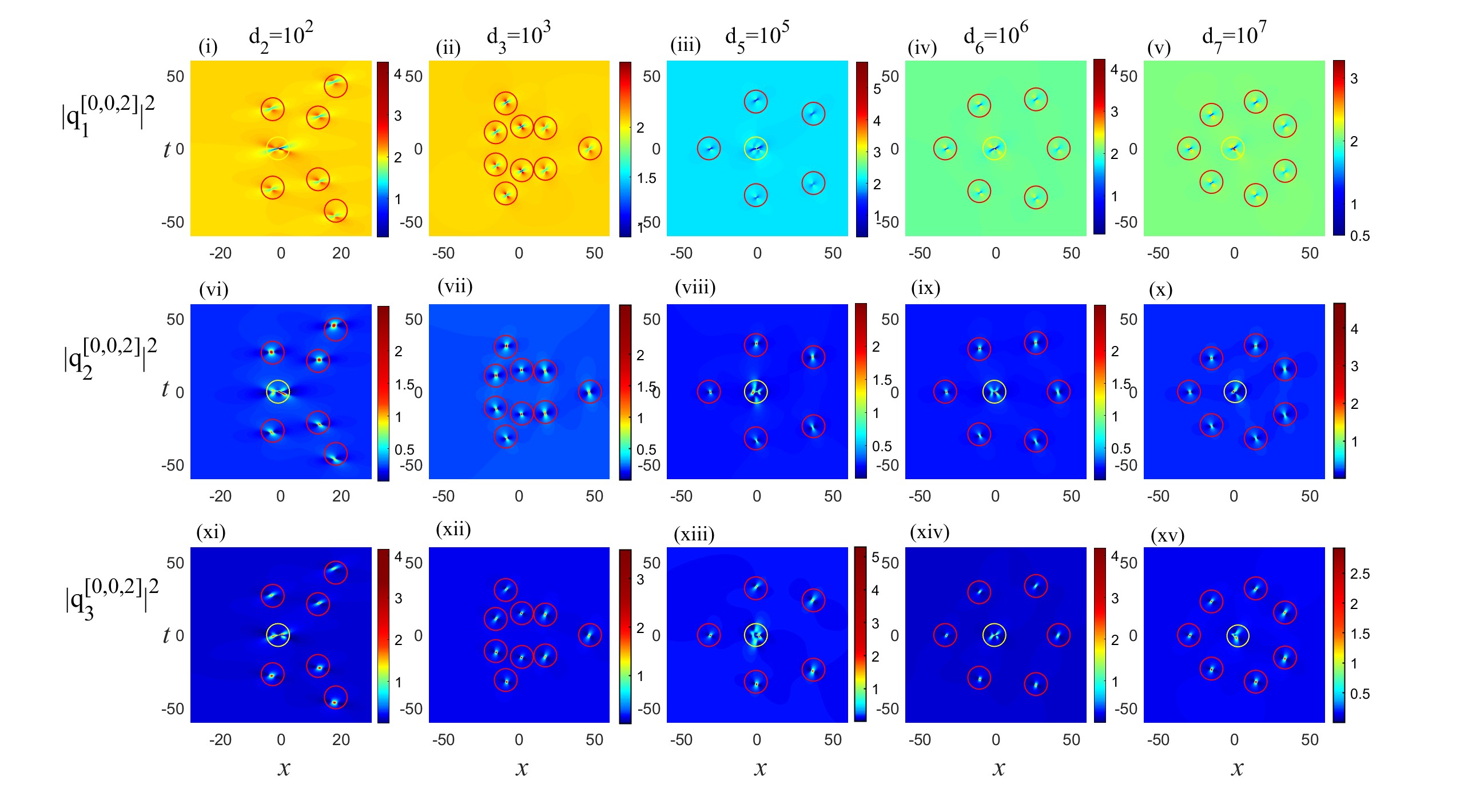}}
	\caption{The $ (0,0,2) $-order $ 3 $-type vector rogue wave solutions $ \mathbf{q}^{[0,0,2]} $ of the $ 3 $-DNLS equation with $ N=2 $, $ l=3 $, $ \xi^{[1]}=2\ii $, and $ d_{i}=10^{i} (i=2,3,5,6,7)$. (i-v) The first component $ |{q}_{1}^{[0,0,2]}|^{2} $ of $ 3 $-type rogue waves. (vi-x)  The second component $ |{q}_{2}^{[0,0,2]}|^{2} $ of $ 3 $-type rogue waves. (xi-xv)  The third component $ |{q}_{3}^{[0,0,2]}|^{2} $ of $ 3 $-type rogue waves. In these figures, from the first column to the fifth column, their waveform structures are double-triangles, triangle, pentagon, hexagon, and heptagon, respectively. Additionally, the numbers of the first-order rogue wave in their outer regions are $ 6, 9,5,6 $, and $ 7 $, individually. The lower-order rogue wave in the inner regions are $ (1,0,1) $-order, $ (0,0,0) $-order, $ (0,1,1) $-order, $ (1,0,1) $-order, and $ (0,1,0) $-order, separately. These red circles represent the predicted positions of each 1st-order rogue wave in the outer region, and these yellow circles represent the predicted positions of the lower-order rogue wave in the inner region. }
	\label{Fig5}
\end{figure*}

As $ \beta_{i} $ $ (1\leq i\leq {\hat{m}}) $ are mutually distinct, it is clear to find that all elements in the first row of determinant \eqref{bldet0} are distinct. Then, the zero root multiplicity $ \Gamma_{0} $ of $ \bar{W}_{N}^{[m,k,l]}(z; a) $ is obtained by computing the exponent of $ z $ of the above determinant \eqref{bldet0}.

To calculate the exponent of $ z $ of the above determinant \eqref{bldet0}, we firstly consider the first row of the determinant. If the first row has a constant term element $ z^{0} $, we can eliminate the row and column containing this element when calculating the determinant \eqref{bldet0}. Subsequently, we focus on the initial row in the remaining part of determinant \eqref{bldet0}. Repeat the above steps until, the initial row in the remaining part of the determinant does not contains the term $ z^{0} $. Later in the proof, we will establish that the elements of this initial row can be divided into $ k-1 $ segments, each defined by
\begin{equation}\label{seg1}
	z^{i}, z^{i+k}, z^{i+2k}, \cdots,
\end{equation}
with $ 1\leq i <k $. We represent the segments \eqref{seg1} with the symbols $ \boxed{z^{i}} $, and define the number of elements in segment \eqref{seg1} as the length of the segment, denoted as $ N_{i} $ here. Therefore, we can calculate the power $ \Gamma_{0} $ of $ z $ in determinant \eqref{bldet0} by
\begin{widetext}
\begin{equation}\label{calzp}
		\begin{aligned}
			\Gamma_{0} &= \sum_{i=1}^{k-1}\left( \frac{N_{i}}{2}(2i+(N_{i}-1)k) - \frac{N_{i}(N_{i}-1)}{2}\right) - \sum_{1\leq i< j\leq k-1}N_{i}N_{j} \\
			&=\sum_{i=1}^{k-1}\frac{N_{i}}{2}((k-1)(N_{i}-1)+2i)-\sum_{1\leq i< j\leq k-1}N_{i}N_{j}.
		\end{aligned}
\end{equation}
\end{widetext}

Next, we need to determine the values of $ {N}_{i} $ and calculate the determinant \eqref{bldet0} according to two distinct conditions, i.e., whether $ (m, k)=1 $. 

\begin{enumerate}[(1).]
	\item  \textbf{Cases of $ (m, k)=1 $}

	Since $ 0\leq \beta_{i}<m $ and $ \beta_{i}\ne \beta_{j} $ for $ i\ne j $ and $ 1\leq i, j\leq \hat{m} $, we deduce that $ \beta_{i} $ $ (1\leq i\leq \hat{m}) $ are ergodic from $ 0 $ to $ m-1 $ when $ (m,k)=1 $. This implies that the powers of $ z $ in the initial $ k_{N}m $ elements of the first row of determinant \eqref{bldet0} are ergodic from $ 0 $ to $ k_{N}m-1 $. Then, the determinant \eqref{bldet0} can be simplified as
		\begin{widetext}
		\begin{equation}\label{bldet}
			\begin{aligned}
				(-1)^{c_{1}} \begin{vmatrix}
					z^{0} & z^{1} & \cdots & z^{k_{N}m-1} & z^{\beta_{1}+k_{N}m} & z^{\beta_{2}+k_{N}m} & \cdots & z^{\beta_{N_{0}}+k_{N}m}\\
					0 & z^{0} & \cdots & z^{k_{N}m-2} & z^{\beta_{1}+k_{N}m-1} &z^{\beta_{2}+k_{N}m-1} & \cdots & z^{\beta_{N_{0}}+k_{N}m-1}\\
					\vdots & \vdots & \ddots & \vdots & \vdots & \vdots & \ddots & \vdots \\
					0 & 0 & \cdots & z^{0} & z^{\beta_{1}+1} & z^{\beta_{2}+1} & \cdots & z^{\beta_{N_{0}+1}} \\
					0 & 0 & \cdots & 0 & z^{\beta_{1}} & z^{\beta_{2}} & \cdots & z^{\beta_{N_{0}}} \\
					\vdots & \vdots & \ddots & \vdots & \vdots & \vdots & \ddots & \vdots
				\end{vmatrix}_{N\times N}
				=\,&(-1)^{c_{1}} \begin{vmatrix}
						z^{\beta_{1}} & z^{\beta_{2}} & \cdots & z^{\beta_{N_{0}}} \\
						z^{\beta_{1}-1} & z^{\beta_{2}-1} & \cdots & z^{\beta_{N_{0}}-1} \\
						\vdots & \vdots & \ddots & \vdots
				\end{vmatrix}_{N_{0}\times N_{0}}\\
				\equiv \, &(-1)^{c_{1}} |\mathbf{B}|.
			\end{aligned}
		\end{equation}
	\end{widetext}
	where $ c_{1} $ is {an} integer. Therefore, the zero root multiplicity $ \Gamma_{0} $ of $ \bar{W}_{N}^{[m,k,l]}(z; a) $ is determined by the above determinant $ |\mathbf{B}| $.
	
	In order to calculate the exponent of $ z $ of $ |\mathbf{B}| $, we conduct further analysis on the period \eqref{per1}. 
	According to the definition of $ \beta_{j} $ in \eqref{setb}, we can formally  divide all elements of the period \eqref{per1} into $ {\hat{k}} $ segments based on the pattern similar to \eqref{seg1}, as follows:
	\begin{widetext}
		\begin{equation}\label{per2b}
			\begin{aligned}
				\quad\,	&\boxed{\, z^{l} \,} \boxed{\, z^{l-m_{0}} \,} \cdots {\boxed{z^{l \,\mathrm{mod}\, {m_{0}}}}}\boxed{z^{l+k -(\lfloor {\frac{l}{m_{0}}}\rfloor +1)m_{0}}}\boxed{z^{l+k -(\lfloor {\frac{l}{m_{0}}}\rfloor +2)m_{0}}}  \cdots \boxed{z^{(l+k) \, \mathrm{mod}\, m_{0}}\,} \cdots \boxed{z^{(l+(\frac{m_{0}}{(m,k)}-1)k) \,\mathrm{mod}\, {m_{0}}}} \cdots  \boxed{z^{(l+(\frac{m_{0}}{(m,k)}-1)k) -(\hat{k}-1)m_{0} }}, 
			\end{aligned}
		\end{equation}
	\end{widetext}
	where
	\begin{equation}\label{hatk}
		{ \hat{k}=\frac{k}{(m,k)}},
	\end{equation}
	denote the length of $ \boxed{z^{i}} $ in \eqref{per2b} as $ \hat{N}_{i} $ given by
	\begin{equation}\label{eq03}
		\hat{N}_{i}=\left\lbrace 
		\begin{array}{ll}
			m_{1}+1, & 0\leq i< m_{0}, \\
			m_{1}, & m_{0}\leq i< k,
		\end{array}
		\right. 
	\end{equation}
	$ m_{1} $ and $ m_{0} $ are defined by Eq. \eqref{eq:mn}. Note that if $ l < m_{0} $, the first segment of \eqref{per2b} is {\boxed{z^{l \,\mathrm{mod}\, {m_{0}}}}}. Moreover, according to the expression of $ m $ in \eqref{eq:mn}, we have $ \hat{m}=m_{1}\hat{k}+\frac{m_{0}}{(m,k)} $.
	Then, it is easily to find that there are $ \frac{m_{0}}{(m,k)}  $ long segments and $ \hat{k}-\frac{m_{0}}{(m,k)} $ short segments in the $ \hat{k} $ segments of \eqref{per2b}, where the long segments and the short segments refer to the segments of length $ m_{1}+1 $ and $ m_{1} $, respectively. Further, we can derive that the long segments in \eqref{per2b} are $ \boxed{z^{(l+(j-1)k) \,\mathrm{mod}\, {m_{0}}}} $ with $ 1\leq j \leq \frac{m_{0}}{(m,k)} $.
	
	Therefore, for the initial $ N_{0} $ elements
	\begin{equation}\label{pper}
		z^{\beta_{1}}, z^{\beta_{2}}, \cdots, z^{\beta_{N_{0}}},
	\end{equation}
	in period \eqref{per1}, we can obtain the following conclusions.
	
\begin{itemize}
	\item[(i).] When $ {0}\leq N_{0} < \lfloor\frac{l}{m_{0}}\rfloor m_{1} +m_{1} +1 $, these elements \eqref{pper} can be expressed by
	\begin{equation}\label{seg3}
		\begin{aligned}
			\quad\boxed{\, z^{l} \,} \boxed{\, z^{l-m_{0}} \,} \cdots
			\boxed{z^{l-(\lfloor\frac{N_{0}}{m_{1}}\rfloor)m_{0}}},
		\end{aligned} 	
	\end{equation}
	where the lengths $ \hat{N}_{i} $ of segments $ \boxed{z^{i}} $ are
	\begin{equation}\label{hn1}
		\qquad \hat{N}_{i}=\left\lbrace 
		\begin{array}{ll}
			m_{1}, & i=l -(j-1) m_{0}, \, 1\leq j \leq \lfloor\frac{N_{0}}{m_{1}}\rfloor, \\
			N_{0}{\,\mathrm{mod}\,} m_{1}, & i= l-\lfloor \frac{N_{0}}{m_{1}} \rfloor m_{0}, \\
			0, & others.
		\end{array}
		\right. 
	\end{equation}

	\item[(ii).] When $ 	\lfloor\frac{l+(r-1)k}{m_{0}}\rfloor m_{1}+m_{1}+r \leq N_{0} < \lfloor\frac{l+rk}{m_{0}}\rfloor m_{1}+m_{1}+r+1 $ and $ 1\leq r\leq \frac{m_{0}}{(m,k)} $, these elements \eqref{pper} are described by the following segments
	\begin{equation}\label{seg4}
		\begin{aligned}
			&\boxed{\, z^{l} \,} \cdots {\boxed{z^{l \,\mathrm{mod}\, {m_{0}}}}}  \cdots \boxed{z^{(l+k) \, \mathrm{mod}\, m_{0}}\,} \cdots \\
			& \boxed{\, z^{(l+(r-1)k) \mathrm{mod}\, m_{0}} \,} \cdots \boxed{z^{l+rk - \lfloor \frac{N_{0}-r}{m_{1}} \rfloor m_{0}}},
		\end{aligned} 	
	\end{equation}
	where the lengths $ \hat{N}_{i} $ of segments $ \boxed{z^{i}} $ are 
	\begin{widetext}
		\begin{equation}\label{hn2}
			\quad\qquad \hat{N}_{i}=\left\lbrace 
			\begin{aligned}
				&m_{1}+1, \quad i=(l+(j-1)k) \, \mathrm{mod}\, m_{0}, \, 1\leq j \leq r, \\
				&m_{1}, \qquad i= \left\lbrace 
				\begin{array}{ll}
					l-(j-1)m_{0}, & 1\leq j \leq \lfloor\frac{l}{m_{0}}\rfloor,\\
					l+k-(j-1)m_{0}, & \lfloor\frac{l}{m_{0}}\rfloor +1 < j \leq \lfloor\frac{l+k}{m_{0}}\rfloor, \\
					\qquad \vdots& \\
					l+rk-(j-1)m_{0}, & \lfloor\frac{l+(r-1)k}{m_{0}}\rfloor +1 < j \leq \lfloor\frac{N_{0}-r}{m_{1}}\rfloor,
				\end{array}
				\right.
				\\
				&(N_{0}-r){\,\mathrm{mod}\,} m_{1}, \quad i= l+rk - \lfloor \frac{N_{0}-r}{m_{1}} \rfloor m_{0},\\
				&0, \qquad \quad others.
			\end{aligned}
			\right. 
		\end{equation}
	\end{widetext}

\end{itemize}	

Now, based on the results \eqref{hn1} and \eqref{hn2} of $ \hat{N}_{i} $, we can obtain the values of $ N_{i} $ to calculate $ |\mathbf{B}| $ below.
	
\begin{itemize}

	\item[(i).] If $ \hat{N}_{0}=0 $ in Eqs. \eqref{hn1} and \eqref{hn2}, we derive
	\begin{equation}\label{eq:ni1}
		N_{i}=\hat{N}_{i}, \qquad 1\leq i < k.
	\end{equation}
	
	\item[(ii).] If $ \hat{N}_{i}\ne 0 (0\leq i \leq j-1)$ and $\hat{N}_{j}= 0 (1\leq j< k)$, the first row in $|\mathbf{B}|$ possesses the following terms
	\begin{equation}\label{term0}
		z^{0}, \quad z^{1}, \quad \cdots, \quad z^{j-1}.
	\end{equation}
	Then, we can ignore the first $ j $ rows and the columns where these items \eqref{term0} are located, namely,
	\begin{widetext}
		\begin{equation}\label{term}
		\quad	|\mathbf{B}| =(-1)^{c_{2}} \begin{vmatrix}
				z^{0} & 0 & \cdots \\
				z^{k} & z^{k-1} &\cdots \\
				\vdots & \vdots & \ddots \\
				z^{k(\hat{N}_{0}-1)} & z^{k(\hat{N}_{0}-1)-1} &\cdots \\
				z^{1} & z^{0} & \cdots \\
				z^{1+k} & z^{k} & \cdots \\ 
				\vdots & \vdots & \ddots \\
				z^{1+k(\hat{N}_{1}-1)} & z^{k(\hat{N}_{1}-1)} &\cdots \\
				\vdots & \vdots & \ddots \\ 
				z^{j-1} & z^{j-2} & \cdots \\
				z^{j-1 +k} & z^{j-2 +k} & \cdots \\
				\vdots & \vdots & \ddots \\  
				z^{j-1 + k(\hat{N}_{j-1}-1)} & z^{j-2 + k(\hat{N}_{j-1}-1)} & \cdots \\
				z^{j+1} & z^{j} & \cdots \\
				z^{j+1+k} & z^{j+k} & \cdots \\
				\vdots & \vdots & \ddots \\  
				z^{j+1+ k(\hat{N}_{j+1}-1)} & z^{j + k(\hat{N}_{j+1}-1)} & \cdots \\
				\vdots & \vdots & \ddots \\ 
			\end{vmatrix}_{N_{0}\times N_{0}}^{T}
			=\quad(-1)^{c_{3}}
			\begin{vmatrix}
				z^{k-j} & z^{k-j-1} &\cdots \\
				\vdots & \vdots & \ddots \\
				z^{k(\hat{N}_{0}-1)-j} & z^{k(\hat{N}_{0}-1)-j-1} &\cdots \\
				z^{1+k-j} & z^{k-j} & \cdots \\ 
				\vdots & \vdots & \ddots \\
				z^{1+k(\hat{N}_{1}-1)-j} & z^{k(\hat{N}_{1}-1)-j} &\cdots \\
				\vdots & \vdots & \ddots \\ 
				z^{k-1} & z^{k-2} & \cdots \\
				\vdots & \vdots & \ddots \\  
				z^{ k(\hat{N}_{j-1}-1)-1} & z^{k(\hat{N}_{j-1}-1)-2} & \cdots \\
				z^{1} & z^{0} & \cdots \\
				z^{1+k} & z^{k} & \cdots \\
				\vdots & \vdots & \ddots \\  
				z^{1+ k(\hat{N}_{j+1}-1)} & z^{ k(\hat{N}_{j+1}-1)} & \cdots \\
				\vdots & \vdots & \ddots \\ 
			\end{vmatrix}_{(N_{0}-j)\times( N_{0}-j)}^{T},
		\end{equation}
	\end{widetext}
	where $ c_{2}$ and  $c_{3}  $ are all the integers.
	Then, we have
	\begin{equation}\label{eq:ni2}
		\qquad	N_{i}=\left\lbrace 
		\begin{aligned}
			&\hat{N}_{i+j}, & 1\leq i \leq k-j-1, \\
			&\hat{N}_{i+j-k}-1, & k-j\leq i < k.
		\end{aligned}
		\right. 
	\end{equation}

	\item[(iii).] If $ \hat{N}_{i}\ne 0 (0\leq i < k)$, the first row in $|\mathbf{B}|$ have the following terms
	\begin{equation}\label{term2}
		z^{0}, \quad z^{1}, \quad \cdots, \quad z^{\bar{r}k-1},
	\end{equation}
	where $\bar{r}=\min\left\lbrace \hat{N}_{0}, \hat{N}_{1},\ldots, \hat{N}_{k-1}\right\rbrace  $. Similar to the calculation in Eq. \eqref{term}, we disregard the first $ rk $ rows and the columns where these terms \eqref{term2} are situated in the determinant $ |\mathbf{B}| $.
	Let 
	\begin{equation}\label{nbi}
		\bar{N}_{i}=\hat{N}_{i}- \bar{r}, \quad 0\leq i< k,
	\end{equation}
	then $ \bar{N}_{i} $ necessarily admit one of the above condition (i) and (ii). We generate the following results.
	\begin{itemize}
		\item[(a).] If $ \bar{N}_{i} $ satisfies the above condition (i), we set
		\begin{equation}\label{eq:ni3}
			N_{i}=\bar{N}_{i}, \qquad 1\leq i < k.
		\end{equation}
		
		\item[(b).] If $ \bar{N}_{i} $ satisfies the above condition (ii), we take
		\begin{equation}\label{eq:ni4}
			\qquad	N_{i}=\left\lbrace 
			\begin{aligned}
				&\bar{N}_{i+j}, & 1\leq i \leq k-j-1, \\
				&\bar{N}_{i+j-k}-1, & k-j\leq i < k.
			\end{aligned}
			\right. 
		\end{equation}
	\end{itemize}
\end{itemize}
Hence, when $ (m,k)=1 $, based on the different values of $ N_{i} $ in Eqs. \eqref{eq:ni1}--\eqref{eq:ni4} for the distinct conditions, the exponent of $ z $ of $ |\mathbf{B}| $, i.e., the zero root multiplicity $ \Gamma_{0} $ of the polynomials $ W_{N}^{[m,k,l]}(z) $, is obtained by utilizing the formula \eqref{calzp}.

\item \textbf{Cases of $ (m,k)\ne1 $}

When $ (m,k)\ne 1 $, we initially plan to prove all elements $ \beta_{j} (1\leq j \leq \hat{m}) $ of period \eqref{per1} are ergodic in the first $ \hat{m} $ terms of the following arithmetic sequence
\begin{equation}\label{qrse1}
\quad	\left\lbrace   l {\,\mathrm{mod}\,} (m,k) + (j-1)(m,k), \quad j=1,2,3,\ldots \right\rbrace.
\end{equation}
Suppose $ \beta_{j}={p_{j}}(m,k)+ \hat{l} $ with $  {p_{j}},  \hat{l} \in \mathbb{N} $ and $ 0\leq  \hat{l}< (m,k)  $.
According to the definitions \eqref{setb} of $ \beta_{j} $, we have
\begin{equation}\label{eq04b}
	\begin{aligned}
		l+(j-1)k &=\hat{p}_{j} m +\beta_{j}\\
		&=\left(\hat{p}_{j}\hat{m} + p_{j} \right) (m,k) + \hat{l},
	\end{aligned}
\end{equation} 
which implies $ l= \left( \hat{p}_{j}{\hat{m}} + p_{j} -\frac{(j-1)k}{(m,k)} \right) (m,k) + \hat{l} $ with $ \hat{m}=\frac{m}{(m,k)} $ , i.e., $ \hat{l}=l {\,\mathrm{mod}\,} (m,k) $. As $ \beta_{j} $ $ (1\leq j \leq \hat{m}) $ are distinct with $ \beta_{j}<m $, and the increasing arithmetic sequence \eqref{qrse1} has only the first $ \hat{m} $ terms less than $ m $, we prove ergodic property \eqref{qrse1} of period \eqref{per1} when $ (m,k)\ne1 $.

Moreover, it is easy to find that the sequences
\begin{equation}\label{seq1}
	\beta_{1}+im, \beta_{2}+im, \cdots, \beta_{\hat{m}}+im
\end{equation}
are ergodic in the above arithmetic sequence \eqref{qrse1} with $ 1 + i\hat{m} \leq j \leq {(i+1)\hat{m}} $ and $ 1\leq i\leq(k_{N}-1) $. Then, we derive that the powers of $ z $ in the preceding $ k_{N}\hat{m} $ elements of the first row of determinant \eqref{bldet0}, i.e.,
\begin{equation}\label{seq2}
	\begin{aligned}
		\quad&z^{\beta_{1}}, \cdots, z^{\beta_{\hat{m}}}, z^{\beta_{1}+m}, \cdots, z^{\beta_{\hat{m}}+m}, \cdots, z^{\beta_{1}+(k_{N}-1)m},\\
		&\cdots,  z^{\beta_{\hat{m}} +(k_{N}-1)m},
	\end{aligned}
\end{equation} 
are ergodic in the above arithmetic sequence \eqref{qrse1} with $ 1\leq j \leq k_{N}\hat{m} $.
Therefore, the preceding $ k_{N}\hat{m} $ elements in the first row of determinant \eqref{bldet0} can be expressed by
\begin{equation}\label{bldet1}
	\begin{aligned}
	 \boxed{z^{{l}_{1}}}\boxed{{z^{{l}_{2}}}} \cdots \boxed{z^{{l}_{\hat{k}}}}
	\end{aligned}
\end{equation}
where the definitions of segments \boxed{z^{{l}_{j}}} are given in Eq. \eqref{seg1}, the lengths of $ \boxed{z^{{l}_{j}}} $ are denoted as $ \tilde{N}_{{l}_{j}}^{(1)} $, $ \tilde{N}_{{l}_{j}}^{(1)} $ are obtained by
\begin{equation}\label{seg2}
	\qquad	\tilde{N}_{{l}_{j}}^{(1)}=\left\lbrace 
	\begin{array}{ll}
		\lfloor \frac{k_{N}\hat{m}}{\hat{k}} \rfloor +1, & 1\leq j \leq k_{N}\hat{m} {\,\mathrm{mod}\,} \hat{k},\\
		\lfloor \frac{k_{N}\hat{m}}{\hat{k}} \rfloor, &   k_{N}\hat{m} {\,\mathrm{mod}\,} \hat{k} < j \leq \hat{k},
	\end{array}
	\right.
\end{equation}
$ {l}_{j}=l {\,\mathrm{mod}\,} (m,k) + (j-1)(m,k)  $, and $\hat{k} $ is given in Eq. \eqref{hatk}.

In addition, we will discuss the last $ N_{0} $ elements 
\begin{equation}\label{bldet2}
	z^{\beta_{1}+k_{N}m},z^{\beta_{2}+k_{N}m},\cdots, z^{\beta_{N_{0}}+k_{N}m},
\end{equation}
in the first row of determinant \eqref{bldet0}. Based on the foregoing discussion of period \eqref{per1}, we can determine that the initial $ N_{0} $ elements \eqref{pper} of the period are also classified as the above $ \hat{k} $ segments \eqref{bldet1}, but their lengths are $ \hat{N}_{i} $ given in Eqs. \eqref{hn1} and \eqref{hn2}. Thus, these elements \eqref{bldet2} are also classified as the following segments
\begin{equation}\label{pper3}
	\begin{aligned}
		\qquad	&z^{{l}_{1}+k_{N}m}, z^{{l}_{1}+k_{N}m+k}, \cdots, z^{{{l}_{1}}+k_{N}m +k(\hat{N}_{l_{1}}-1)};\\
		&z^{{l}_{2}+k_{N}m}, z^{{l}_{2}+k_{N}m+k}, \cdots, z^{{{l}_{2}}+k_{N}m +k(\hat{N}_{l_{2}}-1)};\\
		&\qquad\vdots\\
		&z^{{l}_{\hat{k}}+k_{N}m}, z^{{l}_{\hat{k}}+k_{N}m+k}, \cdots, z^{{l}_{\hat{k}}+k_{N}m +k(\hat{N}_{{l}_{\hat{k}}}-1)}.
	\end{aligned}
\end{equation}
Since the exponents of these elements 
\begin{equation}\label{pper3b}
	z^{{l}_{1}+k_{N}m}, z^{{l}_{2}+k_{N}m},\cdots,z^{{l}_{\hat{k}}+k_{N}m},
\end{equation}
are the $ (k_{N}\hat{m}+1 )$-th to $ (k_{N}\hat{m}+\hat{k})$-th items of the above arithmetic sequence \eqref{qrse1}, we can find that the elements \eqref{pper3b} happen to be the subsequent item of the segments \eqref{bldet1}.
Then, by concatenating the segments
\begin{equation}\label{pper4}
	z^{{l}_{j}+k_{N}m}, z^{{l}_{j}+k_{N}m+k}, \cdots, z^{l_{j}+k_{N}m +k(\hat{N}_{l_{j}}-1)}, 
\end{equation}
in \eqref{pper3} with the segments $ \boxed{z^{({l}_{j}+k_{N}m) {\,\mathrm{mod}\,} k}} $ in \eqref{bldet1}, we can partition all elements in the first row of determinant \eqref{bldet0} into $\hat{k} $ segments resembling  \eqref{bldet1}, but the lengths of these segments have changed. We denote the lengths of these segments $ \boxed{z^{l_{j}}} $ as $\tilde{N}_{{l}_{j}}^{(2)} $ $ (1\leq j\leq \hat{k}) $, given by
\begin{equation}\label{eq:nie0}
	\begin{aligned}
	\quad \tilde{N}_{{l}_{j}}^{(2)}= \left\lbrace 
	\begin{array}{ll}
		\tilde{N}_{{l}_{j}}^{(1)} +\hat{N}_{l_{j}+k -\left( k_{N}m \right) \,\mathrm{mod}\, k}, &   0\leq l_{j} < \left( k_{N}m \right) \,\mathrm{mod}\, k,\\
		\tilde{N}_{{l}_{j}}^{(1)} + \hat{N}_{l_{j}-\left( k_{N}m \right) \,\mathrm{mod}\, k}, 
		&  \left( k_{N}m \right) \,\mathrm{mod}\, k \leq l_{j} < k,\\
	\end{array}
	\right.
	\end{aligned}
\end{equation}
with $ \tilde{N}_{{l}_{j}}^{(1)} $ and $ l_{j} $ being defined in Eq. \eqref{seg2}.

Furthermore, it is evident that whether $ (m,k) \mid {l} $ determines the presence of a constant term within the first row of determinant \eqref{bldet0}. Thus, when $ (m,k)\ne1 $, we discuss the determinant \eqref{bldet0} in two cases.

	\begin{itemize}
	\item[(i).] For $ (m,k) \nmid   {l} $, the first row of determinant \eqref{bldet0} does not contain the element $ z^{0} $. According to the above analysis for elements of the first row in determinant \eqref{bldet0}, we take
	\begin{equation}\label{eq:nie1}
		\quad N_{i}= \left\lbrace
		\begin{array}{ll}
			\tilde{N}_{{l}_{j}}^{(2)}, & i=l_{j}, \, 1\leq j\leq \hat{k},\\
			0, & others,
		\end{array}
		\right.
	\end{equation}
	where $ \tilde{N}_{{l}_{j}}^{(2)} $ and $ l_{j} $ are given in Eqs. \eqref{eq:nie0} and \eqref{seg2}.

	\item[(ii).] For $ (m,k) \mid {l} $, i.e., $ l {\,\mathrm{mod}\,} (m,k)  =0 $, the first row of determinant \eqref{bldet0} certainly has the element $ z^{0} $, and all elements of the first row can be described by
	\begin{equation}\label{bldet3}
		\qquad\qquad	\boxed{z^{0}}\boxed{{z^{(m,k)}}} \boxed{{z^{2(m,k)}}} \cdots \boxed{z^{(\hat{k}-1)(m,k)}}.
	\end{equation}
	We {eliminate} the row and the column where $ z^{0} $ are located in determinant \eqref{bldet0}, which is similar to the calculation of Eq. \eqref{term}. Then, we divide the first row of the remaining part of the determinant into new $ \hat{k} $ segments, as follows: 
	\begin{equation}\label{bldet3b}
	\qquad\qquad	\boxed{z^{k-1}}\boxed{{z^{(m,k)-1}}} \boxed{{z^{2(m,k)-1}}} \cdots \boxed{z^{(\hat{k}-1)(m,k)-1}}.
	\end{equation}
	The length $ N_{i} $ of each segment in \eqref{bldet3b} are provided below:
	\begin{equation}\label{eq:nie2}
		\qquad\qquad N_{i}= \left\lbrace
		\begin{array}{ll}
			\tilde{N}_{{l}_{1}}^{(2)} -1, & i=k-1, \\ 
			\tilde{N}_{{l}_{j}}^{(2)}, & i=l_{j}-1, 2\leq j\leq \hat{k},\\
			0, & others,
		\end{array}
		\right.
	\end{equation}
	where $ \tilde{N}_{{l}_{j}}^{(2)} $ and $ l_{j} $ are given in Eqs. \eqref{eq:nie0} and \eqref{seg2}.
\end{itemize}
Therefore, when $ (m,k)\ne1 $, by substituting the values \eqref{eq:nie1} and \eqref{eq:nie2} of $ N_{i} $ into the formula \eqref{calzp}, we can calculate the power of $ z $ in determinant \eqref{bldet0}, i.e., the multiplicity of the zero root of the polynomials $ {W}_{N}^{[m,k,l]}(z) $.

\end{enumerate}

\noindent \textbf{1.3. The expression \eqref{eq:whp} of $ W_{N}^{[m,k,l]}(z) $}

	Here, we will prove properties \eqref{eq:whp} of $ {W}_{N}^{[m,k,l]}(z) $ provided in Theorem \ref{th:whp}. Based on definition \eqref{eq:p1} of $ p_{j}^{[m]}(z) $, we have
	\begin{equation}\label{eq:ppro}
		p_{j}^{[m]}(bz)=b^{j}p_{j}^{[m]}(z),
	\end{equation}
	where $ b $ is any one of the $ m $-th root of $ 1 $, i.e., $ b^{m}=1 $.
	Eq. \eqref{eq:ppro} leads to
	\begin{equation}\label{eq:pwp}
		{W}_{N}^{[m,k,l]}(bz)=b^{\Gamma}{W}_{N}^{[m,k,l]}(z).
	\end{equation}
	Moreover, as the multiplicity of the zero root of $ {W}_{N}^{[m,k,l]}(z) $ is $ \Gamma_{0} $, the polynomials can be rewrite as
	\begin{equation}\label{eq:whp2}
		{W}_{N}^{[m,k,l]}(z)=z^{\Gamma_{0}} \hat{W}^{[m,k,l]}(z),
	\end{equation}
	where $  \hat{W}^{[m,k,l]}(z) $ is defined by \eqref{eq:whp}. Combining Eqs. \eqref{eq:pwp} and \eqref{eq:whp2}, we obtain
	\begin{equation}\label{eq:whp3}
		\hat{W}^{[m,k,l]}(bz) = b^{\Gamma-\Gamma_{0}} \hat{W}^{[m,k,l]}(z).
	\end{equation}
	It is found that $ \Gamma-\Gamma_{0} $ is the multiple of $ m $ due to Eq. \eqref{eq01} with $ j=\Gamma_{0} $. According to the definition of $ b $, Eq. \eqref{eq:whp3} holds
	\begin{equation}\label{eq:whp4}
		\hat{W}^{[m,k,l]}(bz) = \hat{W}^{[m,k,l]}(z),
	\end{equation} 
	which dictates that $ \hat{W}^{[m,k,l]}(\hat{z}) $ is a polynomial of $ \hat{z}=z^{m} $.

	Finally, we complete the proof of Theorem \ref{th:whp}.
\end{proof}

\subsection{\label{app5} MI analysis of the plane-wave solution}

The MI analysis involves linearizing $n$-DNLS equations and calculating perturbation functions and the MI gain index. Substituting $ \mathbf{{q}}+ \varepsilon \mathbf{\hat{q}} $ $ (|\varepsilon|\ll 1, \varepsilon \in \mathbb{R} )$ into $n$-DNLS equations \eqref{eq:nDNLS}, we can obtain the linearized $n$-DNLS equations
\begin{equation}\label{aa0}
	\begin{aligned}
			\ii \mathbf{\hat{q}}_{t}+\mathbf{\hat{q}}_{xx} - \ii \left(  \mathbf{\hat{q}}\mathbf{q}^{\dagger}\mathbf{q} +\mathbf{q}\mathbf{\hat{q}}^{\dagger}\mathbf{q} +\mathbf{q}\mathbf{q}^{\dagger}\mathbf{\hat{q}}\right)_{x} =0,
	\end{aligned}
\end{equation}
where $ \mathbf{\hat{q}} $ is an $ n\times 1 $ matrix function.

Then, we construct the solutions of linearized $n$-DNLS equations by the squared eigenfunction method. For this purpose, let us start from the stationary zero curvature equations:
\begin{equation}\label{aa1}
	\Psi_{x}=[\mathbf{U}, \Psi], \quad \Psi_{t}=[\mathbf{V}, \Psi], \quad 
	\Psi = \begin{pmatrix}
		\Psi_{11} & \Psi_{12} \\
		\Psi_{21} & \Psi_{22} 
	\end{pmatrix},
\end{equation}
where $ \Psi_{11}, \Psi_{12}, \Psi_{21}$ and $\Psi_{22}$ are $1\times1, 1\times n, n\times 1$ and $n\times n$ matrices, respectively. Then, we have
\begin{equation}\label{aa2}
	\renewcommand{\arraystretch}{1.2}
	\begin{array}{l}
		\Psi_{11,x}=\lambda^{-1}(\mathbf{q}^{\dagger}\Psi_{21}-\Psi_{12}\mathbf{q}), \\  \Psi_{12,x}=\lambda^{-1}(\mathbf{q}^{\dagger}\Psi_{22}-\Psi_{11}\mathbf{q}^{\dagger})+2\ii\lambda^{-2}\Psi_{12}, \\
		\Psi_{21,x}=\lambda^{-1}(\mathbf{q}\Psi_{11}-\Psi_{22}\mathbf{q})-2\ii\lambda^{-2}\Psi_{21}, \\ \Psi_{22,x}=\lambda^{-1}(\mathbf{q}\Psi_{12}-\Psi_{21}\mathbf{q}^{\dagger}).
	\end{array}
\end{equation}
Further, we calculate the second-order derivative of $\Psi_{12}$ and $\Psi_{21}$ concerning $ x$ below:
	\begin{equation}\label{aa3}
		\begin{aligned}
			\Psi_{12,xx}=& -\lambda^{-1}\left( \Psi_{11}(2\ii\lambda^{-2}\mathbf{q}^{\dagger}+\mathbf{q}^{\dagger}_{x}) -(2\ii\lambda^{-2}\mathbf{q}^{\dagger}+\mathbf{q}^{\dagger}_{x})\Psi_{22}\right)  \\ &+\lambda^{-2}\mathbf{q}^{\dagger}(\mathbf{q}\Psi_{12}-\Psi_{21}\mathbf{q}^{\dagger}) +\lambda^{-2}(\Psi_{12}\mathbf{q}-\mathbf{q}^{\dagger}\Psi_{21})\mathbf{q}^{\dagger} \\ 
			&-4 \lambda^{-4}\Psi_{12},\\
			\Psi_{21,xx}=& \lambda^{-1}\left( (-2\ii\lambda^{-2}\mathbf{q}+\mathbf{q}_{x})\Psi_{11} {+}\Psi_{22}(2\ii\lambda^{-2}\mathbf{q}-\mathbf{q}_{x})\right) \\
			&+\lambda^{-2}\mathbf{q}(\mathbf{q}^{\dagger}\Psi_{21}-\Psi_{12}\mathbf{q}) +\lambda^{-2}(\Psi_{21}\mathbf{q}^{\dagger}-\mathbf{q}\Psi_{12})\mathbf{q} \\
			&-4\lambda^{-4}\Psi_{21}.
		\end{aligned}
	\end{equation}
Moreover, for the $t$-part of Eq. \eqref{aa1}, we generate
\begin{equation}\label{aa4}
	\renewcommand{\arraystretch}{1.2}
	\begin{array}{l}
		\Psi_{11,t}=V_{12}\Psi_{21}-\Psi_{12}V_{21}, \\
		\Psi_{12,t}=V_{11}\Psi_{12} +V_{12}\Psi_{22} -\Psi_{11}V_{12}-\Psi_{12}V_{22},\\
		\Psi_{21,t}=V_{21}\Psi_{11} +V_{22}\Psi_{21} -\Psi_{21}V_{11}-\Psi_{22}V_{21}, \\
		\Psi_{22,t}=V_{21}\Psi_{12} +V_{22}\Psi_{22} -\Psi_{21}V_{12}-\Psi_{22}V_{22},
	\end{array}
\end{equation}
where 
\begin{equation}\label{aa5}
	\renewcommand{\arraystretch}{1.2}
	\begin{array}{l}
		V_{11}=\ii\lambda^{-2}(2\lambda^{-2}+\mathbf{q}^{\dagger}\mathbf{q}), \\
		V_{12}=\lambda^{-1}\left( 2\lambda^{-2}\mathbf{q}^{\dagger} +\mathbf{q}^{\dagger}\mathbf{q}\mathbf{q}^{\dagger}-\ii\mathbf{q}^{\dagger}_{x}\right) ,\\
		V_{21}=\lambda^{-1}\left( 2\lambda^{-2}\mathbf{q} +\mathbf{q}\mathbf{q}^{\dagger}\mathbf{q}+\ii\mathbf{q}_{x}\right) , \\
		V_{22}=-\ii\lambda^{-2}\left( 2\lambda^{-2}\mathbb{I}_{n}+ \mathbf{q}\mathbf{q}^{\dagger}\right).
	\end{array}
\end{equation}
By combining with Eqs. \eqref{aa3} and \eqref{aa4}, we obtain the following equations:
\begin{equation}\label{aa6}
	\begin{aligned}
			\ii \Psi_{21,t}=-\Psi_{21,xx} +\ii \Psi_{21,x}\mathbf{q}^{\dagger}\mathbf{q} +\ii\mathbf{q}\Psi_{12,x}\mathbf{q} +\ii\mathbf{q}\mathbf{q}^{\dagger}\Psi_{21,x},\\
			\ii \Psi_{12,t}= \Psi_{12,xx} +\ii \Psi_{12,x}\mathbf{q}\mathbf{q}^{\dagger} +\ii\mathbf{q}^{\dagger}\Psi_{21,x}\mathbf{q}^{\dagger} +\ii\mathbf{q}^{\dagger}\mathbf{q}\Psi_{12,x}.
	\end{aligned}
\end{equation}
It is evident that the symmetry relationship $ \Psi_{12}= \Psi_{21}^{\dagger} $ ensures $ \Psi_{21,x} $ to satisfy the linearized equation \eqref{aa0}.

Therefore, the solution $ \mathbf{\hat{q}} $ of Eq. \eqref{aa0} can be generated by the eigenfunction of Lax pair \eqref{eq:sp}. Suppose $ \phi_{i}(\lambda) = \Phi(\lambda)(1,0,0,\ldots)^{T}$ is a vector solution for Lax pair \eqref{eq:sp}, then $ (\phi_{i}(\lambda^{*}))^{\dagger}\sigma_{3}= (1,0,0,\ldots)\sigma_{3}(\Phi(\lambda^{*}))^{\dagger} \sigma_{3}$ is a vector solution for the adjoint Lax pair \eqref{eq:ad1}, where $ \phi_{i}(\lambda)=[\phi_{i,1}(\lambda), (\phi_{i,2}(\lambda))^{T}]^{T} $, $ \phi_{i,1}(\lambda) $ and $ \phi_{i,2}(\lambda) $ are the $ 1\times 1 $ and $ n\times 1  $ matrix functions, respectively. This implies that $ \Psi_{1}=\phi_{i}(\lambda)(\phi_{j}(\lambda^{*}))^{\dagger}\sigma_{3} $ and $ \Psi_{2}=\phi_{j}(\lambda^{*}) (\phi_{i}(\lambda))^{\dagger} \sigma_{3} $ satisfy stationary zero-curvature equations \eqref{aa1}. Since Eq. \eqref{aa0} is independent of $ \lambda $, we conclude that the solutions
\begin{equation}\label{aa7}
		\mathbf{\hat{q}}=\left[ \phi_{i,2}(\lambda)(\phi_{j,1}(\lambda^{*}))^{*}-\phi_{j,2}(\lambda^{*})(\phi_{i,1}(\lambda))^{*} \right]_{x}
\end{equation}
satisfy the linearized $ n $-DNLS equations.

Now, assume $\phi_{i}(\lambda)=GE_{i}(\lambda)\ee^{\zeta_{i}(\lambda)} $ and $ 	\phi_{j}(\lambda)=GE_{j}(\lambda)\ee^{\zeta_{j}(\lambda)} $ are the solutions of Lax pair \eqref{eq:sp}, where the matrix function $ E_{i}(\lambda)\ee^{\zeta_{i}(\lambda)} $ solves Eq. \eqref{eq:sp2} with
\begin{widetext}
	\begin{equation}\label{aa8}
		\begin{aligned}
			&E_{i}(\lambda)=\left(  1,\, \frac{-\ii a_{1}}{\lambda(\chi_{i}(\lambda)+b_{1})},\, \frac{-\ii a_{2}}{\lambda(\chi_{i}(\lambda)+b_{2})},\dots,\frac{-\ii a_{n}}{\lambda(\chi_{i}(\lambda)+b_{n})}  \right)^{T},\\
			&\zeta_{i}(\lambda)=\ii\left( (\chi_{i}(\lambda)-\lambda^{-2})x +((\chi_{i}(\lambda))^{2}+\|\mathbf{a}\|^{2}_{2}\chi_{i}(\lambda)-2\lambda^{-4})t \right).
		\end{aligned} 
	\end{equation}
\end{widetext}

Thus, we obtain the perturbation form:
\begin{equation}\label{aa9}
	\mathbf{\hat{q}}= \begin{pmatrix}
		\left( f_{1,1} \ee^{\ii g(\lambda)(x+\Omega t)}+f_{2,1} \ee^{-\ii g(\lambda)(x+\Omega^{*}t)}\right) \ee^{\ii \theta_{1}}\\
		\left( f_{1,2} \ee^{\ii g(\lambda)(x+\Omega t)}+f_{2,2} \ee^{-\ii g(\lambda)(x+\Omega^{*}t)}\right) \ee^{\ii \theta_{2}}\\ \vdots \\
		\left( f_{1,n} \ee^{\ii g(\lambda)(x+\Omega t)}+f_{2,n} \ee^{-\ii g(\lambda)(x+\Omega^{*}t)}\right) \ee^{\ii \theta_{n}}
	\end{pmatrix},
\end{equation}
where $ f_{1,k} =\frac{ a_{k}}{\lambda} \left( 1-\frac{\chi_{j}(\lambda)}{(\chi_{i}(\lambda)+b_{k})}\right)  $, $ f_{2,k} =\frac{ a_{k}}{\lambda^{*}} \left( \frac{(\chi_{i}(\lambda))^{*}}{(\chi_{j}(\lambda^{*})+b_{k})} -1\right)  $, $ g(\lambda)=\chi_{i}(\lambda)-\chi_{j}(\lambda) \in \mathbb{R} $, and $ \Omega=\chi_{i}(\lambda)+ \chi_{j}(\lambda) +\|\mathbf{a}\|^{2}_{2}$. Note that the parameters $ g(\lambda) $ and $ \Omega $ are the perturbation frequency and the gain index, respectively, determined by the following equations
\begin{equation}\label{aa10}
	\begin{aligned}
		&1+\sum_{k=1}^{n}\frac{-\lambda^{-2}a_{k}^{2}}{(\chi_{i}(\lambda)+b_{k})(\chi_{i}(\lambda)-g(\lambda)+b_{k})}=0,\\
		&1+\sum_{k=1}^{n}\frac{-4\lambda^{-2}a_{k}^{2}}{\left( \Omega-\|\mathbf{a}\|^{2}_{2} +2b_{k}\right) ^{2}-(g(\lambda))^{2}}=0.
	\end{aligned}
\end{equation}
Therefore, we conclude that 
\begin{enumerate}
	\item If $ g(\lambda)\ne 0 $ and $ \Im(\Omega)=0 $, the plane wave solution \eqref{eq:seed} is modulational stability (MS). Moreover, if $ g(\lambda)\rightarrow 0 $ and $ \Im(\Omega)=0 $, the plane wave solution \eqref{eq:seed} is baseband MS.
	
	\item If $ g(\lambda)\ne 0 $ and $ \Im(\Omega)\ne0 $, the plane wave solution \eqref{eq:seed} is MI. Furthermore, if $ g(\lambda)\rightarrow 0 $ and $ \Im(\Omega)\ne0 $, the plane wave solution \eqref{eq:seed} is baseband MI.
\end{enumerate} 
Note that, when the perturbation frequency $ g(\lambda)\rightarrow 0 $, we have the algebraic equation
\begin{equation}\label{aa11}
	1- \sum_{k=1}^{n}\frac{\lambda^{-2}a_{k}^{2}}{(\chi_{i}(\lambda)+b_{k})^{2}} =0.
\end{equation}

Furthermore, we can discern that the rogue wave solutions obtained under the plane wave backgrounds \eqref{eq:seed} correspond to the conditions of MI. For instance, for the rogue waves discussed in this article, we necessitate $ \Im(\chi(\lambda))\ne0 $, which aligns with the MI conditions with $\Im(\Omega) =\Im(\chi(\lambda))$. 

Moreover, when the plane solution \eqref{eq:seed} is MS, we can also construct the rational solution. However, this solution is a rational traveling wave solution rather than a rogue wave solution. For instance, we can generate a traveling wave solution of the $ 2 $-DNLS equation below. When $ \lambda_{0}\in \mathbb{R} \cup \ii\mathbb{R} $, characteristic polynomial \eqref{eq:cp1} has a pair of conjugate complex roots $ \chi_{1} $ and $ \chi_{2} $. Then, we take a particular solution of spectral problem \eqref{eq:sp} with seed solution \eqref{eq:seed}, as follows:
\begin{equation}\label{aa12}
	\phi=\mathbf{G}\mathbf{E} \, \mathrm{diag}(\ee^{\zeta_{1}},\ee^{\zeta_{2}},\ee^{\zeta_{3}}) \begin{pmatrix}
		1 \\ C_{3} \\ 0
	\end{pmatrix},
\end{equation} 
where $C_{3} $ is a constant. Further, if $ \Im(\chi_{1})=0 $ corresponds to the conditions of MS, by choosing appropriate $ C_{3} $ and using the limit technique, we can obtain the following rational traveling wave solutions
\begin{equation}\label{aa13}
	q_{k}^{[1]}=a_{k}\ee^{\ii \theta_{k}}\frac{  	x+(\|\mathbf{a}\|_{2}^{2} +2\chi_{1})t} {  x+(\|\mathbf{a}\|_{2}^{2} +2\chi_{1})t  +\ii \frac{b_{k}}{2\chi_{1}(\chi_{1}+b_{k})}}, \quad k=1,2.
\end{equation}

\subsection{\label{app6} The proof of Proposition \ref{pr1}}

\begin{proof}
	(Necessity $\Rightarrow$) When $ \chi_{0} $ is an $ (n+1) $-multiple root of polynomial \eqref{eq:cp1} with $ \lambda=\lambda_{0} $, we have
	\begin{equation}\label{eq:cp2}
		D(\chi, \lambda_{0})=(\chi-\chi_{0})^{n+1}.
	\end{equation}
	Grouping the terms according to the power of $ \chi $ in Eqs. \eqref{eq:cp1} and \eqref{eq:cp2}. To the term $ \cO(\chi^{n}) $, we obtain
	\begin{equation}\label{eq:cpex1}
		\sum_{j=1}^{n} b_{j}-2\xi_{0}=-(n+1)\chi_{0},
	\end{equation}
	which implies that the expression of $ \chi_{0} $ in Eq. \eqref{eq:ab2} holds.
	
	Moreover, suppose that $ b_{1}=b_{2}=\cdots=b_{s}  (s\geq 2)$, but $ b_{1} $ and $ b_{j}$ $ (s< j\leq n) $ are mutually different. We deduce
	\begin{equation}\label{eq:cpd}
		D(\chi, \lambda_{0})= \hat{D}(\chi, \lambda_{0})(\chi+b_{1})^{s-1},
	\end{equation}
	where
	\begin{equation}\label{eq:cpd2}
		\begin{aligned}
		\hat{D}(\chi, \lambda_{0}) =&(\chi-2\xi_{0})\prod_{j=s}^{n}(\chi+b_{j})+\xi_{0}\prod_{i=s+1}^{n}(\chi+b_{i}) \left(\sum_{j=1}^{s} {a_{j}^{2}}\right) \\ &+\xi_{0}\sum_{j=s+1}^{n}\left( {a_{j}^{2}} \prod_{\substack{i=s\\i\ne j}}^{n}(\chi+b_{i})\right).
		\end{aligned}
	\end{equation}
	Meanwhile, the fact that $ D(\chi, \lambda_{0}) $ has an $ (n+1) $-multiple root leads to $ \hat{D}(\chi, \lambda_{0})= (\chi+b_{1})^{n-s+2}$. We derive $ 	\sum_{j=1}^{s} {a_{j}^{2}}=0 $, which is contradictory. Thus, we prove that $ b_{j} $'s are mutually different. In addition, the expression of $ a_{j} $ in Eq. \eqref{eq:ab2} can be obtained by substituting $ \chi=-b_{j} $ into Eqs. \eqref{eq:cp1} and \eqref{eq:cp2}.
	
	(Sufficiency $\Leftarrow$) It is evident that the degree of $ D(\chi, \lambda_{0})-(\chi-\chi_{0})^{n+1} $ is less than $ n $ and  \begin{equation}\label{eq:cpd3}
		D(-b_{j}, \lambda_{0})-(-b_{j}-\chi_{0})^{n+1}=0. 
	\end{equation}
	Since $ b_{i}\ne b_{j} (i\ne j) $, then we get $ D(\chi, \lambda_{0})=(\chi-\chi_{0})^{n+1} $.
	
	Therefore, we complete the proof.
\end{proof}

\subsection{\label{app10} The proof of Theorem \ref{th3}}

\begin{proof}
	First, let us establish the case with $  N_{l} \neq 0$ and $  N_{i} = 0 $ $ ( i \neq l, 1\leq l,i\leq n) $ in the integer vector $ \mathcal{N} $. 
	
	For solution formulae \eqref{eq:qn3} and \eqref{eqm1}, we take $ \phi_{i}=\phi(\varepsilon) (1\leq i\leq N)$ given in Eq. \eqref{phiep} with
	\begin{equation}\label{parac2}
		\mathbf{C}_{l}=\left( 1, \hat{\omega}_{n}^{n+1-l}, \hat{\omega}_{n}^{2(n+1-l)}, \ldots, \hat{\omega}_{n}^{n(n+1-l)}\right)^{T}, \quad \hat{\omega}_{n}=\ee^{\frac{2\pi \ii}{n+1}}.
	\end{equation}
	Moreover, by applying the expansions in Eq. \eqref{eq:zh} and the following expansion
\begin{widetext}
	\begin{equation}\label{eq35}
				\begin{aligned}
					\dfrac{1}{\chi(\varepsilon)-\chi^{*}(\hat{\varepsilon})} =& \dfrac{\chi_{0}^{*}-\chi_{0}}{(\chi(\varepsilon)-\chi_{0}^{*})(\chi^{*}(\hat{\varepsilon})-\chi_{0})}\, \dfrac{1}{1-\frac{(\chi(\varepsilon)-\chi_{0})(\chi^{*}(\hat{\varepsilon})-\chi_{0}^{*})}{(\chi(\varepsilon)-\chi_{0}^{*})(\chi^{*}(\hat{\varepsilon})-\chi_{0})}} \\
					=&\dfrac{\chi_{0}^{*}-\chi_{0}}{(\chi(\varepsilon)-\chi_{0}^{*})(\chi^{*}(\hat{\varepsilon})-\chi_{0})} \, \sum_{j=0}^{\infty}  \left( \frac{(\chi(\varepsilon)-\chi_{0})(\chi^{*}(\hat{\varepsilon})-\chi_{0}^{*})}{(\chi(\varepsilon)-\chi_{0}^{*})(\chi^{*}(\hat{\varepsilon})-\chi_{0})} \right)^{j},
				\end{aligned}
			\end{equation} 
\end{widetext}
	we denfine the functions $ \tau^{[0]}(\hat{\varepsilon}, \varepsilon ) $ and $ \tau^{[k]}(\hat{\varepsilon}, \varepsilon ) $ $ (1\leq k\leq n) $ with respect to $ \hat{\varepsilon} $ and $ \varepsilon $, as follows:
		\begin{equation}\label{tau11}
			\begin{aligned}
				&\tau^{[0]}(\hat{\varepsilon}, \varepsilon ) \\
				=&2\frac{\chi^{*}(\hat{\varepsilon})}{\chi(\varepsilon)-\chi^{*}(\hat{\varepsilon})} \ee^{{\zeta^{[l]}}^{*} (\hat{\varepsilon})+\zeta^{[l]}(\varepsilon)}\\
				=&C_{0}\sum_{i,j=0}^{\infty}\left(\sum_{\mu=0}^{\min{(i,j)}} C_{1}^{\mu} S_{i-\mu}(\mathbf{u}^{[l,0,-]}(\mu))  S_{j-\mu}(\mathbf{u}^{[l,0,+]}(\mu)) \right)\hat{\varepsilon}^{i}\varepsilon^{j}\\
				=&C_{0}  \sum_{i,j=0}^{\infty}\tau^{[0]}_{i,j} \hat{\varepsilon}^{i}\varepsilon^{j},
			\end{aligned}
		\end{equation}
		and
		\begin{equation}\label{tau12}
			\begin{aligned}
				&\tau^{[k]}(\hat{\varepsilon}, \varepsilon )\\
				=& 2\frac{\chi(\varepsilon)}{\chi(\varepsilon) -\chi^{*}(\hat{\varepsilon})}\frac{\chi^{*}(\hat{\varepsilon})+b_{k}}{\chi(\varepsilon)+b_{k}} \ee^{{\zeta^{[l]}}^{*} (\hat{\varepsilon}) +\zeta^{[l]}(\varepsilon)} \\
				=&C_{0}^{[k]}\sum_{i,j=0}^{\infty}\left(\sum_{\mu=0}^{\min{(i,j)}}C_{1}^{\mu} S_{i-\mu}(\mathbf{u}^{[l,k,-]}(\mu)) S_{i-\mu}(\mathbf{u}^{[l,k,+]}(\mu))\right) \hat{\varepsilon}^{i}\varepsilon^{j} \\ 
				=&C_{0}^{[k]} \sum_{i,j=0}^{\infty}\tau^{[k]}_{i,j} \hat{\varepsilon}^{i}\varepsilon^{j}, 
			\end{aligned}
		\end{equation}
	where $ C_{0}=\frac{\chi_{0}^{*}}{\chi_{0}-\chi_{0}^{*}} \ee^{(\zeta_{0}^{[l]})^{*}+\zeta_{0}^{[l]}} $, $ C_{0}^{[k]}=\frac{\chi_{0}^{*}\, }{\chi_{0}-\chi_{0}^{*}} \frac{\chi_{0}^{*}+b_{k}}{\chi_{0}+b_{k}} \ee^{(\zeta_{0}^{[l]})^{*}+\zeta_{0}^{[l]}} $, $ C_{1} $ is given in Theorem \ref{th3}, $\zeta^{[l]}(\varepsilon) $ is determined by Eq. \eqref{eq22b} with $ \lambda=\lambda(\varepsilon) $ and $ \chi_{j}=\chi(\varepsilon) $, the vectors $ \mathbf{u}^{[l,s,\pm]}(\mu)=[ u^{[l,s,\pm]}_{1},  u^{[s,\pm]}_{2}, \cdots] $ are defined by
	\begin{equation}\label{eq:xpi}
		\renewcommand{\arraystretch}{1.2}
		\begin{array}{l}
			\mathbf{u}^{[l,0,+]}(\mu)= \zeta^{[l]}+\mu\mathbf{h}_{2}+\mathbf{h}_{3}, \\
			\mathbf{u}^{[l,0,-]}(\mu)= {\zeta^{[l]}}^{*}+\mathbf{h}_{1}^{*}+\mu \mathbf{h}_{2}^{*}+\mathbf{h}_{3}^{*}, \\
			\mathbf{u}^{[l,k,+]}(\mu)= \zeta^{[l]}+\mathbf{h}_{1} +\mu\mathbf{h}_{2} +\mathbf{h}_{3}-{\mathbf{h}_{4}^{[k]}}, \\
			\mathbf{u}^{[l,k,-]}(\mu)= {\zeta^{[l]}}^{*} +\mu\mathbf{h}_{2}^{*} +\mathbf{h}_{3}^{*} +{(\mathbf{h}_{4}^{[k]})}^{*},
		\end{array}
	\end{equation}
	$ S_{i}(\mathbf{u}^{[l,s,\pm]}(\mu)) $, $ \zeta^{[l]} $ and $ \mathbf{h}_{j} $ $ (1\leq j\leq4) $ are defined by Eq. \eqref{eq:scpo} and \eqref{eq:zh}, respectively.
	Therefore, by employing Proposition \ref{pr4} and Eq. \eqref{tau11}-\eqref{tau12}, we can calculate Eq. \eqref{eqm1} with $ \chi_{i,p}^{*}=\chi^{*}(\hat{\varepsilon}\hat{\omega}_{n}^{p-1})$ and $ \chi_{j,r}= \chi(\varepsilon\hat{\omega}_{n}^{r-1}) $, then obtain
	\begin{equation}\label{tau2}
		\begin{aligned}
			M_{i,j}^{[s]}=&2\sum_{p,r=1}^{n+1} \left( (\hat{\omega}_{n}^{*})^{(p-1)(n+1-l)}\hat{\omega}_{n}^{(r-1)(n+1-l)} \right.\\ &\left. \times \tau^{[0]}(\hat{\varepsilon}(\hat{\omega}_{n}^{(p-1)})^{*}, \varepsilon\hat{\omega}_{n}^{r-1} ) \right)  \\
			=&2(n+1)^{2} C_{0} \sum_{i,j=1}^{\infty} \left(  \tau^{[0]}_{(n+1)(i-1)+l,(n+1)(j-1)+l} \right.\\ 
			&\left.\times (\hat{\varepsilon})^{(n+1)(i-1)+l} \varepsilon^{(n+1)(j-1)+l} \right), \quad 0\leq s\leq n.
		\end{aligned}
	\end{equation}

	Then, by taking the limits $ \hat{\varepsilon} \rightarrow 0 $ and $ \varepsilon \rightarrow 0 $, we obtain the rogue wave solutions of $ n $-DNLS equation \eqref{eq:nDNLS}, as follows:
	\begin{equation}\label{eq:qnrw2}
		q_{k}^{[\mathcal{N}]}= \left( -\frac{\ii {q}_{k}}{b_{k}}\frac{\det\left( \left( \tau^{[k]}_{(n+1)(i-1)+l,(n+1)(j-1)+l}\right) _{1\leq i,j\leq N} \right) }{\det\left( \left( \tau^{[0]}_{(n+1)(i-1)+l,(n+1)(j-1)+l}\right) _{1\leq i,j\leq N} \right)}\right)_{x}, 
	\end{equation}
	with $1\leq k\leq n $.
	
	Generally, for the arbitrary vector $ \mathcal{N}= [N_{1}, N_{2}, \cdots, N_{n}] $,  building upon the above proof, we assume $ \phi_{i}=\phi(\varepsilon) $ $ (1\leq i \leq N) $ with the constant column vectors \eqref{parac2} and
	
	\begin{equation}\label{parac}
				l=\left\lbrace 
				\begin{array}{ll}
					1, & 1\leq i \leq N_{1},\\
					2, & N_{1}+1\leq i \leq N_{1}+N_{2},\\
					&\vdots  \\
					n, & \sum_{r=1}^{n-1}N_{r} +1  \leq i \leq N.
				\end{array}
				\right. 
	\end{equation}

	Thus, we can similarly obtain the formula \eqref{eq:qnrw} of rogue wave solution $ \mathbf{q}^{[\mathcal{N}]} $.
	Finally, we complete the proof.

\end{proof}

\subsection{\label{app11} The proof of Theorem \ref{th4}}

\begin{proof}
	Based on the distinct asymptotic behaviors in the outer and inner regions as presented in the Theorem \ref{th4}, we will demonstrate the proof in the following two sections.
	
	\begin{enumerate}[1.]
		\item  In the outer region, where $\sqrt{{x}^{2}+t^{2}} = \mathcal{O}(\eta) $ with $ \eta=d_{m}^{1/m} $. First, we rewrite the determinant $ \det ({ \mathbf{M}^{[s]}}) $ \eqref{eq:qnmk3} by utilizing the Cauchy-Binet formula below:
		
			\begin{equation}\label{eq:mds}
				\begin{aligned}
					\det ({ \mathbf{M}^{[s]}})=& \sum_{0\leq \mu_{1}<\mu_{2}<\cdots<\mu_{N}\leq (n+1)N-1} \left[ C_{1}^{\sum_{k=1}^{N}\mu_{k}} \right. \\
					& \times  \det_{1\leq i,j \leq N} \left( S_{(n+1)(j-1)+l-\mu_{i}} (\mathbf{u}^{[l,s,+]}(\mu_{i})) \right)\\  
					&\times \left. \det_{1\leq i,j \leq N} \left(  S_{(n+1)(i-1)+l-\mu_{j}} (\mathbf{u}^{[l,s,-]}(\mu_{j})) \right) \right] , 
				\end{aligned}
			\end{equation}
		
		where $ 0\leq s\leq n, $ $ \mu_{i} $'s are integers, and the vectors $ 	\mathbf{u}^{[l,s,\pm]}(\mu_{i})$ are defined by Eq. \eqref{eq:xpi}. For convenience, we will omit the symbol $ l $ in the superscript of $ 	\mathbf{u}^{[l,s,\pm]}(\mu_{i})$ in the subsequent proof.
		Given that the order of $ d_{m} $ in $ {S_{i}} $ decreases as the subscript $ i $ decreases, we consider the first two highest order term of $ d_{m} $ by two different indexes, in which one is $ (\mu_{1}, \mu_{2}, \cdots, \mu_{N}) = (0, 1, \cdots, N-1) $, and another is $ (\mu_{1}, \mu_{2}, \cdots, \mu_{N}) = (0, 1, \cdots, N-2, N) $.
		
		Then, suppose $ |d_{m}| $ is large enough and the rest parameters are $ \cO(1) $ in the $ l $-type rogue wave solutions of the $ n $-DNLS equation \eqref{eq:nDNLS}. Since $ \zeta_{1}(x,t)=\cO(\eta) $, $ \mathbf{h}_{1}, \mathbf{h}_{2}, \mathbf{h}_{3},$ and $ \mathbf{h}_{4} $ are all constant vectors, by utilizing the Schur polynomials \eqref{eq:scpo}, we deduce 
		\begin{equation}\label{eq:sasy}
			\begin{aligned}
				&S_{i}\left( \mathbf{u}^{[s,\pm]}(\mu_{k}) \right) \\
				=& S_{i}\left( u^{[s,\pm]}_{1}, u^{[s,\pm]}_{2}, u^{[s,\pm]}_{3}, \ldots  \right)  \\
				=&   \eta^{i} S_{i}\left( u^{[s,\pm]}_{1}\eta^{-1}, u^{[s,\pm]}_{2}\eta^{-2}, \ldots, u^{[s,\pm]}_{m}\eta^{-m}, \ldots, \right)     \\
				\sim&   \eta^{i} S_{i}\left(\zeta_{1}\eta^{-1}, 0,  \ldots,0, d_{m}\eta^{-m}, 0,\ldots \right) \\ 
				=&   S_{i}(\mathbf{v}_{1}^{\pm}),
			\end{aligned}
		\end{equation}
		where  $ \mathbf{v}_{1}^{+}= (\zeta_{1},  0, \ldots, 0, d_{m},0,\ldots) $ and $ \mathbf{v}_{1}^{-} = (\mathbf{v}_{1}^{+})^{*}$. Moreover, it is readily found the fact
		\begin{equation}\label{eq:sv}
			S_{i}(\mathbf{v}_{1}^{+})\, =\, \eta^{i}\, p^{[m]}_{i}(z),
		\end{equation} 
		where $ p^{[2]}_{i}(z) $ is defined by Eq. \eqref{eq:p1}, and 
		\begin{equation}\label{eq:sz}
			\begin{aligned}
				z=\, \eta^{-1}\zeta_{1} 
				=\,  \ii \eta^{-1}\chi^{[1]} (x+(2\chi_{0}+\|\mathbf{a}\|^{2}_{2})t).\\
			\end{aligned}
		\end{equation}
		When choosing the index choices $ \mu_{i}=i-1 (1\leq i\leq N)$ in the expression \eqref{eq:mds} of $ \det ({ \mathbf{M}^{[s]}}) $, we obtain the highest $ t $-power term
		\begin{equation}\label{eq:lead}
			\begin{aligned}
				&C_{1}^{\frac{N(N-1)}{2}} \det_{1\leq i,j \leq N} \left( S_{(n+1)(j-1)+l-i+1} (\mathbf{u}^{[s,+]}(\mu_{i})) \right)\\ 
				&\qquad \times \det_{1\leq i,j \leq N} \left( S_{(n+1)(i-1)+l-j+1} (\mathbf{u}^{[s,-]}(\mu_{j})) \right).
			\end{aligned}
		\end{equation}
		which exerts the most significant influence on the rogue wave structure. Thus, when $ |d_{m}|\gg 1 $, by applying the above formulas \eqref{eq:sasy} and \eqref{eq:sv}, the leading-order term about $ d_{m} $ in $ \det ({ \mathbf{M}^{[s]}}) $, represented by the expression \eqref{eq:lead}, exerts the following approximation:
		\begin{equation}\label{eq:ms1}
		\quad	\det ({ \mathbf{M}^{[s]}}) \sim |C_{1}|^{2} \left| \eta\right|^{2\Gamma} \left|{c_{N}^{[m,n+1,l]}}^{-1}\,	W_{N}^{[2,n+1,l]}(z) \right|^{2},
		\end{equation}
		where $ \Gamma $ is given in \eqref{eq:gamma} with $ k=n+1 $.
		
		Substituting Eq. \eqref{eq:ms1} into the formula \eqref{eq:qnrw}, we determine that the vector rogue wave solutions $ \mathbf{q}^{[\mathcal{N}_{l}]}(x,t) $ trend toward the plane wave background $ \mathbf{{q}}(x,t) $ as $ |d_{m}|\rightarrow \infty $, except for some locations at or near $ {({x}_{0}, t_{0})} $ in the specific region with $\sqrt{x^{2}+t^{2}} = \mathcal{O}(\eta) $. Here, $ (x_{0}, t_{0}) $ satisfies the following equation
		\begin{equation}\label{eq:z01}
			\begin{aligned}
				{z}_{0}= \,\ii&\eta^{-1}({x}_{0} +(2\chi_{0}+\|\mathbf{a}\|^{2}_{2})t_{0}),
			\end{aligned}
		\end{equation}
		with $ 	{z}_{0}$ being a nonzero root of the generalized Wronskian-{Hermite} polynomial $ W_{N}^{[m,n+1,l]}(z) $.
		
		Furthermore, in order to get the asymptotic properties of the solution near the point $ (x,t) ={({x}_{0}, t_{0})} $, a more refined asymptotic analysis is required. Now, taking a coordinate transformation
		\begin{equation}\label{eq:ct}
			\hat{x}=x-\hat{x}_{0}|\eta|, \quad \hat{t}=t-\hat{t}_{0}|\eta|, \quad |\eta|=\eta\ee^{-\ii \arg \eta},
		\end{equation}
		we have the following expansion
	\begin{widetext}
		\begin{equation}\label{eq:scct}
		\begin{aligned}
			&\sum_{k=0}^{\infty}S_{k}(\mathbf{u}^{[s,+]}(\mu_{i}))(\eta^{-1}\varepsilon)^{k}\\
			=& \exp\left( ({u_{1}^{[s,+]}({x},{t})}) \eta^{-1}\varepsilon +({u_{2}^{[s,+]}({x},{t})}) (\eta^{-1}\varepsilon)^{2}+ \cdots  +({u_{m}^{[s,+]}({x},{t})}) (\eta^{-1}\varepsilon)^{m}+\cdots \right) \\
			=&\exp\left[ \left( (u_{1}^{[s,+]}(\hat{x},\hat{t}))\eta^{-1}   +u_{1}^{[s,+]}(\hat{x}_{0},\hat{t}_{0})\ee^{-{\ii} \arg \eta} \right)\varepsilon  +\left( (u_{2}^{[s,+]}(\hat{x},\hat{t}))\eta^{-2}   +u_{2}^{[s,+]}(\hat{x}_{0},\hat{t}_{0})\ee^{-{\ii} \arg \eta}\eta^{-1}  \right) \varepsilon^{2} +\cdots \right.\\
			&\left.  +\left( (u_{m}^{[s,+]}(\hat{x},\hat{t}))\eta^{-m}   +u_{m}^{[s,+]}(\hat{x}_{0},\hat{t}_{0})\ee^{-{\ii} \arg \eta}\eta^{-m+1}  \right) \varepsilon^{m} +\cdots \right] \\
			=&\exp\left( u_{1}^{[s,+]}(\hat{x}_{0},\hat{t}_{0}) \ee^{-{\ii} \arg \eta}\varepsilon +\varepsilon^{m} \right) 
			\exp\left( u_{1}^{[s,+]}(\hat{x},\hat{t})\eta^{-1}\varepsilon  + u_{2}^{[s,+]}(\hat{x}_{0},\hat{t}_{0})\ee^{-{\ii} \arg \eta}\eta^{-1}  \varepsilon^{2} +\cO(|\eta|^{-2}) \right)
			\\
			=&\left( \sum_{i=0}^{\infty}p_{i}^{[m]}(z_{0}) \varepsilon^{i}\right) \left( 1+u_{1}^{[s,+]}(\hat{x},\hat{t})\eta^{-1}\varepsilon + u_{2}^{[s,+]}(\hat{x}_{0},\hat{t}_{0})  \ee^{-{\ii}\arg \eta}\eta^{-1} \varepsilon^{2} +\cO(|\eta|^{-2}) \right), 
		\end{aligned}
		\end{equation}
	\end{widetext}
		where $ z_{0}=u_{1}^{[s,+]}(\hat{x}_{0},\hat{t}_{0})\ee^{-{\ii} \arg \eta} $, and  $ \arg \eta $ denotes the principal value of the argument of $ \eta $. Then we derive the following approximate expression
		\begin{equation}\label{eq:scap}
			\begin{aligned}
				&S_{k}\left( \mathbf{u}^{[s,+]}(\mu_{i}) \right) = \eta^{k}  \left[ p_{k}^{[m]}(z_{0}) +\eta^{-1}\left(u_{1}^{[s,+]}(\hat{x},\hat{t}) p_{k-1}^{[m]}(z_{0}) \right.\right.\\
				&\qquad \left. \left. + u_{2}^{[s,+]}(\hat{x}_{0},\hat{t}_{0})\ee^{-{\ii} \arg \eta}p_{k-2}^{[m]}(z_{0}) \right)\right]    \left[ 1+ \mathcal{O}(|\eta|^{-2}) \right].  
			\end{aligned}
		\end{equation}
		
		When selecting the first index $ \mu_{i}=i-1 (1\leq i\leq N)$ in \eqref{eq:mds}, with the help of Eq. \eqref{eq:scap} and $ W_{N}^{[2,n+1,l]}(z_{0})=0 $, we can deduce the part containing $ S(\mathbf{u}^{[s,+]}{(\mu_{i})}) $ of the dominant term of $ d_{m} $ in the determinant of Eq. \eqref{eq:lead}, as follows:
		\begin{equation}\label{eq:msx1}
			\begin{aligned}
				&(c_{N}^{[m,n+1,1]})^{-1}\eta^{\Gamma-1}\left( u_{1}^{[s,+]}(\hat{x},\hat{t}) (W_{N}^{{[m,n+1,l]}})'(z_{0}) \right.\\
				&\left. + u_{2}^{[s,+]}(\hat{x}_{0},\hat{t}_{0})\ee^{-{\ii} \arg \eta} W_{N,2}^{{[m,n+1,l]}}(z_{0}) \right) \left[ 1+ \mathcal{O}(|\eta|^{-1}) \right],
			\end{aligned}
		\end{equation}
		where 
		\begin{widetext}
			\begin{equation}\label{eq:wn2}
				\begin{aligned}
					W^{[m,n+1,l]}_{N,2}(z_{0})
					=& c_{N}^{[m,n+1,1]} 
					\left[ \det_{1\leq j \leq N}\left( p^{[m]}_{(n+1)(j-1)+l}(z_{0}), 	\cdots,  p^{[m]}_{(n+1)(j-1)+l-N}(z_{0}), p^{[m]}_{(n+1)(j-1)+l-N+1}(z_{0}) \right)^{T} \right.  \\
					&\left. + \det_{1\leq j \leq N}\left( p^{[m]}_{(n+1)(j-1)+l}(z_{0}), 	\cdots, p^{[m]}_{(n+1)(j-1)+l-N+2}(z_{0}), p^{[m]}_{(n+1)(j-1)+l-N-1}(z_{0}) \right)^{T} \right] .
				\end{aligned}
			\end{equation}
		\end{widetext}
		The polynomial $ W^{[m,n+1,l]}_{N,2}(z_{0})$ represents a summation of $ 2 $ variations of $ W^{[m,n+1,l]}_{N} (z_{0}) $, where each variation {is} to subtract $ 2 $ from the subscript indices of all elements of a certain row in the last $ 2 $ rows of the determinant $ W^{[m,n+1,l]}_{N}(z_{0}) $. Similarly, the part containing $S(\mathbf{u}^{[s,-]}{(\mu_{j})}) $ in Eq. \eqref{eq:lead} can also be obtained, which is omitted here.
		
		On the other hand, when choosing another index $ (0,1,\ldots,N-2,N) $, we can calculate the highest power term of $ d_{m} $  containing $ S(\mathbf{u}^{[s,+]}(\mu_{i}) ) $ in Eq. \eqref{eq:lead} by using Eqs. \eqref{eq:ct}-\eqref{eq:scap}, as follows:
		\begin{widetext}
		\begin{small}
			\begin{equation}\label{eq:msx2}
			\begin{aligned}
				&\det_{1\leq j \leq N} \left(  S_{(n+1)(j-1)+l}(\mathbf{u}^{[s,+]}(\mu_{1}) ),   S_{(n+1)(j-1)+l-1}(\mathbf{u}^{[s,+]}(\mu_{2}) ), \cdots,    S_{(n+1)(j-1)+l-(N-2)}(\mathbf{u}^{[s,+]} (\mu_{N-1})), S_{(n+1)(j-1)+l-N}(\mathbf{u}^{[s,+]}(\mu_{N}) ) \right)^{T} \\
				&=\det_{1\leq j \leq N} \left(  p^{[m]}_{(n+1)(j-1)+l}(z_{0}),  p^{[m]}_{(n+1)(j-1)+l-1}(z_{0}), \ldots,  p^{[m]}_{(n+1)(j-1)+l-(N-2)}(z_{0}),  p^{[m]}_{(n+1)(j-1)+l-N}(z_{0}) \right)^{T}  C_{1}^{1/2} \, d_{m}^{\Gamma-1}  \left[ 1+ \mathcal{O}(|\eta|^{-1}) \right]\\
				&=C_{1}^{1/2}\eta^{\Gamma-1}(W_{N}^{[m,n+1,l]})'(z_{0}) \, \left[ 1+ \mathcal{O}(|\eta|^{-1}) \right].
			\end{aligned}
		\end{equation} 
		\end{small}
		\end{widetext}
		Moreover, its conjugate part involved $ S(\mathbf{x}^{[s,-]})$ are similarly obtained, which is omitted here.
		
		Therefore, for $ |\eta|\gg 1 $, referring to the above results \eqref{eq:msx1}-\eqref{eq:msx2}, the dominant term of $ d_{m} $ in $ \det ({ \mathbf{M}^{[s]}}) $ can be expressed as 
		\begin{widetext}
			\begin{equation}\label{eq:msap1}
				\begin{aligned}
					\det ({ \mathbf{M}^{[0]}})=\, &  C_{1}^{N(N-1)/2} \left| \eta \right|^{2\Gamma-2}    
					\left[ \left(\zeta_{1}(\hat{x},\hat{t}) +h_{3,1}+\Delta^{[k]}(z_{0}) \right) \left(\zeta_{1}(\hat{x},\hat{t}) +h_{1,1}+h_{3,1} +\Delta^{[k]}(z_{0}) \right)^{*} +C_{1}  \right] \\
					&\times \left| \left[(c_{N}^{[m,n+1,l]})^{-1} W_{N}^{[m,n+1,l]} \right]'(z_{0})\right|^{2}    \left[ 1+ \mathcal{O}(|d_{m}|^{-1/m}) \right],\\
					\det ({ \mathbf{M}^{[k]}})=\, &    
					C_{1}^{N(N-1)/2}\left| \eta \right|^{2\Gamma-2} 
					\left[ \left(\zeta_{1}(\hat{x},\hat{t}) +h_{1,1}+h_{3,1}-h_{4,1}^{[k]} +\Delta^{[k]}(z_{0}) \right) \left(\zeta_{1}(\hat{x},\hat{t}) +h_{3,1}+h_{4,1}^{[k]} +\Delta^{[k]}(z_{0}) \right)^{*} +C_{1}  \right]  \\
					&\times   \left| \left[(c_{N}^{[m,n+1,l]})^{-1} W_{N}^{[m,n+1,l]} \right]'(z_{0})\right|^{2}  \left[ 1+ \mathcal{O}(|d_{m}|^{-1/m}) \right], \quad 1\leq k\leq n,
				\end{aligned}
			\end{equation}
		\end{widetext}
		where $ \zeta_{1}(\hat{x},\hat{t})=\ii \chi^{[1]}\left(\hat{x}+(2\chi_{0}+\|\mathbf{a}\|^{2}_{2})\hat{t} \right), h_{1,1}= \frac{\chi^{[1]}}{\chi_{0}}, h_{2,1}=0, h_{3,1}=\frac{\ii \chi^{[1]}}{2\Im(\chi_{0})}, h_{4,1}^{[k]}=\frac{\chi^{[1]}}{\chi_{0}+b_{k}}  $, $ \chi^{[1]} $ is determined in Theorem \ref{th3}, and $ \Delta(z_{0})= \cO(1) $ is a complex constant defined by
		\begin{equation}\label{eq:del1}
			\begin{aligned}
				\Delta^{[s]}(z_{0})= &{ \hat{x}_{2}^{[s,+]}(\hat{x}_{0},\hat{t}_{0}) \ee^{-{\ii} \arg \eta}} \frac{W_{N,2}^{[m,n+1,l]}(z_{0})}{\left[ W_{N}^{[m,n+1,l]} \right]'(z_{0})}. 
			\end{aligned}
		\end{equation}
		Hence, by absorbing $ \Delta(z_{0}) $ into $ (\hat{x}_{0}, \hat{t}_{0}) $, we yield
		\begin{widetext}
			\begin{equation}\label{eq:msap11}
				\begin{aligned}
					\det ({ \mathbf{M}^{[0]}})=\, &   C_{1}^{N(N-1)/2} \left| \eta \right|^{2\Gamma-2}  \left[ \left(\zeta_{1}(x-x_{0},t-t_{0}) +h_{3,1} \right) \left(\zeta_{1}(x-x_{0},t-t_{0}) +h_{1,1}+h_{3,1} \right)^{*} +C_{1}  \right]  \\
					&\times  \left| \left[(c_{N}^{[m,n+1,l]})^{-1} W_{N}^{[m,n+1,l]} \right]'(z_{0})\right|^{2}  \left[ 1+ \mathcal{O}(|d_{m}|^{-1/m}) \right], \\
					\det ({ \mathbf{M}^{[k]}})=\, &  C_{1}^{N(N-1)/2}  \left| \eta \right|^{2\Gamma-2}   \left[ \left(\zeta_{1}(x-x_{0},t-t_{0}) +h_{1,1}+h_{3,1}-h_{4,1}^{[k]} \right) \left(\zeta_{1}(x-x_{0},t-t_{0}) +h_{3,1}+h_{4,1}^{[k]}  \right)^{*} +C_{1}  \right]  \\
					&\times  \left| \left[(c_{N}^{[m,n+1,l]})^{-1} W_{N}^{[m,n+1,l]} \right]'(z_{0})\right|^{2}  \left[ 1+ \mathcal{O}(|d_{m}|^{-1/m}) \right], \quad 1\leq k\leq n,
				\end{aligned}
			\end{equation}
		\end{widetext}
		where
		\begin{equation}\label{eq:x01}
			\begin{aligned}
				x_{0}\,&=\hat{x}_{0}|\eta| + \mathcal{O}(1)\\
				&=\,\Im \left[\frac{z_{0}d_{m}^{1/m} }{\chi^{[1]}}\right]+ \frac{2\Re(\chi_{0})+\|\mathbf{a}\|^{2}_{2} }{2\Im(\chi_{0})} \Re\left[\frac{z_{0}d_{m}^{1/m} }{\chi^{[1]}}\right]  + \mathcal{O}(1),
			\end{aligned}
		\end{equation}
		\begin{equation}\label{eq:t01}
			\begin{aligned}
				t_{0}\,&=\hat{t}_{0}|\eta| + \mathcal{O}(1) \\
				&=\,-\frac{1}{2\Im(\chi_{0})}\Re \left[\frac{z_{0}d_{m}^{1/m} }{\chi^{[1]}}\right] + \mathcal{O}(1).
			\end{aligned}
		\end{equation}
		This demonstrates Eq. \eqref{eq:xt0} of Theorem \ref{th4}. And the asymptotics \eqref{eq:qko} is proved by substituting the above Eq. \eqref{eq:msap11} into the formula \eqref{eq:qnrw}.
		
		\item In the inner region with $\sqrt{x^{2}+t^{2}} =  \mathcal{O}(1)$, when $ d_{m}\rightarrow \infty $, to obtaining the asymptotics of the rogue wave solution, we need to calculate the coefficient of the highest-order term of $ d_{m} $. To this end, we expand the elements of $ \det ({ \mathbf{M}^{[s]}}) $ \eqref{eq:qnmk3} by employing the following formula
		\begin{equation}\label{eq:mre}
			\begin{aligned}
			&S_{k}\left( \mathbf{{u}}^{[s,+]}(\mu) \right) = \sum_{j=0}^{\lfloor k/m\rfloor}\frac{d_{m}^{j}}{j!}S_{k-jm} \left( \mathbf{\hat{u}}^{[s,+]}(\mu)\right), \\ 
			&S_{k}\left( \mathbf{{u}}^{[s,-]}(\mu) \right) = \sum_{i=0}^{\lfloor k/m\rfloor}\frac{(d_{m}^{*})^{i}}{i!}S_{k-im} \left( \mathbf{\hat{u}}^{[s,-]}(\mu)\right),
			\end{aligned}
		\end{equation}
		where $\mathbf{\hat{u}}^{[s,+]}(\mu)=\mathbf{{u}}^{[s,+]}(\mu) -d_{m}\mathbf{\hat{e}}_{m} $,  $ \mathbf{\hat{u}}^{[s,-]}(\mu)=\mathbf{{u}}^{[s,-]}(\mu) -d_{m}^{*}\mathbf{\hat{e}}_{m} $, and $ \mathbf{\hat{e}}_{m} $ represents the infinite-dimensional standard unit vector with the $ m $-th component equal to $ 1 $.
		
		Next, we will utilize the same method in Theorem \ref{th:whp} to obtain the highest-order terms of $ d_{m} $ in the expansion of determinant $ \det ({ \mathbf{M}^{[s]}}) $ \eqref{eq:qnmk3}. Given the variations in parameter values, our results are subject to diversity. To account for this, we present a clear discussion below. 
	\begin{itemize}
		\item[(1).] When $ (n+1,m) = 1  $, for obtain the highest-order term of $ d_{m} $ in determinant \eqref{eq:qnmk3}, we can employ the same technique in Theorem \ref{th:whp} to simplify the submatrices $ M_{l}^{[s,+]} $ in \eqref{eq:qnmk3} as an upper triangular block matrix
		\begin{equation}\label{eq:mrb1}
		\begin{pmatrix}
			\mathbf{{B}}_{1,1} & \mathbf{{B}}_{1,2}\\
			\mathbf{0}_{(nN+\tilde{N})\times(N-\tilde{N}) } & \mathbf{{B}}_{2,2}
		\end{pmatrix}_{(n+1)N\times N},
		\end{equation}
		where $ \tilde{N}=\sum_{i=1}^{n}N_{l,i} $, the values of $ N_{l,i} $ refers to the value of $ N_{i} $ against $ l\in\{1,2, \ldots,n \} $ given in Theorem \ref{th:whp}, and
		\begin{equation}\label{mb11}
			\qquad	\mathbf{{B}}_{1,1}=
		\begin{pmatrix}
				S_{0} & S_{1} &\cdots & S_{N-\tilde{N}-1} \\
				0 & S_{0} &\cdots & S_{N-\tilde{N}-2} \\
				0 & 0& \cdots & S_{0}
		\end{pmatrix}_{(N-\tilde{N})\times (N-\tilde{N})},
		\end{equation}
\begin{widetext}
		\begin{equation}\label{mb12}
		\mathbf{{B}}_{1,2}=
		\begin{pmatrix}
			S_{N-\tilde{N}+1} & S_{N-\tilde{N}} &  \cdots & S_{2}\\
			S_{N-\tilde{N}+n+2} & S_{N-\tilde{N}+n+1} &  \cdots & S_{n+3}\\
			\vdots & \vdots & \ddots & \vdots\\
			S_{N-\tilde{N}+1+(n+1)(N_{l,1}-1)} & S_{N-\tilde{N}+(n+1)(N_{l,1}-1)} & \cdots & S_{(n+1)(N_{l,1}-1)+2}\\
			S_{N-\tilde{N}+2} & S_{N-\tilde{N}+1} & \cdots & S_{3}\\
			S_{N-\tilde{N}+n+3} &  S_{N-\tilde{N}+n+2} & \cdots & S_{n+4}\\
			\vdots & \vdots & \ddots & \vdots\\
			S_{N-\tilde{N}+2+(n+1)(N_{l,2}-1)} & S_{N-\tilde{N}+1+(n+1)(N_{l,2}-1)} &  \cdots & S_{(n+1)(N_{l,1}-1)+3}\\
			\vdots & \vdots & \ddots & \vdots\\
			S_{N-\tilde{N}+n} & S_{N-\tilde{N}+n-1} & \cdots & S_{n+1} \\
			S_{N-\tilde{N}+2n+1} & S_{N-\tilde{N}+2n} & \cdots & S_{2n+2} \\
			\vdots & \vdots & \ddots & \vdots\\
			S_{N-\tilde{N}+n+(n+1)(N_{l,n}-1)} & S_{N-\tilde{N}+n+(n+1)(N_{l,n}-1)-1} & \cdots & S_{n+(n+1)(N_{l,n}-1)+1}\\
		\end{pmatrix}^{T}_{\tilde{N}\times (N-\tilde{N})},
		\end{equation}	
		\begin{equation}\label{mb22}
		\mathbf{{B}}_{2,2}=
		\begin{pmatrix}
			S_{1} & 0 & 0 & \cdots \\
			S_{n+2} & S_{n+1} &  S_{n} &\cdots \\
			\vdots & \vdots & \vdots & \ddots \\
			S_{1+(n+1)(N_{l,1}-1)} & S_{(n+1)(N_{l,1}-1)} & S_{(n+1)(N_{l,1}-1)-1} & \cdots \\
			S_{2} & S_{1}& S_{0} & \cdots \\
			S_{n+3} &  S_{n+2} & S_{n+1} & \cdots \\
			\vdots & \vdots & \ddots & \vdots\\
			S_{2+(n+1)(N_{l,2}-1)} & S_{1+(n+1)(N_{l,2}-1)} & S_{(n+1)(N_{l,2}-1)} &  \cdots \\
			\vdots & \vdots & \vdots & \ddots \\
			S_{n} & S_{n-1} & S_{n-2} & \cdots \\
			S_{2n+1} & S_{2n} & S_{2n-1} & \cdots \\
			\vdots & \vdots & \vdots& \ddots \\
			S_{n+(n+1)(N_{l,n}-1)} & S_{n+(n+1)(N_{l,n}-1)-1} & S_{n+(n+1)(N_{l,n}-1)-2} & \cdots \\
		\end{pmatrix}^{T}_{\tilde{N}\times  (nN+\tilde{N})},
		\end{equation}	
\end{widetext}
		where $ S_{k}=S_{k}(\hat{u}^{[s,+]}(\mu)) $, and the constant coefficients of $ S_{k} $ in the matrices \eqref{mb11}-\eqref{mb22} are omitted. Moreover, these elements in every column of the above matrices are the coefficients of the highest power term of $ d_{m} $. Thus, for the highest power of $ d_{m} $, we just consider the matrix $ \mathbf{{B}}_{2,2} $. For the submatrices $ M_{l}^{[s,-]} $ of \eqref{eq:qnmk3}, we take the same simplification procedure.  Then, the determinants $ \det ({ \mathbf{M}^{[s]}}) $ \eqref{eq:qnmk3} can be further simplified as the following form
			\begin{equation}\label{eq:mre2}
				\begin{aligned}
					\qquad\quad	\det ({ \mathbf{M}^{[s]}})= C_{2}|d_{m}|^{C_{3}} 
					\begin{vmatrix} 
						\mathbf{0}_{N\times N} & -\hat{M}_{l}^{[s,-]} \\
						\hat{M}_{l}^{[s,+]} & \mathbb{I}_{\hat{N}\times \hat{N}}
					\end{vmatrix}\left[ 1+\cO(d_{m}^{-1}) \right],
				\end{aligned}
			\end{equation}
			where $ C_{2}\ne 0 $, $ C_{3}>0 $ are constants, $  0\leq s\leq n $, $ \hat{N}=\max\{ (n+1)N_{l,i}+i-n \} $,
			\begin{equation}\label{eq:mrep1}
				\begin{aligned}
					\quad	\qquad	&\hat{M}_{l}^{[s,+]}= \left( \hat{M}_{\hat{N}\times N_{l,1}}^{[s,1,+]}, \hat{M}_{\hat{N}\times N_{l,2}}^{[s,2,+]}, \cdots, \hat{M}_{\hat{N}\times N_{l,n}}^{[s,n,+]} \right), \\ &\hat{M}_{l}^{[s,-]} = \left( (\hat{M}_{N_{l,1}\times\hat{N} }^{[s,1,-]})^{T}, (\hat{M}_{N_{l,2}\times\hat{N} }^{[s,2,-]})^{T}, \cdots, (\hat{M}_{N_{l,n}\times\hat{N} }^{[s,n,-]})^{T} \right)^{T},\\
					&\hat{M}_{\hat{N}\times N_{l,v}}^{[s,v,+]}= \left( \hat{\tau}_{i,j}^{[s,v,+]} \right)_{\substack{1\leq i \leq \hat{N} \\ 1\leq j \leq N_{l,v}} },  \\ 
					&\hat{M}_{ N_{l,v}\times \hat{N}}^{[s,v,-]}= \left( \hat{\tau}_{i,j}^{[s,v,-]} \right)_{\substack{1\leq i \leq N_{l,v} \\ 1\leq j \leq \hat{N} }}, \\
					&\hat{\tau}_{i,j}^{[s,v,+]}=C_{1}^{\frac{i-1}{2}}S_{(n+1)(j-1)+l-i+1} (\mathbf{\hat{u}}^{[i,+]}+ (i-1+\mu_{0}) \mathbf{h}_{2}), \\
					&\hat{\tau}_{i,j}^{[s,v,-]}= C_{1}^{\frac{j-1}{2}}S_{(n+1)(i-1)+l-j+1} (\mathbf{\hat{u}}^{[j,-]}+ (j-1+\mu_{0}) \mathbf{h}_{2}^{*}),  \\
					& \mu_{0}=N-\sum_{i=1}^{n}N_{l,i},\quad 1\leq v\leq n,
				\end{aligned}
			\end{equation}
			and the values of $ N_{l,i} $ refers to the value of $ N_{i} $ against $ l\in\{1,2, \ldots,n \} $ given in Theorem \ref{th:whp}. 
			
			Since the constant term $ C_{2}|d_{m}|^{C_{3}} $ in Eq. \eqref{eq:mre2} does not affect the rogue wave structure, we can further rewrite the determinant \eqref{eq:mre2} as the form \eqref{eq:qnmk} with the index $ {\mathcal{\hat{N}}}_{l} = \sum_{i=1}^{n} N_{l,i}\mathbf{e}_{i} $, and
			\begin{equation}\label{eq:tau4}
				\begin{aligned}
					\qquad	\tau^{[s]}_{i,j}=&\sum_{\mu=0}^{\min{(i,j)}} \left( C_{1}^{\mu} \, S_{i-\mu}(\mathbf{\hat{u}}^{[j,-]}(\mu)+ (N-\sum_{i=1}^{n}N_{l,i}) \mathbf{h}_{2}^{*}) \right.  \\
					&\left. \quad\times S_{j-\mu}(\mathbf{\hat{u}}^{[i,+]}(\mu) + (N-\sum_{i=1}^{n}N_{l,i}) \mathbf{h}_{2}) \right). \\
				\end{aligned}
			\end{equation}
			Hence, we conclude that the higher-order rogue wave can be asymptotically reduced to a $ \mathcal{\hat{N}}_{l} $-order rogue wave in this inner region with the above constraint conditions, and its  the approximation error is $ \cO(d_{m}^{-1}) $. Moreover, these internal parameters $ \hat{d}_{j}^{[k]} $ are calculated below
			\begin{equation}\label{eq:ds}
				\qquad	\hat{d}^{[k]}_{j}=\left\lbrace 
				\begin{aligned}
					&d_{j}+ \mu_{0} h_{2,j},  \quad j\ne m, \quad\\
					&0, \qquad j=m,
				\end{aligned}
				\right. 
			\end{equation}
			where $ \mu_{0} $ is given in Eq. \eqref{eq:mrep1}, $  j\geq 1 $, and $  1\leq k\leq n $.

			\item[(2).] When $ (n+1,m) \ne 1  $, we similarly reduce the submatrices $ M_{l}^{[s,\pm]} $ of Eq. \eqref{eq:qnmk3}. Now, refering to the Theorem \ref{th:whp}, the determinants \eqref{eq:qnmk3} can also be simplified to the form of Eq. \eqref{eq:mre2} with the parameter $ \mu_{0} $ in Eq. \eqref{eq:mrep1} replaced with
			\begin{equation}\label{eq:mu02}
				\mu_{0}=\left\lbrace 
				\begin{array}{lll}
					1, &  (n+1,m)\mid l, \\
					0, &  (n+1,m)\nmid l,
				\end{array}
				\right. 
			\end{equation}
			where $ (n+1, m) $ represents the greatest common divisor of $ n+1 $ and $ m $. Now, the index $ N_{l,i} $ refers to the value of $ N_{i} $ against $ l\in\{1,2, \ldots,n \} $ given in Theorem \ref{th:whp} under the condition of $ (n+1,m) \ne 1  $.
			Subsequently, for this scenario, by substituting the simplified determinants into the formula \eqref{eq:qnrw} for higher-order rogue wave solutions, it becomes evident that as $ d _{m}\rightarrow\infty$, the vector solutions $ \mathbf{q}^{[\mathcal{N}_{l}]} $ gradually approach the lower-order vector rogue wave solutions $ \mathbf{q}^{[\mathcal{\hat{N}}_{l}]} $ with the approximate error $ \cO(d_{m}^{-1}) $ in this inner region. Additionally, the lower-order vector rogue wave solutions have the internal parameters
			\begin{equation}\label{eq:ds2}
				\hat{d}^{[k]}_{j}=\left\lbrace 
				\begin{aligned}
					&d_{j}+ \mu_{0} h_{2,j},  \quad j\ne m, \quad\\
					&0, \qquad j=m,
				\end{aligned}
				\right. 
			\end{equation}
			where $ \mu_{0} $ is given in Eq. \eqref{eq:mu02}, $  j\geq 1, $ and $  1\leq k\leq n $.

		\end{itemize}

	\end{enumerate}
	
	Thus, we complete the proof of Theorem \ref{th4}.
\end{proof}

\section*{ Conflict of interests}
The authors have no conflicts to disclose.

\section*{DATA AVAILABILITY}
Data sharing is not applicable to this article as no new data were created or analyzed in this study.

\section*{Acknowledgments}
Liming Ling is supported by the National Natural Science Foundation of China (No. 12122105).



\bibliography{Ref_for_n-DNLS}

\begin{thebibliography}{29}%
\makeatletter
\providecommand \@ifxundefined [1]{%
 \@ifx{#1\undefined}
}%
\providecommand \@ifnum [1]{%
 \ifnum #1\expandafter \@firstoftwo
 \else \expandafter \@secondoftwo
 \fi
}%
\providecommand \@ifx [1]{%
 \ifx #1\expandafter \@firstoftwo
 \else \expandafter \@secondoftwo
 \fi
}%
\providecommand \natexlab [1]{#1}%
\providecommand \enquote  [1]{``#1''}%
\providecommand \bibnamefont  [1]{#1}%
\providecommand \bibfnamefont [1]{#1}%
\providecommand \citenamefont [1]{#1}%
\providecommand \href@noop [0]{\@secondoftwo}%
\providecommand \href [0]{\begingroup \@sanitize@url \@href}%
\providecommand \@href[1]{\@@startlink{#1}\@@href}%
\providecommand \@@href[1]{\endgroup#1\@@endlink}%
\providecommand \@sanitize@url [0]{\catcode `\\12\catcode `\$12\catcode
  `\&12\catcode `\#12\catcode `\^12\catcode `\_12\catcode `\%12\relax}%
\providecommand \@@startlink[1]{}%
\providecommand \@@endlink[0]{}%
\providecommand \url  [0]{\begingroup\@sanitize@url \@url }%
\providecommand \@url [1]{\endgroup\@href {#1}{\urlprefix }}%
\providecommand \urlprefix  [0]{URL }%
\providecommand \Eprint [0]{\href }%
\providecommand \doibase [0]{http://dx.doi.org/}%
\providecommand \selectlanguage [0]{\@gobble}%
\providecommand \bibinfo  [0]{\@secondoftwo}%
\providecommand \bibfield  [0]{\@secondoftwo}%
\providecommand \translation [1]{[#1]}%
\providecommand \BibitemOpen [0]{}%
\providecommand \bibitemStop [0]{}%
\providecommand \bibitemNoStop [0]{.\EOS\space}%
\providecommand \EOS [0]{\spacefactor3000\relax}%
\providecommand \BibitemShut  [1]{\csname bibitem#1\endcsname}%
\let\auto@bib@innerbib\@empty
\bibitem [{\citenamefont {Kodama}\ \emph {et~al.}(1987)\citenamefont {Kodama},
  \citenamefont {Y.}, \citenamefont {Hasegawa},\ and\ \citenamefont
  {A.}}]{Kodama1987}%
  \BibitemOpen
  \bibfield  {author} {\bibinfo {author} {\bibnamefont {Kodama}}, \bibinfo
  {author} {\bibnamefont {Y.}}, \bibinfo {author} {\bibnamefont {Hasegawa}}, \
  and\ \bibinfo {author} {\bibnamefont {A.}},\ }\bibfield  {title} {\enquote
  {\bibinfo {title} {Nonlinear pulse propagation in a monomode dielectric
  guide},}\ }\href@noop {} {\bibfield  {journal} {\bibinfo  {journal} {IEEE J.
  Quantum Electron.}\ } (\bibinfo {year} {1987})}\BibitemShut {NoStop}%
\bibitem [{\citenamefont {Christiansen}, \citenamefont {Sorensen},\ and\
  \citenamefont {Scott}(2000)}]{2000Nonlinear}%
  \BibitemOpen
  \bibfield  {author} {\bibinfo {author} {\bibfnamefont {P.~L.}\ \bibnamefont
  {Christiansen}}, \bibinfo {author} {\bibfnamefont {M.~P.}\ \bibnamefont
  {Sorensen}}, \ and\ \bibinfo {author} {\bibfnamefont {A.~C.}\ \bibnamefont
  {Scott}},\ }\href@noop {} {\emph {\bibinfo {title} {Nonlinear Science at the
  Dawn of the 21st Century}}}\ (\bibinfo  {publisher} {Nonlinear Science at the
  Dawn of the 21st Century},\ \bibinfo {year} {2000})\BibitemShut {NoStop}%
\bibitem [{\citenamefont {Moses}, \citenamefont {Malomed},\ and\ \citenamefont
  {Wise}(2007)}]{2007Self}%
  \BibitemOpen
  \bibfield  {author} {\bibinfo {author} {\bibfnamefont {J.}~\bibnamefont
  {Moses}}, \bibinfo {author} {\bibfnamefont {B.~A.}\ \bibnamefont {Malomed}},
  \ and\ \bibinfo {author} {\bibfnamefont {F.~W.}\ \bibnamefont {Wise}},\
  }\bibfield  {title} {\enquote {\bibinfo {title} {Self-steepening of
  ultrashort optical pulses without self-phase modulation},}\ }\href@noop {}
  {\bibfield  {journal} {\bibinfo  {journal} {Phys. Rev. A}\ }\textbf {\bibinfo
  {volume} {76}},\ \bibinfo {pages} {021802(R)} (\bibinfo {year}
  {2007})}\BibitemShut {NoStop}%
\bibitem [{\citenamefont {Mio}\ \emph {et~al.}(1976)\citenamefont {Mio},
  \citenamefont {Ogino}, \citenamefont {Minami},\ and\ \citenamefont
  {S.}}]{1976Modified}%
  \BibitemOpen
  \bibfield  {author} {\bibinfo {author} {\bibfnamefont {K.}~\bibnamefont
  {Mio}}, \bibinfo {author} {\bibfnamefont {T.}~\bibnamefont {Ogino}}, \bibinfo
  {author} {\bibfnamefont {K.}~\bibnamefont {Minami}}, \ and\ \bibinfo {author}
  {\bibfnamefont {T.}~\bibnamefont {S.}},\ }\bibfield  {title} {\enquote
  {\bibinfo {title} {Modified nonlinear {S}chr{\"o}dinger equation for
  {A}lfv{\'e}n waves propagating along the magnetic field in cold plasmas},}\
  }\href@noop {} {\bibfield  {journal} {\bibinfo  {journal} {J. Phys. Soc.
  Jpn.}\ }\textbf {\bibinfo {volume} {41}},\ \bibinfo {pages} {265--271}
  (\bibinfo {year} {1976})}\BibitemShut {NoStop}%
\bibitem [{\citenamefont {Ruderman}(2002)}]{2002DNLS}%
  \BibitemOpen
  \bibfield  {author} {\bibinfo {author} {\bibfnamefont {M.~S.}\ \bibnamefont
  {Ruderman}},\ }\bibfield  {title} {\enquote {\bibinfo {title} {Dnls equation
  for large-amplitude solitons propagating in an arbitrary direction in a
  high-beta {H}all plasma},}\ }\href@noop {} {\bibfield  {journal} {\bibinfo
  {journal} {J. Plasma Phys.}\ }\textbf {\bibinfo {volume} {67}},\ \bibinfo
  {pages} {271--276} (\bibinfo {year} {2002})}\BibitemShut {NoStop}%
\bibitem [{\citenamefont {Bludov}, \citenamefont {Konotop},\ and\ \citenamefont
  {Akhmediev}(2009)}]{2009Matter}%
  \BibitemOpen
  \bibfield  {author} {\bibinfo {author} {\bibfnamefont {Y.~V.}\ \bibnamefont
  {Bludov}}, \bibinfo {author} {\bibfnamefont {V.~V.}\ \bibnamefont {Konotop}},
  \ and\ \bibinfo {author} {\bibfnamefont {N.}~\bibnamefont {Akhmediev}},\
  }\bibfield  {title} {\enquote {\bibinfo {title} {Matter rogue waves},}\
  }\href@noop {} {\bibfield  {journal} {\bibinfo  {journal} {Phys. Rev. A}\
  }\textbf {\bibinfo {volume} {80}} (\bibinfo {year} {2009})}\BibitemShut
  {NoStop}%
\bibitem [{\citenamefont {Mjølhus}(1976)}]{Einar1976On}%
  \BibitemOpen
  \bibfield  {author} {\bibinfo {author} {\bibfnamefont {E.}~\bibnamefont
  {Mjølhus}},\ }\bibfield  {title} {\enquote {\bibinfo {title} {On the
  modulational instability of hydromagnetic waves parallel to the magnetic
  field},}\ }\href@noop {} {\bibfield  {journal} {\bibinfo  {journal} {J.
  Plasma Phys.}\ }\textbf {\bibinfo {volume} {16}},\ \bibinfo {pages}
  {321--334} (\bibinfo {year} {1976})}\BibitemShut {NoStop}%
\bibitem [{\citenamefont {Agrawal}(2000)}]{agrawal2000}%
  \BibitemOpen
  \bibfield  {author} {\bibinfo {author} {\bibfnamefont {G.~P.}\ \bibnamefont
  {Agrawal}},\ }\enquote {\bibinfo {title} {Nonlinear fiber optics},}\ in\
  \href@noop {} {\emph {\bibinfo {booktitle} {Nonlinear Science at the Dawn of
  the 21st Century}}}\ (\bibinfo  {publisher} {Springer},\ \bibinfo {year}
  {2000})\ pp.\ \bibinfo {pages} {195--211}\BibitemShut {NoStop}%
\bibitem [{\citenamefont {Hoefer}\ \emph {et~al.}(2012)\citenamefont {Hoefer},
  \citenamefont {Chang}, \citenamefont {Hamner},\ and\ \citenamefont
  {Engels}}]{2012Dark}%
  \BibitemOpen
  \bibfield  {author} {\bibinfo {author} {\bibfnamefont {M.~A.}\ \bibnamefont
  {Hoefer}}, \bibinfo {author} {\bibfnamefont {J.~J.}\ \bibnamefont {Chang}},
  \bibinfo {author} {\bibfnamefont {C.}~\bibnamefont {Hamner}}, \ and\ \bibinfo
  {author} {\bibfnamefont {P.}~\bibnamefont {Engels}},\ }\bibfield  {title}
  {\enquote {\bibinfo {title} {Dark-dark solitons and modulational instability
  in miscible two-component {B}ose-{E}instein condensates},}\ }\href@noop {}
  {\bibfield  {journal} {\bibinfo  {journal} {Phys. Rev. A}\ }\textbf {\bibinfo
  {volume} {84}},\ \bibinfo {pages} {041605} (\bibinfo {year}
  {2012})}\BibitemShut {NoStop}%
\bibitem [{\citenamefont {Hisakado}\ and\ \citenamefont
  {Wadati}(1994)}]{Hisakado1994}%
  \BibitemOpen
  \bibfield  {author} {\bibinfo {author} {\bibfnamefont {M.}~\bibnamefont
  {Hisakado}}\ and\ \bibinfo {author} {\bibfnamefont {M.}~\bibnamefont
  {Wadati}},\ }\bibfield  {title} {\enquote {\bibinfo {title} {Gauge
  transformations among generalised nonlinear {S}chr{\"o}dinger equations},}\
  }\href@noop {} {\bibfield  {journal} {\bibinfo  {journal} {J. Phys. Soc.
  Jpn.}\ }\textbf {\bibinfo {volume} {63}},\ \bibinfo {pages} {3962--3966}
  (\bibinfo {year} {1994})}\BibitemShut {NoStop}%
\bibitem [{\citenamefont {Hisakado}\ and\ \citenamefont
  {Wadati}(1995)}]{Hisakado1995}%
  \BibitemOpen
  \bibfield  {author} {\bibinfo {author} {\bibfnamefont {M.}~\bibnamefont
  {Hisakado}}\ and\ \bibinfo {author} {\bibfnamefont {M.}~\bibnamefont
  {Wadati}},\ }\bibfield  {title} {\enquote {\bibinfo {title} {Integrable
  multi-component hybrid nonlinear {S}chr{\"o}dinger equations},}\ }\href@noop
  {} {\bibfield  {journal} {\bibinfo  {journal} {J. Phys. Soc. Jpn.}\ }
  (\bibinfo {year} {1995})}\BibitemShut {NoStop}%
\bibitem [{\citenamefont {Tsuchida}\ and\ \citenamefont
  {Wadati}(1999)}]{Tsuchida1999}%
  \BibitemOpen
  \bibfield  {author} {\bibinfo {author} {\bibfnamefont {T.}~\bibnamefont
  {Tsuchida}}\ and\ \bibinfo {author} {\bibfnamefont {M.}~\bibnamefont
  {Wadati}},\ }\bibfield  {title} {\enquote {\bibinfo {title} {New integrable
  systems of derivative nonlinear {S}chr{\"o}dinger equations with multiple
  components},}\ }\href@noop {} {\bibfield  {journal} {\bibinfo  {journal}
  {Phys. Lett. A}\ } (\bibinfo {year} {1999})}\BibitemShut {NoStop}%
\bibitem [{\citenamefont {Guo}\ and\ \citenamefont {Ling}(2012)}]{2012Guo}%
  \BibitemOpen
  \bibfield  {author} {\bibinfo {author} {\bibfnamefont {B.~L.}\ \bibnamefont
  {Guo}}\ and\ \bibinfo {author} {\bibfnamefont {L.~M.}\ \bibnamefont {Ling}},\
  }\bibfield  {title} {\enquote {\bibinfo {title} {{R}iemann-{H}ilbert approach
  and {N}-soliton formula for coupled derivative {S}chr{\"o}dinger equation},}\
  }\href@noop {} {\bibfield  {journal} {\bibinfo  {journal} {J. Math. Phys.}\
  }\textbf {\bibinfo {volume} {53}},\ \bibinfo {pages} {133--3966} (\bibinfo
  {year} {2012})}\BibitemShut {NoStop}%
\bibitem [{\citenamefont {Ling}\ and\ \citenamefont {Liu}(2010)}]{Ling2010}%
  \BibitemOpen
  \bibfield  {author} {\bibinfo {author} {\bibfnamefont {L.~M.}\ \bibnamefont
  {Ling}}\ and\ \bibinfo {author} {\bibfnamefont {Q.~P.}\ \bibnamefont {Liu}},\
  }\bibfield  {title} {\enquote {\bibinfo {title} {{D}arboux transformation for
  a two-component derivative nonlinear {S}chr{\"o}dinger equation},}\
  }\href@noop {} {\bibfield  {journal} {\bibinfo  {journal} {J. Phys. A-Math.
  Theor.}\ }\textbf {\bibinfo {volume} {43}},\ \bibinfo {pages} {434023}
  (\bibinfo {year} {2010})}\BibitemShut {NoStop}%
\bibitem [{\citenamefont {Liu}\ and\ \citenamefont {Geng}(2018)}]{2018The}%
  \BibitemOpen
  \bibfield  {author} {\bibinfo {author} {\bibfnamefont {H.}~\bibnamefont
  {Liu}}\ and\ \bibinfo {author} {\bibfnamefont {X.~G.}\ \bibnamefont {Geng}},\
  }\bibfield  {title} {\enquote {\bibinfo {title} {The vector derivative
  nonlinear {S}chr{\"o}dinger equation on the half-line},}\ }\href@noop {}
  {\bibfield  {journal} {\bibinfo  {journal} {IMA J. Appl. Math.}\ }\textbf
  {\bibinfo {volume} {83}} (\bibinfo {year} {2018})}\BibitemShut {NoStop}%
\bibitem [{\citenamefont {Guo}\ \emph {et~al.}(2019)\citenamefont {Guo},
  \citenamefont {Wang}, \citenamefont {Cheng},\ and\ \citenamefont
  {He}}]{GUO2019}%
  \BibitemOpen
  \bibfield  {author} {\bibinfo {author} {\bibfnamefont {L.~J.}\ \bibnamefont
  {Guo}}, \bibinfo {author} {\bibfnamefont {L.~H.}\ \bibnamefont {Wang}},
  \bibinfo {author} {\bibfnamefont {Y.}~\bibnamefont {Cheng}}, \ and\ \bibinfo
  {author} {\bibfnamefont {J.~S.}\ \bibnamefont {He}},\ }\bibfield  {title}
  {\enquote {\bibinfo {title} {Higher-order rogue waves and modulation
  instability of the two-component derivative nonlinear {S}chr{\"o}dinger
  equation},}\ }\href {\doibase https://doi.org/10.1016/j.cnsns.2019.104915}
  {\bibfield  {journal} {\bibinfo  {journal} {Commun. Nonlinear Sci. Numer.
  Simul.}\ }\textbf {\bibinfo {volume} {79}},\ \bibinfo {pages} {104915}
  (\bibinfo {year} {2019})}\BibitemShut {NoStop}%
\bibitem [{\citenamefont {Jia}\ \emph {et~al.}(2021)\citenamefont {Jia},
  \citenamefont {Zuo}, \citenamefont {Li},\ and\ \citenamefont
  {Xiang}}]{JIA2021}%
  \BibitemOpen
  \bibfield  {author} {\bibinfo {author} {\bibfnamefont {H.~X.}\ \bibnamefont
  {Jia}}, \bibinfo {author} {\bibfnamefont {D.~W.}\ \bibnamefont {Zuo}},
  \bibinfo {author} {\bibfnamefont {X.~H.}\ \bibnamefont {Li}}, \ and\ \bibinfo
  {author} {\bibfnamefont {X.~S.}\ \bibnamefont {Xiang}},\ }\bibfield  {title}
  {\enquote {\bibinfo {title} {Breather, soliton and rogue wave of a
  two-component derivative nonlinear {S}chr{\"o}dinger equation},}\ }\href
  {\doibase https://doi.org/10.1016/j.physleta.2021.127426} {\bibfield
  {journal} {\bibinfo  {journal} {Phys. Lett. A}\ }\textbf {\bibinfo {volume}
  {405}},\ \bibinfo {pages} {127426} (\bibinfo {year} {2021})}\BibitemShut
  {NoStop}%
\bibitem [{\citenamefont {Wang}\ \emph {et~al.}(2022)\citenamefont {Wang},
  \citenamefont {Geng}, \citenamefont {Chen},\ and\ \citenamefont
  {Xue}}]{2022Riemann}%
  \BibitemOpen
  \bibfield  {author} {\bibinfo {author} {\bibfnamefont {K.~D.}\ \bibnamefont
  {Wang}}, \bibinfo {author} {\bibfnamefont {X.~G.}\ \bibnamefont {Geng}},
  \bibinfo {author} {\bibfnamefont {M.~M.}\ \bibnamefont {Chen}}, \ and\
  \bibinfo {author} {\bibfnamefont {B.}~\bibnamefont {Xue}},\ }\bibfield
  {title} {\enquote {\bibinfo {title} {Riemann–{H}ilbert approach and
  long-time asymptotics for the three-component derivative nonlinear
  {S}chr{\"o}dinger equation},}\ }\href@noop {} {\bibfield  {journal} {\bibinfo
   {journal} {Anal. Math. Phys.}\ }\textbf {\bibinfo {volume} {12}} (\bibinfo
  {year} {2022})}\BibitemShut {NoStop}%
\bibitem [{\citenamefont {Yan}(2020)}]{YAN2020}%
  \BibitemOpen
  \bibfield  {author} {\bibinfo {author} {\bibfnamefont {X.~W.}\ \bibnamefont
  {Yan}},\ }\bibfield  {title} {\enquote {\bibinfo {title} {Lax pair,
  {D}arboux-dressing transformation and localized waves of the coupled mixed
  derivative nonlinear {S}chr{\"o}dinger equation in a birefringent optical
  fiber},}\ }\href {\doibase https://doi.org/10.1016/j.aml.2020.106414}
  {\bibfield  {journal} {\bibinfo  {journal} {Appl. Math. Lett.}\ }\textbf
  {\bibinfo {volume} {107}},\ \bibinfo {pages} {106414} (\bibinfo {year}
  {2020})}\BibitemShut {NoStop}%
\bibitem [{\citenamefont {Rao}\ \emph {et~al.}(2022)\citenamefont {Rao},
  \citenamefont {Malomed}, \citenamefont {Mihalache},\ and\ \citenamefont
  {He}}]{rao2022}%
  \BibitemOpen
  \bibfield  {author} {\bibinfo {author} {\bibfnamefont {J.}~\bibnamefont
  {Rao}}, \bibinfo {author} {\bibfnamefont {B.~A.}\ \bibnamefont {Malomed}},
  \bibinfo {author} {\bibfnamefont {D.}~\bibnamefont {Mihalache}}, \ and\
  \bibinfo {author} {\bibfnamefont {J.}~\bibnamefont {He}},\ }\bibfield
  {title} {\enquote {\bibinfo {title} {General higher-order breathers and rogue
  waves in the two-component long-wave--short-wave resonance interaction
  model},}\ }\href@noop {} {\bibfield  {journal} {\bibinfo  {journal} {Stud.
  Appl. Math.}\ }\textbf {\bibinfo {volume} {149}},\ \bibinfo {pages}
  {843--878} (\bibinfo {year} {2022})}\BibitemShut {NoStop}%
\bibitem [{\citenamefont {Ma}\ and\ \citenamefont {Zhu}(2023)}]{ma2023}%
  \BibitemOpen
  \bibfield  {author} {\bibinfo {author} {\bibfnamefont {X.~X.}\ \bibnamefont
  {Ma}}\ and\ \bibinfo {author} {\bibfnamefont {J.~Y.}\ \bibnamefont {Zhu}},\
  }\bibfield  {title} {\enquote {\bibinfo {title} {{R}iemann-{H}ilbert problem
  and {N}-soliton solutions for the n-component derivative nonlinear
  {S}chr{\"o}dinger equations},}\ }\href@noop {} {\bibfield  {journal}
  {\bibinfo  {journal} {Commun. Nonlinear Sci. Numer. Simul.}\ }\textbf
  {\bibinfo {volume} {120}},\ \bibinfo {pages} {107147} (\bibinfo {year}
  {2023})}\BibitemShut {NoStop}%
\bibitem [{\citenamefont {Guo}, \citenamefont {Ling},\ and\ \citenamefont
  {Liu}(2013)}]{guo2013}%
  \BibitemOpen
  \bibfield  {author} {\bibinfo {author} {\bibfnamefont {B.~L.}\ \bibnamefont
  {Guo}}, \bibinfo {author} {\bibfnamefont {L.~M.}\ \bibnamefont {Ling}}, \
  and\ \bibinfo {author} {\bibfnamefont {Q.~P.}\ \bibnamefont {Liu}},\
  }\bibfield  {title} {\enquote {\bibinfo {title} {High-order solutions and
  generalized {D}arboux transformations of derivative nonlinear
  {S}chr{\"o}dinger equations},}\ }\href@noop {} {\bibfield  {journal}
  {\bibinfo  {journal} {Stud. Appl. Math.}\ }\textbf {\bibinfo {volume}
  {130}},\ \bibinfo {pages} {317--344} (\bibinfo {year} {2013})}\BibitemShut
  {NoStop}%
\bibitem [{\citenamefont {Zhang}, \citenamefont {Ling},\ and\ \citenamefont
  {Yan}(2021)}]{zhang2021}%
  \BibitemOpen
  \bibfield  {author} {\bibinfo {author} {\bibfnamefont {G.~Q.}\ \bibnamefont
  {Zhang}}, \bibinfo {author} {\bibfnamefont {L.~M.}\ \bibnamefont {Ling}}, \
  and\ \bibinfo {author} {\bibfnamefont {Z.~Y.}\ \bibnamefont {Yan}},\
  }\bibfield  {title} {\enquote {\bibinfo {title} {Multi-component nonlinear
  {S}chr{\"o}dinger equations with nonzero boundary conditions: higher-order
  vector peregrine solitons and asymptotic estimates},}\ }\href@noop {}
  {\bibfield  {journal} {\bibinfo  {journal} {J. Nonlinear Sci.}\ }\textbf
  {\bibinfo {volume} {31}},\ \bibinfo {pages} {81} (\bibinfo {year}
  {2021})}\BibitemShut {NoStop}%
\bibitem [{\citenamefont {Felder}, \citenamefont {Hemery},\ and\ \citenamefont
  {Veselov}(2012)}]{felder2012}%
  \BibitemOpen
  \bibfield  {author} {\bibinfo {author} {\bibfnamefont {G.}~\bibnamefont
  {Felder}}, \bibinfo {author} {\bibfnamefont {A.~D.}\ \bibnamefont {Hemery}},
  \ and\ \bibinfo {author} {\bibfnamefont {A.~P.}\ \bibnamefont {Veselov}},\
  }\bibfield  {title} {\enquote {\bibinfo {title} {Zeros of {W}ronskians of
  {H}ermite polynomials and {Y}oung diagrams},}\ }\href@noop {} {\bibfield
  {journal} {\bibinfo  {journal} {Physica D}\ }\textbf {\bibinfo {volume}
  {241}},\ \bibinfo {pages} {2131--2137} (\bibinfo {year} {2012})}\BibitemShut
  {NoStop}%
\bibitem [{\citenamefont {Bonneux}, \citenamefont {Dunning},\ and\
  \citenamefont {Stevens}(2020)}]{bonneux2020}%
  \BibitemOpen
  \bibfield  {author} {\bibinfo {author} {\bibfnamefont {N.}~\bibnamefont
  {Bonneux}}, \bibinfo {author} {\bibfnamefont {C.}~\bibnamefont {Dunning}}, \
  and\ \bibinfo {author} {\bibfnamefont {M.}~\bibnamefont {Stevens}},\
  }\bibfield  {title} {\enquote {\bibinfo {title} {Coefficients of {W}ronskian
  {H}ermite polynomials},}\ }\href@noop {} {\bibfield  {journal} {\bibinfo
  {journal} {Stud. Appl. Math.}\ }\textbf {\bibinfo {volume} {144}},\ \bibinfo
  {pages} {245--288} (\bibinfo {year} {2020})}\BibitemShut {NoStop}%
\bibitem [{\citenamefont {Yang}\ and\ \citenamefont {Yang}(2021)}]{yang2021}%
  \BibitemOpen
  \bibfield  {author} {\bibinfo {author} {\bibfnamefont {B.}~\bibnamefont
  {Yang}}\ and\ \bibinfo {author} {\bibfnamefont {J.~K.}\ \bibnamefont
  {Yang}},\ }\bibfield  {title} {\enquote {\bibinfo {title} {Rogue wave
  patterns in the nonlinear {S}chr{\"o}dinger equation},}\ }\href@noop {}
  {\bibfield  {journal} {\bibinfo  {journal} {Physica D}\ }\textbf {\bibinfo
  {volume} {419}},\ \bibinfo {pages} {132850} (\bibinfo {year}
  {2021})}\BibitemShut {NoStop}%
\bibitem [{\citenamefont {Yang}\ and\ \citenamefont {Yang}(2023)}]{yang2023}%
  \BibitemOpen
  \bibfield  {author} {\bibinfo {author} {\bibfnamefont {B.}~\bibnamefont
  {Yang}}\ and\ \bibinfo {author} {\bibfnamefont {J.~K.}\ \bibnamefont
  {Yang}},\ }\bibfield  {title} {\enquote {\bibinfo {title} {Rogue wave
  patterns associated with {O}kamoto polynomial hierarchies},}\ }\href@noop {}
  {\bibfield  {journal} {\bibinfo  {journal} {Stud. Appl. Math.}\ } (\bibinfo
  {year} {2023})}\BibitemShut {NoStop}%
\bibitem [{\citenamefont {Zhang}\ \emph {et~al.}(2023)\citenamefont {Zhang},
  \citenamefont {Huang}, \citenamefont {Feng},\ and\ \citenamefont
  {Wu}}]{zhang2022}%
  \BibitemOpen
  \bibfield  {author} {\bibinfo {author} {\bibfnamefont {G.~X.}\ \bibnamefont
  {Zhang}}, \bibinfo {author} {\bibfnamefont {P.}~\bibnamefont {Huang}},
  \bibinfo {author} {\bibfnamefont {B.~F.}\ \bibnamefont {Feng}}, \ and\
  \bibinfo {author} {\bibfnamefont {C.~F.}\ \bibnamefont {Wu}},\ }\bibfield
  {title} {\enquote {\bibinfo {title} {Rogue waves and their patterns in the
  vector nonlinear {S}chr{\"o}dinger equation},}\ }\href@noop {} {\bibfield
  {journal} {\bibinfo  {journal} {J. Nonlinear Sci.}\ }\textbf {\bibinfo
  {volume} {33}},\ \bibinfo {pages} {116} (\bibinfo {year} {2023})}\BibitemShut
  {NoStop}%
\bibitem [{\citenamefont {Lenells}(2010)}]{lenells2010}%
  \BibitemOpen
  \bibfield  {author} {\bibinfo {author} {\bibfnamefont {J.}~\bibnamefont
  {Lenells}},\ }\bibfield  {title} {\enquote {\bibinfo {title} {Dressing for a
  novel integrable generalization of the nonlinear {S}chr{\"o}dinger
  equation},}\ }\href@noop {} {\bibfield  {journal} {\bibinfo  {journal} {J.
  Nonlinear Sci.}\ }\textbf {\bibinfo {volume} {20}},\ \bibinfo {pages}
  {709--722} (\bibinfo {year} {2010})}\BibitemShut {NoStop}%
\end{thebibliography}%

\end{document}